\pdfoutput=1 %added to ensure arXiv to use PDFlatex
\documentclass[12pt,british,english]{book}
\usepackage[T1]{fontenc}
\usepackage[latin9]{inputenc}
\usepackage{geometry}
\usepackage{varioref}
\usepackage{prettyref}
\usepackage{textcomp}
\usepackage{amsthm}
\usepackage{amsmath}
\usepackage{amssymb}
\usepackage{graphicx}
\usepackage{setspace}
\usepackage{esint}
\PassOptionsToPackage{normalem}{ulem}
\usepackage{ulem}

\makeatletter

\newcommand*\LyXbar{\rule[0.585ex]{1.2em}{0.25pt}}
\DeclareRobustCommand{\greektext}{%
  \fontencoding{LGR}\selectfont\def\encodingdefault{LGR}}
\DeclareRobustCommand{\textgreek}[1]{\leavevmode{\greektext #1}}
\DeclareFontEncoding{LGR}{}{}
\DeclareTextSymbol{\~}{LGR}{126}
\newcommand{\lyxmathsym}[1]{\ifmmode\begingroup\def\b@ld{bold}
  \text{\ifx\math@version\b@ld\bfseries\fi#1}\endgroup\else#1\fi}

\providecommand{\tabularnewline}{\\}

  \theoremstyle{remark}
  \newtheorem*{rem*}{\protect\remarkname}
\theoremstyle{plain}
\newtheorem{thm}{\protect\theoremname}
  \theoremstyle{definition}
  \newtheorem{example}[thm]{\protect\examplename}
  \theoremstyle{definition}
  \newtheorem{defn}[thm]{\protect\definitionname}
  \theoremstyle{remark}
  \newtheorem{rem}[thm]{\protect\remarkname}
  \theoremstyle{definition}
  \newtheorem*{example*}{\protect\examplename}
 \theoremstyle{definition}
 \newtheorem*{defn*}{\protect\definitionname}
  \theoremstyle{remark}
  \newtheorem{claim}[thm]{\protect\claimname}

\usepackage[all,knot]{xy} 
\xyoption{knot}
\xyoption{frame}
\xyoption{arc}
\xyoption{curve}
\xyoption{poly}
\usepackage{pdfsync}
\makeatother

\usepackage{babel}
  \addto\captionsbritish{\renewcommand{\claimname}{Claim}}
  \addto\captionsbritish{\renewcommand{\definitionname}{Definition}}
  \addto\captionsbritish{\renewcommand{\examplename}{Example}}
  \addto\captionsbritish{\renewcommand{\remarkname}{Remark}}
  \addto\captionsbritish{\renewcommand{\theoremname}{Theorem}}
  \addto\captionsenglish{\renewcommand{\claimname}{Claim}}
  \addto\captionsenglish{\renewcommand{\definitionname}{Definition}}
  \addto\captionsenglish{\renewcommand{\examplename}{Example}}
  \addto\captionsenglish{\renewcommand{\remarkname}{Remark}}
  \addto\captionsenglish{\renewcommand{\theoremname}{Theorem}}
  \providecommand{\claimname}{Claim}
  \providecommand{\definitionname}{Definition}
  \providecommand{\examplename}{Example}
  \providecommand{\remarkname}{Remark}
\providecommand{\theoremname}{Theorem}

\begin{document}

\title{Introduction to Spin Networks and Towards a Generalization of the
Decomposition Theorem}

\author{Hans-Christian Ruiz}

\maketitle

\section*{Acknowledgements}

I am truly indebted and thankful for the help and advice of Prof.
John Barrett; without him this work could never have existed. I would
also like to recognize the valuable help and support of the staff
at the University of Nottingham and the University of Munich (LMU).
I am sincerely and heartily grateful to both my student advisor Bernhard
Emmer and the Erasmus Institutional Coordinator of the LMU Jean Schleiss.
I am sure this work would not have been possible without their support.
Thanks are also due to Dr. Hal Haggard for useful comments.

\tableofcontents{}

\section{Introduction}

The objective of this dissertation is twofold. On one hand, it is
intended as a short introduction to spin networks and invariants of
3-manifolds. It covers the main areas needed to have a first understanding
of the topics involved in the development of spin networks. The topics
are describe in a detailed but not exhaustive manner and in order
of their conceptual development such that the reader is able to use
this work as a first reading. On the other hand, some results aiming
towards a decomposition theorem for non-planar spin networks are presented
in \prettyref{cha:Non-planar-Spin-Networks}. 

We start in \prettyref{cha:From-SN-to-GR} with the first conceptual
development of spin networks by Penrose, \cite{penrose1971angular,penrose1971applications},
and a very brief presentation of the main concepts in the theory,
\cite{major1999spin}, which will then be explained in more detail
throughout the dissertation. \prettyref{sec:Spin-Networks} gives
the motivation for considering spin networks as a way of constructing
a 3-dimensional Euclidean space, as well as its formal framework of
Abstract Tensor Systems, which is a generalization of the concept
of tensor algebra. In fact, the diagrammatical language of spin networks
are a representation of such systems. The arguments given are then
reinforced when the description of General Relativity without coordinates
due to Regge is presented, \cite{regge1961general}, and the Ponzano-Regge
theory connecting both concepts is described, \cite{ponzano1969semiclassical}.
In \prettyref{sec:General-Relativity} we will see that the information
about the curvature of an $n$-manifold is encoded in the $(n-2)$-skeleton
of its triangulation and in \prettyref{sec:Connection-between-GR-SN}
we will discuss the relation between the asymptotic formula for the
$6j$-symbols and the path integral over the exponential of the Einstein-Hilbert
action. This relation represents the possibility of a similar description
to QFT of a ``Feynman integral'' over geometries of a combinatorial
manifold.

In \prettyref{cha:Mathematical-Framework} we present the basic mathematical
framework for the algebraic description of spin networks via quantum
groups, \prettyref{sec:Hopf-alg-and-QGrps} following \cite{majid2000foundations}.
These algebraic objects belong to the family of quasitriangular Hopf
algebras and the one deserving our special attention is the $q$-deformed
universal enveloping Hopf algebra $U_{q}(\mathfrak{sl}_{2})$ of the
Lie algebra $\mathfrak{sl}_{2}$ which gives the data needed to regularize
the Ponzano-Regge theory encountered before. The fact that the set
of representations of this quantum group is finite, allows the construction
of a well-defined invariant of 3-manifolds. We will discuss the corresponding
state sums given by Turaev and Viro in \prettyref{cha:Invariants-of-3-Manifolds}.
In \prettyref{sub:Category-Theory} we introduce category theory,
\cite{MacLane1971}, in order to understand spherical categories,
\cite{barrett1999spherical,barrett1996invariants}, and their relation
to the diagrammatic representations of spin networks given by the
Temperley-Lieb recoupling theory in \prettyref{sec:Some-More-Diagrammatics}
following \cite{kauffman1994temperley,kauffman2001knots}. The main
concept needed from spherical categories is theorem \vref{thm:non-degenerate-quotient},
which states that, given an additive spherical category, it is always
possible to build a quotient in order to construct a non-degenerate
spherical category. This construction allows the condition of semisimplicity
needed for the construction of the above mentioned invariants. Hence
the concept of semisimple spherical categories given by Barrett generalizes
the Turaev-Viro invariants given by the state sum of $U_{q}(\mathfrak{sl}_{2})$,
\cite{turaev1992state}. We give in \prettyref{sec:Invariants-from-Spherical}
a brief account of this generalization. 

On the other hand, the diagrammatic language of spin networks is contained
in the very broad field of knot theory. We will give a very brief
account of knot theory and refer the reader to \cite{kauffman2001knots,jones2005jones,sawin1996links}
for a more detailed description. The most important concept of this
area for us is the Temperley-Lieb recoupling theory, which allows
a definition of the $n$-edges in a spin network as a weighted sum
over representations of the generators of the Artin braid group $B_{n}$.
Furthermore, this theory is important as a tool for the evaluation
of closed spin networks in terms of the $q$-deformed $6j$-symbols,
especially in the case where $q$ is a root of unity, since it delivers
all relations related to Moussouris' algorithm and the correspondence
of a tetrahedron with the recoupling coefficients, relevant to the
Ponzano-Regge partition function.

In \prettyref{sec:TQFT} we present briefly the axioms of a topological
quantum field theory (TQFT) related to the concepts mentioned above,
\cite{atiyah1988topological,atiyah1990geometry,turaev1992state}.
The main idea needed in the description of the invariants of 3-manifolds,
in fact of spin networks as well, is the association of finitely generated
modules over some ring to (oriented) closed smooth manifolds of a
fixed dimension $d$. In addition, to any $(d+1)$-cobordism between
two such manifolds one associates an element of the module associated
to its boundary, which is usually a tensor product of modules associated
to each component of the boundary. The presentation of TQFT here is
intended merely as a support to describe the Turaev-Viro invariants
and, regrettably, it was not possible to dedicate more space to this
topic, which could be a (very interesting) dissertation on itself.
For a detailed discussion of the relation between link and knot theory
and TQFT we refer to \cite{sawin1996links}. \prettyref{sec:TQFT}
is directly connected to \prettyref{sec:State-Sum-Invariants}, where
a TQFT arises naturally from the Turaev-Viro state sum after the construction
of a suitable quotient defining a functor from the category of cobordisms
of triangulated 2-manifolds to the category of modules over a ring
introduced for the initial data, \cite{turaev1992state}. The data
for the construction of the state sum is introduced in a general setting.
After discussing the topological aspects of the theory in \prettyref{sec:Moves-on-Triangulations},
such as the transformations on triangulations called Alexander moves,
\cite{viro1992moves}, and their dual form called the Matveev-Piergallini
moves, we conclude with the correspondence of the Turaev-Viro and
the Kauffman-Lins invariants. We give this correspondence in an informal
manner based on the comparison of the results in \cite{turaev1992state}
and \cite{kauffman1994temperley}. For a formal discussion the reader
is referred to \cite{piunikhin1992turaev}. 

In \prettyref{cha:Non-planar-Spin-Networks} we carry out the original
research towards a decomposition theorem for non-planar spin networks.
We start by presenting some basic concepts of (topological) graph
theory, mainly Kuratowski's Theorem \vref{thm:Kuratowski's-Theorem},
in order to prepare the setting in which we will work, \cite{beineke1997graph,hartsfield2003pearls}.
The main idea in \prettyref{sec:Kuratowski's-Theorem-Embeddings}
is to identify the graphs corresponding to non-planar spin networks
and construct the surfaces in which they are \textit{cellular} embeddable
by applying the so called rotation rule. It turns out that the only
graph that needs to be studied is the $(3,3)$-bipartite graph $K_{3,3}$
since the other subgraph responsible for the non-planarity of a given
spin network is the complete graph on 5 vertices, denoted $K_{5}$.
This graph is expandable to the Petersen graph which has as one of
its minors $K_{3,3}$. We used exclusively cellular embeddings to
assure that the information about the topology of the surface is encoded
in the graph itself and found readily that, for each graph, there
is a topological constraint reducing strongly its possible cellular
embeddings. Some examples for $K_{3,3}$ were calculated from the
Rotation Scheme Theorem \vref{thm:Rotation-Scheme-Theorem} and it
was proven that these are the only possible ones up to permutation
of the vertices. In \prettyref{sec:non-planar-SN} we present Moussouris'
algorithm for the evaluation of planar spin networks, \cite{moussouris1983quantum},
and extend it to account for the phase factors containing the information
of a toroidal spin network. In addition, the Decomposition Theorem
\vref{thm:Decomposition-Theorem} is improved to account for this
simple non-planar case and its irreducible network, named toroidal
phase factor, is given. An attempt to evaluate the $K_{3,3}$ spin
network is made by associating the irreducible toroidal network to
the twisting factor given in \cite{kauffman1994temperley,carter1995classical}.
This factor changes effectively the orientation of a vertex and was
used (may be in a naive way) to turn a toroidal phase factor into
a theta-net translating the diagrammatic information of the topology
into an algebraic factor, leaving a possible evaluation in terms of
$q$-$6j$-symbols and the twist factor; this evaluation was called
toroidal symbol. The results of some calculations regarding the different
embeddings are given and compared. This shows some ambiguity in the
evaluations calculated, which are discussed and a solution is proposed.
Finally, the further work needed to achieve a general theorem concerning
the evaluation of non-planar spin networks of higher genus is described
briefly. We hope that the simple categorization of orientable surfaces
as a connected sum of tori is reflected in the evaluation of spin
networks in this general case, namely, as a sum of products of toroidal
symbols. The further analysis regarding this generalization and the
possible solution to the ambiguity of the evaluations calculated will
be presented soon in an article containing also the results presented
in this dissertation.

For the sake of readability the author was careful to avoid as much
as possible mathematical details, keeping in mind the necessity of
further explanation in key aspects of the dissertation. In many cases,
however, this was difficult to achieve without extending too much
the scope of the dissertation. Hence, for a detailed description and
further aspects of the different topics, as well as the proofs of
the statements made, the reader is referred to the literature given.

Most diagrams were made with the ``Xy-Pic'' package and the help
of Aaron Lauda's tutorial. For this and more documentation concerning
this useful package we refer to the Xy-pic website ( http://www.tug.org/applications/Xy-pic/
).

\chapter{From Spin Networks To General Relativity\label{cha:From-SN-to-GR}}

\section{Spin Networks\label{sec:Spin-Networks}}

\subsection{The Origins of Spin Networks}

Penrose constructed a discrete model of space based on the concept
of quantum mechanical angular momentum with the goal to build a consistent
model from which classical, continuous geometry emerged in a limit.
It was shown that spin networks could reproduced a 3 dimensional Euclidean
space, this result is known as the \emph{spin-geometry theorem} (cf.
\prettyref{thm:spin-geom-thm}). The basic idea of Penrose%
\footnote{The following discussion can be found in \cite{penrose1971angular}.%
} was to build up space-time and quantum mechanics simultaneously from
combinatorial principles by using as primary concepts the rules for
combining angular momenta together. Continuous concepts, such as directions
in space, should then emerge in a limit where the systems of angular
momenta get more complicated. 

In order to avoid referring to the idea of spins having $2s+1$ available
states as preexisting directions of a background space, one has to
work only with the total angular momentum ($j$-value) rather than
the direction of quantization ($m$-value)%
\footnote{It is important to notice that a direction of quantization only appears
when the system is related to a bigger one, e.g. a magnetic field,
which defines a \textquotedbl{}preferred\textquotedbl{} direction.%
}. Now imagine an object, called spin network, like the following

\begin{center}
\includegraphics[scale=0.5]{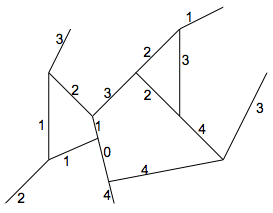}
\par\end{center}

where the number next to each line represents twice the angular momentum
of an isolated and stationary subsystem%
\footnote{The system has to be regarded as isolated and stationary so it has
a well-defined angular momentum, $j\hbar=\frac{1}{2}n\hbar,\:(n=0,1,2,\dots).$%
}, represented by the line and called an n-unit or n-edge. The above
diagram is only a representation of the, for now, rather abstract
concept of spin network, thus it has no spatial meaning, only a relational
one, i.e. its defining properties are the relations between the edges%
\footnote{From this point of view, we can consider a spin network as a (cubic)
graph.%
}. Notice that an important element of this diagram is the trivalent
vertex%
\footnote{The number of edges at a vertex is not limited to three, e.g. in quantum
gravity spin networks are generalized to include higher valence vertices,
as pointed out in \cite{major1999spin}.%
}

\begin{center}
\includegraphics[scale=0.5]{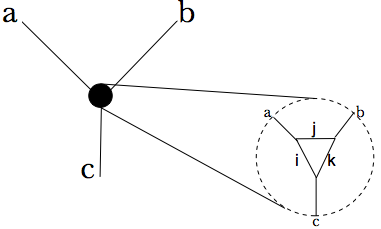}
\par\end{center}

where the dashed circle indicates spin network structure at the vertex
with \emph{internal labels} $i,j,k$ being positive integers determined
by the external labels $a,b,c$ 
\[
i=(a+c-b)/2,\: j=(b+c-a)/2,\: k=(a+b-c)/2.
\]

The external labels must satisfy the triangle inequalities and add
up to an even integer, these conditions are necessary as an expression
of conservation of angular momenta.
\begin{rem*}
The subsystems are not moving relative to one another, there is only
a transfer of angular momentum allowed and the regrouping into different
subsystems, not even time-ordering of the events play a role, only
the topological aspects of the system are relevant.
\end{rem*}
Since only the relational properties of the edges define the spin
network some combinatorial rules must be given, but how are this combinatorial
rules to be interpreted? Every diagram will be assigned a non-negative
integer called its norm, which can be calculated from any given spin
network in a purely combinatorial way. We can think, carefully, of
this as the measure of the frequency of occurrence of the given spin
network, so we could use this norms in some cases to calculate the
probabilities of different spin values occurring. Since the norm is
always an integer these probabilities will always be rational numbers.
How can these rules be obtained? A natural choice would be to derived
them from irreducible representations of $SO(3)$.

How does this enable us to build up a concept of space? In other words,
how can anybody say anything about directions in space, if there is
only the non-directional concept of total angular momentum? One could
ask for the \emph{``orientation'' of a n-unit in relation to some
larger structure} belonging to the system under consideration. Hence,
the system should involve a fairly large total angular momentum number
$N$ (called a large unit), in order to have the possibility of a
well-defined direction as the spin axis of the system%
\footnote{Recall that, because of the correspondence principle, systems with
large quantum numbers behave nearly as classical ones. %
}. After defining this ``direction'', one can then ask further how
the concept of angles between these ``directions'' could be defined,
and then prove if we get a consistent interpretation of this in terms
of directions in a 3-dimensional Euclidean space.

To define an angle between two large units Penrose considered following
``experiment''. Suppose a 1-unit is detached from a large $N$-unit
such that it leaves an $(N-1)$-unit behind and it is then re-attached
to some other large $M$-unit. According to the rules allowed, which
are described below, we have then either an $(M-1)$-unit or an $(M+1)$-unit,
thus there will be certain probabilities for these two different outcomes.
With the information given by these probability values we can obtain
the angle between the $N$-unit and the $M$-unit. To see how this
is possible, notice that if the units are ``parallel'', then we
would expect zero probability for the $M-1$ value and if the units
are anti-parallel we would expect zero probability for the $M+1$
value. For a ``perpendicular'' position, we expect then equal probability
values for each of the two outcomes. Hence, for an angle $\theta$
between the directions of the two large units we would expect a probability
$P\left(M\pm1\right)=\frac{1}{2}\pm\frac{1}{2}\cos\theta$ for the
$M$-unit to be increased/reduced one unit. In this way, from knowledge
of a spin network, one can calculate by means of a combinatorial procedure
the probability of each of the two possible outcomes and with them
define an angle between these two large systems. Since these type
of probabilities will always be rational numbers, one could only obtain
angles with rational cosines, but these ``angles'' would normally
not agree with the actual angles of the Euclidean space until one
goes to the limit of large systems. This means that rational probabilities,
e.g. $p=m/n$, can be seen as something more fundamental than ordinary
real number probabilities. The former might be regarded as arising
because nature has to make a choice between $m$ alternative possibilities
of one kind and $n$ alternative ones of another, all of which are
to be \textit{equally probable.} Only in the limit then, when numbers
get to infinity, we would get the full continuous spectrum of probability
values.

Now, consider a number of disconnected systems, each of them producing
a large N-unit. To measure the ``angle'' between two of them use
the above experiment. Then, the probability (cf. \prettyref{eq:probability-SN})
of the second N-unit to become an $\left(N\pm1\right)$-unit is $\frac{1}{2}\left(N+1\pm1\right)/\left(N+1\right)$
. Hence, the probabilities become equal in the limit $N\rightarrow\infty$
and so, for large $N$, we could assign a right-angle between these
two units. With this result one could put any number of $N$-units
at right-angles to each other, thus, there is no restriction of the
dimensionality of our space. 

Clearly something went wrong. Remember there were no connections between
any of the N-units so there is absolutely \textit{no information }concerning
the interconnections between the different N-units (there is a lack
of knowledge concerning the origins of the systems), so we can think
of the probability, in this particular case, as arising \textit{entirely
}out of the ignorance of the system, rather than being genuine quantum
mechanical probabilities due to an ``angle''.

Suppose there is some ``known'' connecting network $\kappa$ with
two large units coming out, and we realize the above experiment two
times consecutively.

\begin{center}
\includegraphics[scale=0.5]{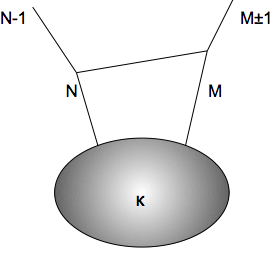}
\par\end{center}

If this is an ``ignorance'' situation, i.e. we do not know about
how the spin axes are pointing, then the probabilities in the second
experiment will change by the information of the relative orientation
of the spin axes obtained in the first one. Hence, if the probabilities
calculated for the second experiment are substantially altered, then
there is a large ``ignorance'' factor involved. On the other hand,
if they are not substantially altered, then the angle between the
two large units is well-defined. 

Suppose we have now a system which has many large units emerging $(A,B,...)$
and that the angle between any two of them is well-defined. Then,
without proof%
\footnote{For proof see \cite{moussouris1983quantum}.%
},
\begin{thm}
\textbf{\textup{Spin-geometry Theorem:}}\label{thm:spin-geom-thm}
In the situation above, all angles between the large units are consistent
with angles between the directions in a three-dimensional Euclidean
space.
\end{thm}
Does this space corresponds to the original, given three-dimensional
Euclidean space in the sense of a background? Well, there can be systems
($n$-units) with a large total angular momentum, but which do not
give well-defined directions in the original space, remember, there
are states with a large angular momentum which point all over the
place (e.g. $m=0$ states). Thus, the angles coming form these units
do not correspond to anything one can see as angles in the original
space, but they are nevertheless consistent with the angles between
directions in some abstract Euclidean 3-space. One could take the
view that the Euclidean three-dimensional space that comes out of
the large units of spin networks is the \emph{real} space, and that
the original space is just a convenience, like coordinates in general
relativity. The main idea here is the claim that \textit{the system
defines the geometry. }

In order to define the probabilities of a given process, the definition
of the norm of a spin network is needed. This will be given in terms
of a concept called the \textbf{value of a closed spin network}. A
spin network is closed if it has no free ends. If it is not closed
a value will not be assigned, but one can always obtain a closed network
by making a copy of the \textquotedbl{}open\textquotedbl{} one and
gluing the corresponding edges together. To define the value of a
closed spin network, replace each n-unit in the diagram by n parallel
strands. At each vertex, the strands belonging to different edges
must be connected together in pairs, such that two strands of the
same unit are not allowed to be connected (cf. Remark \prettyref{rem:properties-of-antisymmetrizer}).
Such a connection is called a \textit{vertex connection.} The \emph{sign
of a vertex connection} $v$ is defined as\emph{ }$s_{v}=\left(-1\right)^{x}$
where $\mathit{x}$ is the number of intersection points between different
strands at the vertex%
\footnote{This can also be seen as the signum of the permutation of strands
involved in their pairing at each vertex connection, cf. Definition
\prettyref{def:The-weaving}.%
}. When the vertex connections have been completed at every vertex
of a closed spin network, then we get a number $c$ of closed loops.
The value of the closed spin network is the sum over every possible
way of completing the vertex connections. From these considerations,
Penrose got the following expression

\begin{equation}
value=\frac{\sum s\left(-2\right)^{c}}{\prod n!},\label{eq:Vaule-of-SN}
\end{equation}
 where $s=\prod_{v}s_{v}$\textit{ }and $\prod n!$ ranges over all
the $n$-units of the spin network. 
\begin{rem*}
\textit{(i)} In the calculation of the value only the intersections
at the vertex connections count for the sign of the diagram, see example
below. \textit{(ii)} Note that each closed loop yield a value of $-2$
to the main diagram. \textit{(iii)} The value of any closed spin network
always turns out to be an integer.\end{rem*}
\begin{example}
Consider the following diagram and its possible vertex connections

\includegraphics[scale=0.4]{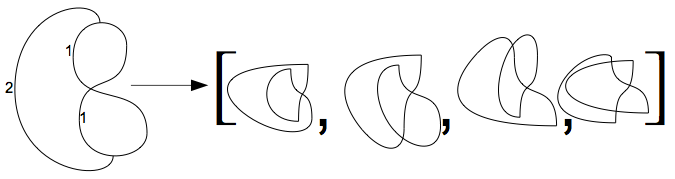}

Note that the intersection arising from the crossing of the two 1-units
does not contribute to the sign of the terms in the sum. The evaluation
of the above diagram according to \prettyref{eq:Vaule-of-SN} gives
\[
value=\frac{1}{2!1!1!}\{+(-2)-(-2)^{2}-(-2)^{2}+(-2)\}=-6.
\]

\end{example}
The value of a spin network is multiplicative, i.e. the value of the
union of disjoint spin networks is equal to the product of their individual
values. Before coming to the definition of the norm we remark the
existence of some useful reduction formula for evaluating spin networks
which may be substituted into any closed spin network in order to
obtain a valid relation between values; they can be found in \cite{penrose1971angular}.
\begin{defn}
The \textbf{norm} of a given spin network is obtained by joining together
the corresponding free end units of two copies of the diagram to make
a closed network and taking the modulus of its value. In other words,
the norm of a spin network is the \textbf{modulus of the value} of
its corresponding closed spin network.
\end{defn}
Finally, the norm is used to calculate the probabilities for spin-numbers
in the sense of the above described ``experiment''. Given a spin
network $\alpha$ with a free $a-$ and $b-$unit, suppose these two
units come together to form an x-unit in a resulting spin network
$\beta$. What are the various probabilities for the different possible
values of x? Let $\gamma$ denote the vertex formed by a, b and x
and $\xi$ denote the spin-network consisting of the x-unit alone.
The probability for the resulting spin-number to be x is

\begin{center}
\begin{equation}
probability(x)=\frac{norm\beta}{norm\alpha}\frac{norm\xi(x)}{norm\gamma(x)}.\label{eq:probability-SN}
\end{equation}

\par\end{center}

With the interpretation of an ``angle'' given by Penrose the three-dimensional
Euclidean nature of the ``directions in space'' is a consequence
of the combinatorial probabilities of spin networks. One should keep
in mind that this space is the one defined by the system and there
\emph{is} a distinction between this one and the space introduced
as a background in a conventional formalism.

\subsection{A Formal Framework for Spin Networks: Abstract Tensor Systems}

In \cite{penrose1971applications} Penrose describes his theory of
\emph{abstract tensor systems} (ATS) for more general objects than
ordinary tensors and which are related to diagrams like the ones above.
These objects are denoted formally identically as in the tensor index
notation, but the meaning of indices is now different. Indices are
now just a label, they do not stand for $1,2,3,...,n$, and an element
$\xi^{a}$ of an ATS is not a set of components of a vector but rather
a whole element of a module $\mathcal{T}^{a}$ over a ring $\mathcal{T}$.
This means that $\xi^{a}\neq\xi^{b}$ if $a\neq b$ but $\mathcal{T}^{a}\cong\mathcal{T}^{b}$.
This is naturally extended to other objects like $\chi_{f\dots h}^{ab\dots d}\in\mathcal{T}_{f\dots h}^{ab\dots d}$
where $a,b,...,d,f,...,h$ are all distinct and $\mathcal{T}_{f\dots h}^{ab\dots d}$
is a module over $\mathcal{T}$ %
\footnote{The order of the indices of the objects is important, $\chi_{f\dots h}^{ab\dots d}\neq\chi_{f\dots h}^{db\dots a}$,
but not for the modules $\mathcal{T}_{f\dots h}^{ab\dots d}=\mathcal{T}_{f\dots h}^{db\dots a}$. %
}. The ATS has four basic operations which are
\begin{enumerate}
\item Addition: $\mathcal{T}_{u\dots w}^{x\dots z}\times\mathcal{T}_{u\dots w}^{x\dots z}\rightarrow\mathcal{T}_{u\dots w}^{x\dots z}$
\item Outer multiplication: $\mathcal{T}_{p\dots r}^{a\dots d}\times\mathcal{T}_{u\dots w}^{x\dots z}\rightarrow\mathcal{T}_{p\dots ru\dots w}^{a\dots dx\dots z}$
\item Contraction%
\footnote{We will use the dummy index notation and the Einstein's sum convention.%
}: $\mathcal{T}_{qu\dots w}^{px\dots z}\rightarrow\mathcal{T}_{u\dots w}^{x\dots z}$
\item Index substitution: $\mathcal{T}_{u\dots w}^{x\dots z}\rightarrow\mathcal{T}_{k\dots m}^{f\dots h}$
where there is a one-to-one correspondence between both set of indices.
\end{enumerate}
The axioms for these operations are the following
\begin{enumerate}
\item The addition defines an Abelian group structure for each module in
the ATS.
\item The multiplication is distributive over the addition.
\item The contraction is distributive over addition and commutes with multiplication
and other contractions.
\item The index substitution is caused by any permutation of the set of
indices but the validity of any formula stays unaltered.
\end{enumerate}
These abstract tensor systems may be expressed diagrammatically, what
allows us to see connections between indices at a glance. Penrose
pointed out, that if we regard the labels as points on a plane, we
may denote an object of the ATS by a symbol with ``arms'' corresponding
the upper indices and ``legs'' corresponding the lower ones.
\begin{example}
\begin{onehalfspace}
\label{exa:algebra-with-diagrams}If $\theta_{c}^{ab}\in\mathcal{T}_{c}^{ab}$
and $\chi_{bcd}^{a}\in\mathcal{T}_{bcd}^{a}$ we denote these object
by

\includegraphics[scale=0.6]{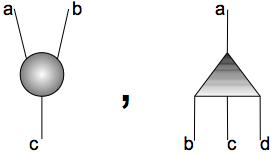}

respectively. Now, outer products are expressed as a juxtaposition
of individual symbols and contractions are depicted by joining the
corresponding arm and leg:

\includegraphics[scale=0.6]{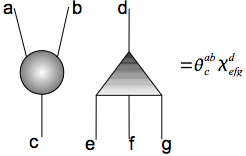},
\includegraphics[scale=0.5]{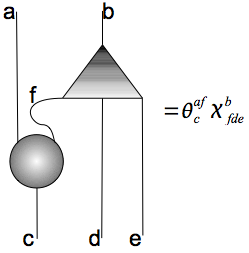}
\end{onehalfspace}

The addition of two objects is analogous. One simply draws the diagrams
for each object and puts a ``+'' sing in-between.
\end{example}
We require the ATS to have ``unit elements'' $1\in\mathcal{T}$
such that $1\chi_{\dots}^{\dots}=\chi_{\dots}^{\dots}$ for all elements
of $\mathcal{T}_{\dots}^{\dots}$ and $\delta_{b}^{a}\in\mathcal{T}_{b}^{a}$
such that $\chi_{\dots}^{..p..}\delta_{p}^{q}=\chi_{\dots}^{..q..}$
and $\chi_{..x..}^{\dots}\delta_{y}^{x}=\chi_{..y..}^{\dots}$. These
elements are unique and the element $\delta_{b}^{a}$ has the formal
properties of a Kronecker delta, so the definition of the ``dimension''
of $\mathcal{T}^{a}$ arises naturally to be the scalar $\delta_{a}^{a}=\nu$. 
\begin{rem*}
In ordinary tensor systems $\nu$ is a positive integer, while in
the more general case presented here it could also be a negative integer,
see remark \prettyref{rem:properties-of-antisymmetrizer}. 

Diagrammatically, the ``dimension'' is depicted as a closed loop
as the following

\includegraphics[scale=0.5]{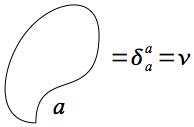}
\end{rem*}
This is the general framework for the discussion about spin networks
in the beginning of this section. Now, we can re-introduce the spin
networks by associating $2\times2$ matrices with diagrams, as in
\cite{major1999spin}. However, we have to make certain that the diagrammatics
are planar isotopic, i.e. invariant under smooth deformations in the
plane, cf. \prettyref{sec:Some-More-Diagrammatics}. From the considerations
above, we start by associating%
\footnote{This association implicitly fixes a preferred direction from the bottom
to the top of the page, i.e. $\bigl|_{A}^{B}\nleftrightarrow${\large $_{A}$\LyXbar{}$_{B}$.}%
} the Kronecker delta symbol $\delta_{A}^{B}$, which in this case
is the $2\times2$ identity matrix, to a line:

\begin{center}
$\delta_{A}^{B}$$\rightarrow$$\Bigl|_{A}^{B}$
\par\end{center}

Another possible identification%
\footnote{Again, the association implies a direction from the first index to
the second which has to be kept in mind when evaluating closed networks
and notice that the orientation has to be consistent throughout the
diagram.%
} is between a curve and the antisymmetric matrix $\varepsilon_{AB}=\varepsilon^{AB}=\biggl(\begin{array}{cc}
0 & 1\\
-1 & 0
\end{array}\biggl)$; i. e. $\varepsilon_{AB}\rightarrow_{A}\bigcap_{B}$ and $\varepsilon^{AB}\rightarrow^{A}\bigcup^{B}$.

These definitions are, however, not planar isotopic since the identities%
\footnote{This identities are associated to the more general notion of a pivotal
category, cf. Def. \vref{def:pivotal-category}.%
} $\delta_{A}^{C}\varepsilon_{CD}\varepsilon^{DE}\delta_{E}^{B}=\varepsilon_{AD}\varepsilon^{DB}=-\delta_{A}^{B}$
and $\varepsilon_{AD}\varepsilon_{BC}\varepsilon^{CD}=-\varepsilon_{AB}$
correspond to the following diagrams

\begin{center}
\includegraphics[scale=0.5]{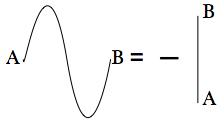}, \includegraphics[scale=0.5]{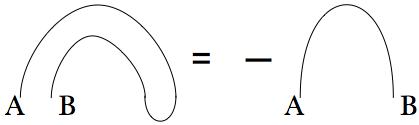}.
\par\end{center}

As pointed out in \cite{major1999spin}, in order to solve this problem
we simply re-define the bent line according to $\varepsilon_{AB}\rightarrow\tilde{\varepsilon}_{AB}=i\varepsilon_{AB}$.

Now, consider the relation $\delta_{A}^{D}\delta_{B}^{C}\tilde{\varepsilon}_{CD}=-\tilde{\varepsilon}_{AB}$,
which gives the diagram

\begin{center}
\includegraphics[scale=0.4]{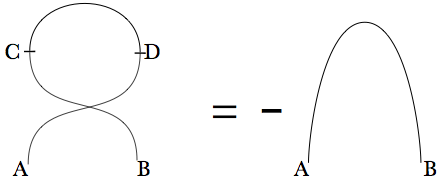}
\par\end{center}

Such difficulty can be solved by associating a minus sign to each
crossing. All these modifications ensure the planar isotopy and as
explained above gives us the possibility to perform algebraic calculations
in a more transparent way (cf. Example \prettyref{exa:algebra-with-diagrams}).
For instance, the value of a simple closed loop is negative, since
$\tilde{\varepsilon}_{AB}\tilde{\varepsilon}^{BA}=-\varepsilon_{AB}\varepsilon^{BA}=-2$.

There is an important relation called the skein relation (or spinor
identity) which is derived from a known identity concerning the product
of two antisymmetric tensors%
\footnote{Note that for the diagrammatic relation the sign of the crossing changes.%
}:

\begin{center}
\includegraphics[scale=0.4]{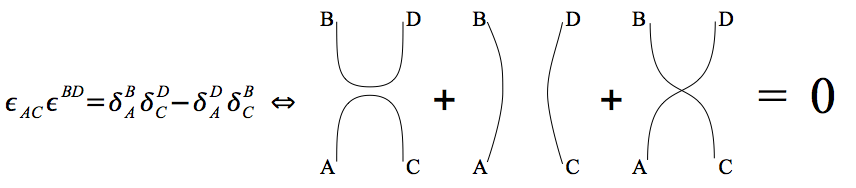}
\par\end{center}

This algebra is ``topological'', i.e. any two diagrams which are
homotopic to each other represent two equivalent algebraic expressions.
The spin network diagrammatics are topologically invariant in a plane
as a consequence of a result of Reidemeister: a knot, i.e. an embedding
of $S^{1}$ in a three dimensional Euclidean space, is homotopic to
another knot, if and only if, the planar projection of the knots can
be transformed into each other via a finite sequence of Reidemeister
moves%
\footnote{In three dimensions one has different types of crossing, the ``over
crossing'' and the ``under-cross''. Here we just depicted the projections
on the plane so, in the general case, on each intersection one has
to keep in mind the additional structure of the crossing, cf. \prettyref{sec:Some-More-Diagrammatics}. %
}:
\begin{itemize}
\item \textbf{Move 0: }In the plane of projection, one can deform the following
curve smoothly \\
\includegraphics[scale=0.5]{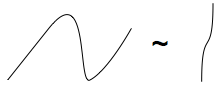}
\item \textbf{Move I:} A curl may be undone since we are dealing with one
dimensional objects \\
\includegraphics[scale=0.5]{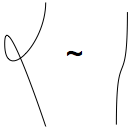}
\item \textbf{Move II:} Overlaps on the projection plane of distinct curves
are not knotted\\
 \includegraphics[scale=0.5]{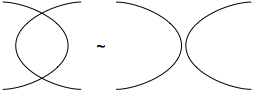}
\item \textbf{Move} \textbf{III:} Planar deformations under or over a diagram
\\
\includegraphics[scale=0.6]{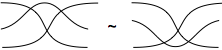}
\end{itemize}
With a finite number of these moves the projection of a knot may be
transformed into the projection of any other knot which is topologically
equivalent to the first one, this equivalence is called ambient isotopy.
Planar isotopy is a special case of this, with the restriction that
there are no crossings, only intersections, cf. \prettyref{sec:Some-More-Diagrammatics}.

Finally, we give the definition of an edge as the one used by Penrose
in the discussion above. This edge must be independent of all the
possible linear relations the graphs might have, e.g. skein relations.
Thus, for the $n$ ``strands'' of an edge one sums over all permutations
of the lines considering the sign of the permutation involved:
\begin{defn}
An \textbf{n-edge} (or n-unit) is a set of lines woven into a single
graph denoted by\\

~~~~~~~~~~~~~~~~~~~~~~~~~~~~~~~~~\includegraphics[scale=0.4]{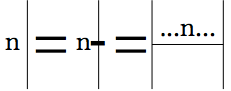},

The weaving is according to the following procedure, \label{def:The-weaving}

~~~~~~~~~~~~~~~~~~~~~~~~~~~\includegraphics[scale=0.5]{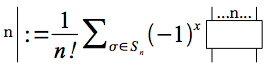},

where $x$ is the minimum number of transpositions needed to generate
the corresponding permutation and the box with the $n$ strands denotes
their permutation $\sigma$, cf. Section \prettyref{sub:Temperley-Lieb-Algebra}.
\end{defn}
As pointed out in \cite{penrose1971applications} this is exactly
(in the case of ordinary tensors) the so-called \emph{generalized
Kronecker delta} 
\[
\left|\begin{array}{cccc}
\delta_{p}^{a} & \delta_{p}^{b} & \dots & \delta_{p}^{f}\\
\delta_{q}^{a} & \delta_{q}^{b} & \dots & \delta_{q}^{f}\\
\vdots & \vdots & \ddots & \vdots\\
\delta_{u}^{a} & \delta_{u}^{b} & \dots & \delta_{u}^{f}
\end{array}\right|=\delta_{pq\dots u}^{ab\dots f}
\]
since each term in the definition above is an outer product of $n$
``Kronecker deltas''.
\begin{rem}
\label{rem:properties-of-antisymmetrizer}\textit{(i)} The antisymmetrizer
of the lines in the definition of the $n$-edge actually symmetrizes
the indices in the $\delta\epsilon-$world since each crossing provides
an extra sign to the corresponding term, \cite{major1999spin}; \textit{(ii)}
The properties of the antisymmetrizer are the following, cf. Section
\prettyref{sub:Temperley-Lieb-Algebra}:\end{rem}
\begin{itemize}
\item \textbf{Irreducibility}, this means that an edge vanishes when a pair
of lines is retracted:\\
\includegraphics[scale=0.4]{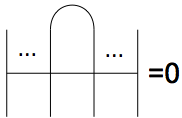}
\item The antisymmetrizer act as \textbf{projectors}:\\
\includegraphics[scale=0.5]{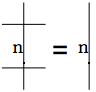}
\item The \textbf{loop value} of an $n$-edge is an integer%
\footnote{\label{fn:loop-value-recursion-relation}The multiplicity $(n+1)$
follows from the recursion relation $\Delta_{n+2}=(-2)\Delta_{n+1}-\Delta_{n}\,;\:\Delta_{0}=1,\,\Delta_{1}=-2$,
cf. Sec. \prettyref{sub:Temperley-Lieb-Algebra}.%
}:\\
\includegraphics[scale=0.5]{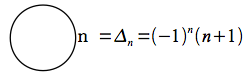} 
\end{itemize}
Now we can define a spin network as consisting of a graph with edges,
vertices and labels. The labels represent the number of strands woven
into edges and any vertex with more than three incident edges must
also be labeled to specify a decomposition into trivalent vertices,
cf. Section \prettyref{sub:Recoupling-Theory} and \cite{rovelli1995spin}.

As the strands from which spin networks are woven can take two values,
they are well-suited to represent two-state systems, hence the name
``spin network''. Thus, it is possible to include the $\left|jm\right\rangle $
representation of angular momentum into the diagrammatics of spin
networks, where each $n$-edge corresponds to the $\frac{n}{2}$-irreducible
representation of SU(2), hence all results of representation theory
have a diagrammatic form. For a detailed description of the angular
momentum representation consult \cite{major1999spin}. In fact, as
described above, this specific representation is only a special case
from a much richer and general theory involving representations of
other algebras which will be explained in the next chapter, but first
it is important to explore the motivation to consider these abstract
objects in the first place. We will do so by considering general relativity
in the framework of the Ponzano-Regge theory.

\section{General Relativity\label{sec:General-Relativity}}

In 1961 T. Regge described an approach to the theory of Riemannian
manifolds that enable the description of general relativity without
the use of coordinates, \cite{regge1961general}. Following the original
work of T. Regge in this section we describe the main ideas of his
approach.

As it is generally known, one can define an intrinsic Gaussian curvature
on any surface by carrying out measurements on geodesic triangles,
\cite[p. 336f]{misner1973gravitation}. This is due to the fact that,
if the triangle $T$ has non Euclidean geometry, in general $\alpha+\beta+\gamma\neq\pi$,
where $\alpha,\beta,\gamma$ are the internal angles of $T$. From
this, one may define the Gaussian curvature $\epsilon_{T}$ of $T$
by $\epsilon_{T}=\alpha+\beta+\gamma-\pi$. 

To get a relation between the area of a spherical triangle $T$ and
the radius $R$ of the sphere consider the area of a segment $S_{\alpha}$
of a sphere. If this segment is built from two arc segments, i.e.
segments of two great circles, with an angle $\alpha$ between them,
leaving both from the point $A$ on the sphere and intersecting in
the opposite point $A'$, then $S_{\alpha}=2\alpha R^{2}$. Now consider
three great circles forming $T$ by intersecting at the points $A,B,C$
and dividing the surface of the sphere in different segments. With
the above relation for each segment of the sphere one obtains the
following relation between the radius $R$, the area $A_{T}$ of the
triangle $T$ and the Gaussian curvature $\epsilon_{T}$ (cf. \cite{Soprunov}),
\[
\epsilon_{T}=\frac{A_{T}}{R^{2}}.
\]

Suppose that $T$ shrinks to a point $P$ and the limit $\frac{\epsilon_{T}}{A_{T}}\underset{A\rightarrow0}{\rightarrow}K(P)$
exists independent of the limiting procedure, then we can take 
\begin{equation}
K(P)=\underset{A\rightarrow0}{lim}\frac{\text{\ensuremath{\epsilon}}_{T}}{A_{T}}\label{eq:def-Gaussian-curvature}
\end{equation}
as the definition of (local) Gaussian curvature at the point $P$. 

Consider a polyhedron M and suppose the triangle $T$ lies in the
interior of a face of $M$ or an edge of $M$ crosses it. Then $\epsilon_{T}=0$
since in the former case T is flat and in the latter the neighborhood
of a point in an edge is homeomorphic to a neighborhood in the plane.
Therefore $K(P)=0$ if $P$ is not a vertex. On the other hand, if
$T$ contains a vertex $V$ the Gaussian curvature is independent
of the form of $T$ and $\epsilon_{T}=\epsilon_{V}$ is a constant
obtained from the relation
\begin{equation}
\epsilon_{V}=2\pi-\sum_{i}\alpha_{i}\label{eq:defect-angle-and-curvature}
\end{equation}
where the sum is over the angles $\alpha_{i}$ at the vertex $V$.
If $T$ contains several vertices, then we have 
\[
\epsilon_{T}=\sum_{V}\epsilon_{V}
\]
so if we think of $K(P)$ as a Dirac distribution with the vertices
V as support, i.e. $K(P)\sim\sum_{V}\delta(P-V)$, we can use the
original formula for the Gaussian curvature $\epsilon_{T}=\int_{T}K(P)dA$.
Thus, from the above, we have a relation between the curvature of
a polyhedron and the deficiency of its vertices, \cite{regge1961general}. 
\begin{example*}
The simplest triangulation of $S^{2}$ is a (regular) tetrahedron.
Hence, the deficiency of each vertex is $\epsilon_{V}=2\pi-(\alpha_{1}+\alpha_{2}+\alpha_{3})=\pi$.
This gives an Gaussian curvature for the tetrahedron of $4\pi$. Compare
this result with the integral curvature of a sphere with radius $R$
and $K=1/R^{2}$, which is $\epsilon_{S}=\int_{S}KdA=4\pi$.
\end{example*}
Now, the Gauss-Bonnet theorem for a closed surface%
\footnote{By closed surface it is meant a compact two-dimensional manifold without
boundary.%
} links the curvature of a manifold $M$ with the Euler characteristic
$\chi(M)=2-2g$, where $g$ is the genus of $M$,
\[
\int_{M}KdA=2\pi\chi(M).
\]

Hence, in the case of a polyhedron, the deficiency angles at the vertices
and the Euler characteristic are connected such that the Poincaré-Euler
formula is obtained directly by first carrying out the sum over all
faces $f$ having $V$ as a common vertex (cf. \prettyref{eq:defect-angle-and-curvature})
following a summation over all vertices and then the other way around,
first vertices and then faces. Since all faces of $M$ are triangles%
\footnote{From this follows that the sum of angles over vertices $V$ and faces
$F$, $\sum_{F,V}\alpha_{FV}=\pi F$. Furthermore, we have the relation
$2E=3F$ between the number of edges and faces of a polyhedron.%
}, we obtain
\[
V-E+F=2-2g
\]
where $V$ is the number of the vertices, $F$ the number of faces
and $E$ the number of edges.

After this small deviation showing the tight relation between geometry
and topology, consider a triangulation of a manifold $M$. Knowing
all lengths of all edges and the connection matrix%
\footnote{The connection matrix contains the information about the relation
between faces, edges and vertices. Hence it corresponds to a metric
in the continuous case.%
} gives us enough information to know also all internal angles $\alpha_{fn}$
of the face $f$ with vertex $V_{n}$. Thus we know all deficiencies
$\epsilon_{n}$ at all vertices and therefore the intrinsic curvature
of $M$. Notice that whenever the number of vertices (and with it
the number of edges and faces) increases the local Gaussian curvature
approximates to a continuous function of the density of vertices $\rho$
and their deficiency $\epsilon$, $K(P)=\rho\epsilon$, when the variation
of the product in the triangle T around the point $P$ is small, cf.
Eq. \prettyref{eq:def-Gaussian-curvature}.

In the case of higher dimensional cell complexes the notion of a geodesic
triangle has to be replaced by the notion of parallel transport, i.e.
an orthogonal mapping between the tangent space $T_{P}M$ at P, and
the tangent space $T_{Q}M$ at another point Q, whenever the points
P and Q are connected by a path $a$ in $M$. If $a$ is a loop, $T_{P}M$
is mapped to itself, the mapping being a rotation around P by an angle
$\epsilon(a)$. Since the rotation is proportional to the curvature
(cf. \cite[Sec. 7.3.2]{Nakahara2003Geometry}), we obtain 
\[
\epsilon(a)=\int_{a}KdA.
\]

This is an additive function of the loops, i.e. if $a,b$ are loops,
then $\epsilon(a\text{\textbullet}b)=\epsilon(a)+\epsilon(b)$, where
\textbullet{} denotes the product of loops.%
\footnote{This product is defined as following: \\
Let $M$ be a topological space. For each two loops $a,b\,:\left[0,\,1\right]\rightarrow M$
at a point P with $a(1)=b(0)$, $a\text{\textbullet}b(t)=\begin{cases}
a(2t)\;0\leq t\leq\frac{1}{2} & b(2t-1)\;\frac{1}{2}\leq t\leq1\end{cases}$ where $t\in\left[0,\,1\right]$. This product is naturally extended
to the homotopy class of the loops in $M$ such that, together with
the homotopy class of inverse loops defined as $a^{-1}(s)\equiv a(1-s),\; s\in\left[0,\,1\right]$
for all $a:\left[0,\,1\right]\rightarrow M$ and the unit element
as the homotopy class of loops homotopic to $P$, it defines a group
called the \textit{fundamental group}, \cite{Nakahara2003Geometry}.%
} Whenever a vertex lies on the curve the parallel transport is not
defined unambiguously, thus we consider the homotopy equivalence only
in $M/V$ where $V$ denotes the set of all vertices in $M$ so the
fundamental group of $M/V$, denoted by $\pi_{1}(M/V)$, is not trivial
if $V\neq\emptyset$. 

Now, each loop at P is associated to an orthogonal matrix $S(a)$
acting on $T_{P}M$. Since a loop homotopic to P is the boundary of
a simply connected region in $M/V$, there is no vertex in that region
of $M$ so the curvature vanishes there. Hence the unit element $[u]\in\pi_{1}(M/V)$
represents the identity matrix in the association made before. On
the other hand, consider a loop $c=a\text{\textbullet}b$ which is
not homotopic to P. Then the vector $x\in T_{P}M$ will be parallel
transported first along $a$ then along $b$, such that the resulting
rotation of $x$ is associated to a matrix $S(c)=S(b)S(a)$, where
$S(c),\, S(b),\, S(a)$ are the matrices associated to the loops $c,\, b,\, a$
respectively. In fact, the association is not between loops and orthogonal
matrices, but rather the homotopy classes, i.e. the elements of $\pi_{1}(M/V)$,
are represented by orthogonal matrices; to see this, we use the fact
that if $c,\, c'$ are loops in the same class, then $c'=v\text{\textbullet}c$
where $v\in[u]$. Thus, with the previous, we have the relation $S(c')=S(c)S(v)=S(c)$,
\cite{regge1961general}.

As an example, consider a one-sheeted cone, which is a polyhedron
with one vertex only. It can be parametrized similar to a plane by
polar coordinates $(\rho,\,\theta)$ such that the metric is given
by $ds^{2}=d\rho^{2}+\rho^{2}d\theta^{2}$, but with the identification
of points with the same $\rho$ and angles $\theta$ differing by
a multiple of $2\pi-\epsilon$, where $\epsilon$ is the deficiency
of the vertex $\rho=0$. This manifold is called the $\epsilon-$cone
and it can be extended to higher dimensional spaces by considering
the product $\mathbb{R}^{n-2}\times\epsilon-cone$ with the metric
$ds^{2}=d\bar{z}^{2}+d\rho^{2}+\rho^{2}d\theta^{2}$ where $\bar{z}\in\mathbb{R}^{n-2}$.
This manifold, called the $\epsilon-n-cone$, is Euclidean everywhere
but the $n-2$ dimensional flat subset $\rho=0$. 

Now, consider a loop in the $\epsilon-3-$cone around the line $\rho=0$,
called the \textbf{bone}, with $\theta(0)=0$ and $\theta(1)=N(2\pi-\epsilon)$,
where $ $$N$ gives the winding number. Since only loops that are
completed around the bone are not null homotopic, i.e. homotopic to
a point, loops with the same N belong to the same homotopy class.
Let the homotopy class be denoted by $a(N)$, then the multiplication
of two equivalent classes is given by $a(N)\text{\textbullet}a(M)=a(N+M)$,
so the fundamental group is isomorphic to $(\mathbb{Z},\,+)$. The
orthogonal matrix associated to N is found by considering the fact
that $S(0)=1_{n\times n}$ and $S(M)S(N)=S(N+M)$ must hold, thus
we can write $S(N)=S^{N}$ where $S$ is called the generator of the
bone. Let $V\in T_{P}\mathbb{R}^{n-2}\times\epsilon-cone$ be a vector
in P, which can be split into orthogonal components $V^{||},\, V^{\perp}$
lying in the subspaces $\mathbb{R}^{n-2}$ and $\epsilon-cone$ respectively.
If we take $V$ along $a(1)$, the component $V^{||}$ will remain
unaffected by the process, i.e. $S(1)V^{||}=V^{||}$, so the subspace
$\mathbb{R}^{n-2}$ is invariant under the action of the orthogonal
$n\times n-$matrix $S(N)$. However, the component $V^{\perp}$ will
be rotated by the angle $\epsilon$. Hence, the general form $S(N)$
is a $n\times n-$matrix with two block-matrices in the diagonal,
the first one is the identity matrix $1_{(n-2)\times(n-2)}$, and
the second one is a rotation matrix describing a rotation on the $\epsilon-cone$
by the angle $N\epsilon$. In the particular case of an $\epsilon-n-$cone
the fundamental group and its associated transformation group of orthogonal
$n\times n-$matrices are Abelian, \cite{regge1961general}.

The $\epsilon-n-$cone described above is an example of n-dimensional
generalizations of polyhedra, called skeleton spaces, where the curvature
of the manifold $M$ is a consequence of a $(n-2)-$dimensional subset
$w\subset M$, called the skeleton of $M$. For the definition of
these generalizations we start from a simplectic decomposition of
the n-dimensional manifold. This determines the topology of this space
but not the metric, which has to be defined.
\begin{defn}
As defined by Regge in \cite{regge1961general}, the (Euclidean) metric
in $M$ is given by the following axioms:\end{defn}
\begin{enumerate}
\item The metric in the interior of any n-dimensional closed simplex $T_{n}$
is Euclidean%
\footnote{This means that if one defines a coordinate system in $T_{n}$, one
can also give the coordinates of the points of the boundary of $T_{n}$
\uline{in this frame}.%
}.
\item In the metric of $T_{n}$, its boundary is decomposable into the sum
of $n+1$ closed simplexes $T_{n-1}$, which are flat.
\item If a simplex $T_{n-1}$ is common boundary of $T_{n}$ and $T'_{n}$,
the distance of any two points of $T_{n-1}$ is the same in both frames
of $T_{n}$ and $T'_{n}$.%
\footnote{Even if the coordinates might be different, the distance between two
points in $T_{n-1}$ is invariant under the change of coordinates
from one simplex to the other. %
}
\item If $P\in T_{n},\; P'\in T'_{n}$ and $P,\, P'$ are ``close enough''
to $T_{n-1}$ we define the distance $PP'$ as $d(P,\, P')=\inf_{Q\in T_{n-1}}\{PQ+QP'\}$.
\end{enumerate}
On the interior of $T_{n-1}$ the metric is well defined and Euclidean
as well. However, it might not be the case on the $T_{n-2}$'s of
the boundary of $T_{n-1}$. The problem with the metric on $T_{n-2}$
is that, with these axioms, it is not possible to define a metric
since starting with the metric on a $T_{n-1}$ that is common boundary
of some n-simplexes does not ensure that at its boundary, i.e. at
the $T_{n-2}$ which is common boundary of \textit{other} $T_{n-1}$'s,
the metric will continue to be well defined, it could be that the
metric defined in other adjacent $T_{n-1}$ is different such that
the metric would not be unambiguous.

This definition is enough to join smoothly neighboring simplexes so
that one can construct a manifold which is everywhere Euclidean in
$M/w$. In the 3-dimensional case the bones are straight segments
connecting two 0-simplexes, so locally the bone is an $\epsilon-3-$cone.
If the lengths of these bones are chosen at random, the sum of all
dihedral angles around the $i$-th bone will be $2\pi-\epsilon_{i}$. 

Consider a 0-simplex $T_{0}$ in the above 3-dimensional manifold
$M_{3}$ being the common end of several oriented bones $T_{1}^{1},\dots,T_{1}^{m}$.
The orientation of the bones is such that all $T_{1}^{i}$'s have,
from $T_{0}$, outgoing direction; such a 0-simplex is called an oriented
m-joint. Consider an idealized isolated m-joint where all the bones
extend up to infinity and the set of plane sectors $A_{p}$ formed
by two contiguous bones $T_{1}^{p},\, T_{1}^{p+1}$ have the property
that $A_{p}\cap A_{p+1}=T_{1}^{p+1}\;\forall p\in\{1,\dots,\, m\}$
and $A_{i}\cap A_{j}=\emptyset\;\forall i,\, j\in\{1,\dots,\, m\}$,
with $j\neq i+1$, where we identify $m+1\triangleq1$. Then $\mathcal{A}=\underset{p}{\cup}A_{p}$
divides the space $M$ in two regions $M'\text{ and }M''$. Suppose
that $P\in M'$, $Q\in M''$ and that there are m paths $t_{p}$ connecting
P to Q. Let $t_{p}$ be such that it intersects $\mathcal{A}$ only
one time and only at $A_{p}$, therefore $a_{p}=t_{p}t_{p+1}^{-1}$
is a loop%
\footnote{The definition of the product of two paths is the same as the one
for the product of loops. Here we do not write explicitly the dot
\textbullet{}. %
} around $T_{1}^{p}$. Notice that $a_{p}$ encircles $T_{1}^{p}$
only once, thus we have the identity $a_{1}\text{\textbullet}a_{2}\text{\textbullet}\dots\text{\textbullet}a_{p}=t_{1}t_{2}^{-1}\text{\textbullet}t_{2}t_{3}^{-2}\text{\textbullet}\dots\text{\textbullet}t_{m}t_{m+1}^{-1}=u$,
which translate for the generators
\begin{equation}
S_{1}S_{2}\dots S_{m}=1_{3\times3}\label{eq:fundamental-relation-regge}
\end{equation}

\begin{example}
Consider a 4-joint. Then we have the relation $S_{1}S_{2}S_{3}S_{4}=1_{3\times3}$,
where $S_{i}$ is the representation of the homotopy class corresponding
to one loop around the i-th bone. Since our manifold is three-dimensional
we can write this representation as
\[
S_{i}=\left(\begin{array}{ccc}
1 & 0 & 0\\
0 & cos(\epsilon_{i}) & -sin(\epsilon_{i})\\
0 & sin(\epsilon_{i}) & cos(\epsilon_{i})
\end{array}\right)
\]
where $\epsilon_{i}$ is the deficiency at the i-th bone. Carrying
out the multiplication of the matrices we find in this case 
\[
\sum_{i=1}^{4}\epsilon_{i}=2\pi n.
\]

\end{example}
The relation \prettyref{eq:fundamental-relation-regge} is the geometrical
version of the Bianchi's identity encountered in a differential Riemannian
manifold, as shown in \cite{regge1961general}. In order to see this,
Regge first shows the relation between the concepts discussed previously
and the (coordinate expression of) Riemann curvature tensor by considering
the transition from a skeleton space to a differentiable manifold,
which is accomplished by increasing the density $\rho$ of the bones
but keeping the local curvature $\rho\epsilon$ slowly varying. 

Consider a bundle of parallel bones on $M_{3}$ which induce a small
curvature. Test the curvature of the manifold by carrying a vector
$V$ along a small loop with area%
\footnote{We change the notation at this point to emphasize that the area has
a normal vector $\overrightarrow{n}$.%
} $\overrightarrow{\Sigma}=\Sigma\overrightarrow{n}$. Hence, the vector
$V$ would be rotated around a unit vector $U$ parallel to the bones
by the angle $\sigma=N\epsilon$, where $N$ is the number of bones
through $\Sigma$ and it was assumed that each bone contribute the
same deficiency $\epsilon.$ Assuming a uniform density $\rho$ of
bones and $\overrightarrow{\Sigma}||U$, we have $N=\rho(U,\overrightarrow{\Sigma})=\rho\Sigma$,
thus $\sigma=\rho\epsilon\Sigma\text{ or }\rho\epsilon\Sigma_{\gamma}U^{\gamma}$.
Suppose the rotation is infinitesimal, such that in first order the
rotated vector can be expressed as $V'=(1+\sigma U\overrightarrow{S})V$
where $S_{\kappa}$ are the generators of the rotation around $U$.
Hence, for the change of the vector we obtain%
\footnote{With the following notation is intended to signalize that, even if
the example is given for a 3-Manifold, the results can be generalized
for higher dimensional cases.%
}
\begin{eqnarray*}
\Delta V & = & \sigma(U\overrightarrow{S})V=\sigma U\wedge V,\\
\Delta V^{\alpha} & = & (\rho\epsilon\Sigma_{\kappa}U^{\kappa})(\epsilon^{\alpha\beta\gamma}U_{\beta}V_{\gamma}).
\end{eqnarray*}

Inserting $\delta_{\kappa}^{\tau}=\frac{\text{1}}{2}\epsilon^{\delta\omega\tau}\epsilon_{\delta\omega\kappa}$
and defining $U^{\alpha\gamma}=\epsilon^{\alpha\gamma\beta}U_{\beta}$
and $\Sigma^{\delta\omega}=\epsilon^{\delta\omega\tau}\Sigma_{\tau}$
we obtain
\[
\Delta V^{\alpha}=\frac{1}{2}\rho\epsilon U^{\gamma\alpha}U_{\delta\omega}\Sigma^{\delta\omega}V_{\gamma}
\]
which we compare with the standard form of the coordinate expression
for the Riemann curvature tensor (cf. \cite{Nakahara2003Geometry})
\[
\Delta V^{\mu}=R_{\quad\lambda\nu}^{\kappa\mu}\varepsilon^{\lambda}\delta^{\nu}V_{\kappa};\text{ where }\varepsilon^{\lambda}\delta^{\nu}=\frac{1}{2}\epsilon^{\lambda\nu\tau}\Sigma_{\tau}.
\]

Hence we obtain the relation for the approximation to differential
manifolds given by Regge
\[
R_{\mu\sigma\alpha\beta}=\rho\epsilon U_{\mu\sigma}U_{\alpha\beta}
\]
which has all symmetry properties of the Riemann tensor, like the
first Bianchi's identity $ $$U_{\mu\sigma}U_{\alpha\beta}+U_{\mu\alpha}U_{\beta\sigma}+U_{\mu\beta}U_{\sigma\alpha}=0$
and with $U_{\mu\sigma}U^{\mu\sigma}=2$ we have that the scalar curvature
$R=2\rho\epsilon$ is independent of the dimension $n$ of the space
and twice the Gaussian curvature $K(P)$. Observe also that the orientation
of a bone in $M_{n}$ is determined by the skew symmetric tensor $U_{\mu\sigma}$,
since in the case of a parallel bundle of bones, we are able to choose
coordinates $\{x_{i};\, i=1,\dots,n\}$, such that $x_{1},\, x_{2}$
are perpendicular and $ $$x_{j}\,(j\geq3)$ are parallel to the bone.
Thus, $U_{12}=-U_{21}=1$ while the other components of the tensor
vanish.
\begin{rem*}
\textit{(i)} If the deficiencies $\epsilon_{p}$ are small, \prettyref{eq:fundamental-relation-regge}
can be written as
\[
\sum_{p=1}^{m}\epsilon_{p}U_{\mu\sigma}^{p}=0.
\]

\textit{(ii)} Assuming a distribution $\rho$ of identical m-joints
we can write the Riemann tensor as a sum over the $m$ bundles of
parallel bones, where the $p$-th bundle has defect $\epsilon_{p}$,
density $\rho_{p}$ and direction $U^{p}$, as follows:
\[
R_{\alpha\beta\lambda\delta}=\sum_{p=1}^{m}\epsilon_{p}\rho_{p}U_{\alpha\beta}^{p}U_{\lambda\delta}^{p}.
\]
This leads to the correspondence between \prettyref{eq:fundamental-relation-regge}
and the Bianchi's identities, \cite{regge1961general}.

\textit{(iii)} For a non-positive definite metric of a skeleton space,
the generators of the bones are Lorentz matrices. In the case $n=4$
the bones are triangles with area 
\begin{equation}
4L^{2}=(A_{\mu}B^{\mu})^{2}-(A_{\mu}A^{\mu})(B_{\mu}B^{\mu}),\label{eq:appropiate-bone-area}
\end{equation}
$A_{\mu},\, B_{\mu}$ being two sides of the triangle. For an indefinite
metric with one time coordinate and the metric signature $(-+++)$
there are three types of bones: spacelike, null and timelike. Note
that the form of the above relation has the opposite sign with respect
to the relation for the area of a triangle in the three dimensional
case, which follows from $2L=A\wedge B$. This is due to the fact
that, if $B=(B_{0},0)$ and $A$ is lying in the $z-$direction, then
$L$ lies in the $xy-$plane, thus its magnitude has to be thought
of as a \textit{real} quantity. Since in this case $A_{\mu}B^{\mu}=0$,
$B^{2}<0$ and $A^{2}>0$ it follows that $A^{2}B^{2}<0$ but $L^{2}>0$
in order for $L$ to have a real magnitude, \cite{misner1973gravitation}.
Therefore the relation \prettyref{eq:appropiate-bone-area} is the
appropriate one and we consider only timelike bones, such that the
deficiencies are all real.
\end{rem*}
To get the field equations in this framework Regge considered approximating
an Einstein space with the action
\[
\mathcal{S}=\frac{1}{16\pi}\int R\sqrt{-g}d^{4}x
\]
with a skeleton space using the variational principle for which $\mathcal{S}$
is calculated on the skeleton. Recalling the discussion at the beginning
of this section, notice that there is no contribution to the action
coming from the interior of the $T_{4}$ and $T_{3}$, this means
that $R$ is a distribution with support $w$. Hence, it can be expressed
in terms of the deficiencies and areas of the bones, labelled by $n$:
\[
\mathcal{S}=\frac{1}{8\pi}\sum_{n}L_{n}\epsilon_{n}.
\]

Since the lengths $l_{p}$ of all $T_{1}^{p}$ contain the same information
as the metric tensor, the variation of $\mathcal{S}$ is carried out
varying these lengths. Regge proved in the appendix of \cite{regge1961general}
that it is possible to carry out the variation as if the deficiencies
$\epsilon_{n}$ were constants, thus we obtain a set of field equations
of the form%
\footnote{Here we only consider the 4-dimensional case, but it can be generalized
by taking $L_{n}$ as the $(m-2)$-dimensional measure of $T_{m-2}^{n}$.%
}
\[
\sum_{n}\epsilon_{n}\frac{\partial L_{n}}{\partial l_{p}}=0
\]
where $p$ runs through all the 1-simplexes in the decomposition of
the manifold. With the help of the law of cosines $l_{p}^{2}=l_{p'}^{2}+l_{p''}^{2}-2l_{p'}l_{p''}cos\theta_{pn}$,
for the angle $\theta_{pn}$ opposite to $l_{p}$ in $T_{2}^{n}$,
the fact that $l_{p}=\sqrt{A_{\mu}A^{\mu}}$ for the edge $A_{\mu}$,
and $\frac{\partial L^{2}}{\partial l_{p}}=2L\frac{\partial L}{\partial l_{p}}$
we obtain 
\[
\frac{\partial L_{n}}{\partial l_{p}}=\frac{1}{2}l_{p}ctg\theta_{pn}.
\]

Hence, Einstein's equations%
\footnote{It can be shown that these equations degenerate into $R_{\mu\nu}=0$
for a differential manifold, \cite{regge1961general}. Hence, the
equations \prettyref{eq:combinatorial-EFE-empty-space} describe the
combinatorial version of an empty space without cosmological constant. %
} for an empty space in the Regge approximation are given by
\begin{equation}
\sum_{n}\epsilon_{n}ctg\theta_{pn}=0,\label{eq:combinatorial-EFE-empty-space}
\end{equation}
 where the sum runs over all 2-simplexes that have the $p$-th edge
in common. 

To summarize, Regge calculus is an approximation of a smoothly curved
$n$-dimensional Riemannian manifold in terms of a collection of $n$-dimensional
simplexes, each of them being flat in its interior and joined at their
lower-dimensional faces. The curvature of the manifold is contained
in the deficiency angles of the $(n-2)$-skeleton geometry together
with its measure. By allowing the number of $(n-2)$-faces to go to
infinity and keeping the limit \prettyref{eq:def-Gaussian-curvature}
slowly varying one obtains the volume integral of the scalar curvature
of the original smooth manifold.

\section{Connection between General Relativity and Spin Networks\label{sec:Connection-between-GR-SN}}

Once we have the concept of spin networks and the Regge calculus an
immediate question regarding the relation between both arises. How
can an object developed from basic notions of quantum theory for which
there was no notion of space-time be related to a conceptually different
theory, where space-time is the fundamental object? In 1969 G. Ponzano
and T. Regge gave us some insight into this relation by explaining
the relation between the asymptotic formula for the 6j-symbols, which
are the coupling coefficients for a system consisting of three spins,
and the path integral over the exponential of the integral over the
Lagrangian density in a 3-dimensional Einstein theory, which resembles
in a remarkable way a Feynman summation over histories as in QFT,
\cite{ponzano1969semiclassical}. In this section we will follow Ponzano's
and Regge's article leaving out the details and concentrating in the
main points which are the asymptotic form and physical interpretation
of the 6j-symbols, their generalization to 3nj-symbols and the rise
of the ``Feynman integral'' out of the asymptotic formula.

Before we start describing the above mentioned relation, one should
notice the association of the 6j-symbol $\left\{ \begin{array}{ccc}
a & b & c\\
d & e & f
\end{array}\right\} $ to a tetrahedron $T$ depicted by

\begin{center}
\includegraphics[scale=0.5]{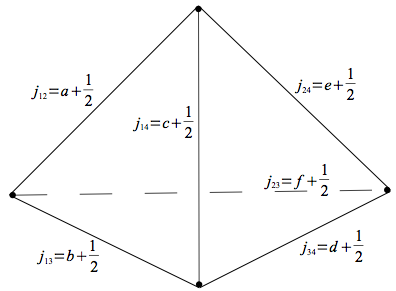}
\par\end{center}

where the length of an edge is chosen to be $a+\frac{1}{2}$ for a
better approximation, rather than just $a$, to the length $\left[a(a+1)\right]^{\frac{1}{2}}$
for higher quantum numbers. A sufficient condition for the existence
of a non-vanishing 6j-symbol is the restriction to those values of
the angular momenta which satisfy the triangular inequalities for
the set of three edges building a triangle, e.g. $|b-c|\leq a\leq b+c$.
This is, however, not enough for the existence of $T$. With the help
of Tartaglia's formula for the square of the volume $V$ of the tetrahedron
\[
2^{3}(3!)^{2}V^{2}=\left|\begin{array}{ccccc}
0 & j_{34}^{2} & j_{24}^{2} & j_{23}^{2} & 1\\
j_{34}^{2} & 0 & j_{14}^{2} & j_{13}^{2} & 1\\
j_{24}^{2} & j_{14}^{2} & 0 & j_{12}^{2} & 1\\
j_{23}^{2} & j_{13}^{2} & j_{12}^{2} & 0 & 1\\
1 & 1 & 1 & 1 & 0
\end{array}\right|
\]
we see that the condition $V^{2}\geq0$ is a necessary and sufficient
condition for the existence of $T$. In what follows we will always
assume that this condition is satisfied.%
\footnote{In \cite{ponzano1969semiclassical} the case where $V^{2}<0$ is also
treated in an analogous way to the WKB approximation method. In this
region, the 6j-symbols have an exponential decrease, hence they do
not vanish completely, but it can be interpreted as the impossibility
of having six angular momenta forming a non-existing T.%
}

\subsection{Asymptotic formula for the 6j-symbols}

The $6j$-symbols $\left\{ \begin{array}{ccc}
j_{1} & j_{2} & j_{12}\\
j_{3} & J & j_{23}
\end{array}\right\} $ are defined as recoupling coefficients of the quantum mechanical
coupling of three angular momenta%
\footnote{In the semiclassical limit we expect to be able to write the angular
momenta as vectors in a three dimensional Euclidean space and the
coupling as a vector sum of the angular momenta involved.%
} $\overrightarrow{j_{1}}$, $\overrightarrow{j_{2}}$ and $\overrightarrow{j_{3}}$
to $\overrightarrow{J}$. The defining relation is, similar to the
Clebsch-Gordan coefficients, the change of basis from a given configuration
of the angular momenta to another, both represent the two different
ways of coupling $\overrightarrow{j_{1}}$, $\overrightarrow{j_{2}}$
and $\overrightarrow{j_{3}}$ to a given $\overrightarrow{J}$, by
(1) first coupling $\overrightarrow{j_{1}}$ and $\overrightarrow{j_{2}}$
to $\overrightarrow{j_{12}}$ and \emph{then} coupling this with $\overrightarrow{j_{3}}$
or (2) first coupling $\overrightarrow{j_{2}}$ and $\overrightarrow{j_{3}}$
to $\overrightarrow{j_{23}}$ and \emph{then }coupling this to $\overrightarrow{j_{1}}$%
\footnote{The order of coupling determines the phase of the resulting state,
cf. \cite{edmonds1996angular}.%
}. The resulting basis vectors for the angular momenta are in general
different and are related by
\[
\left|(j_{1},\,(j_{2},j_{3})\, j_{23}),\, J\right\rangle =\sum_{j_{12}}\left[(2j_{12}+1)(2j_{23}+1)\right]^{\frac{1}{2}}(-1)^{j_{1}+j_{2}+j_{3}+J}\left\{ \begin{array}{ccc}
j_{1} & j_{2} & j_{12}\\
j_{3} & J & j_{23}
\end{array}\right\} \left|((j_{1},j_{2})\, j_{12},\, j_{3}),\, J\right\rangle .
\]

Consider the option (1), then the probability that the sum of $\overrightarrow{j_{2}}$
and $\overrightarrow{j_{3}}$ has length in the interval $\left[j_{23},\, j_{23}+dj_{23}\right]$
is given by 
\begin{equation}
(2j_{12}+1)(2j_{23}+1)\left\{ \begin{array}{ccc}
j_{1} & j_{2} & j_{12}\\
j_{3} & J & j_{23}
\end{array}\right\} ^{2}dj_{23}\label{eq:6j-symb-and-probability}
\end{equation}

On the other hand, if $j_{1},\, j_{2},\, j_{3},\, j_{12},\, J$ are
fixed the dihedral angle $\theta$ between the planes spanned by $\overrightarrow{j_{1}},\,\overrightarrow{j_{2}}$
and $\overrightarrow{j_{3}},\,\overrightarrow{J}$ is undetermined.
Suppose that every configuration giving $\theta$ has equal probability
$|d\theta|/(2\pi)$, then the probability for the length $j_{23}$
to have the value between $j_{23}$ and $j_{23}+dj_{23}$ is given
by%
\footnote{The factor of 2 comes from the fact that there are two configurations
giving the same value of $j_{23}$.%
}
\begin{equation}
\left\{ \frac{2}{(2\pi)}\left|\frac{d\theta}{dj_{23}}\right|\right\} dj_{23}=\frac{1}{6\pi}\frac{j_{12}j_{23}}{V}\label{eq:probability-volume}
\end{equation}
where $V$ is the volume of the tetrahedron spanned by all vectors
involved and the following relation%
\footnote{\label{fn:Area-Volume-relation-footnote}This relation follows from
\[
A_{h}A_{k}sin\theta_{hk}=\frac{3}{2}Vj_{hk}\text{ and }cos\theta_{hk}=-\frac{9}{A_{h}A_{k}}\frac{\partial V^{2}}{\partial(j_{rs}^{2})}\text{ for }h\neq k\neq r\neq s
\]

given in \cite[Appendix B]{ponzano1969semiclassical}.%
} was used 
\[
\frac{d\theta_{hk}}{dj_{rs}}=-\frac{j_{rs}j_{hk}}{6V}\text{ for }h\neq k\neq r\neq s.
\]

Comparing \prettyref{eq:6j-symb-and-probability} with \prettyref{eq:probability-volume}
we obtain for large angular momenta a result, due to Wigner, relating
the volume $V$ of a tetrahedron with the square of its corresponding
$6j$-symbol:
\[
\left\{ \begin{array}{ccc}
j_{1} & j_{2} & j_{12}\\
j_{3} & J & j_{23}
\end{array}\right\} ^{2}\propto\frac{1}{24\pi V}
\]

This result is, however, unacceptable since the numerical calculations
show that the symbols are, in fact, rapidly oscillating functions
of the indices. Hence, the above proportionality suggests that it
is rather an approximation of the average of the 6j-symbol over a
large enough interval of values its indices. Ponzano and Regge claimed
that a correct approximation would be of the form
\[
\left\{ \begin{array}{ccc}
j_{1} & j_{2} & j_{12}\\
j_{3} & J & j_{23}
\end{array}\right\} \simeq\frac{1}{\sqrt{12\pi V}}\mathfrak{C}
\]
 where $\mathfrak{C}$ is a rapidly oscillating function such that
the average over a long interval of indices gives $ $$\left\langle \mathfrak{C}^{2}\right\rangle =\frac{1}{2}$.

In order to find out which form $\mathfrak{C}$ could have, Ponzano
and Regge gave heuristic arguments, which we will give in a very short
fashion. It is known, \cite{edmonds1996angular}, that following relation
between 6j-symbols and matrix representations%
\footnote{Here the index $f$ gives the dimension $(2f+1)$ of the representation
and $\delta,\,\delta'$ are the matrix indices.%
} of the rotation group holds
\begin{equation}
\left\{ \begin{array}{ccc}
c & a & b\\
f & b+\delta & a+\delta'
\end{array}\right\} \simeq\frac{(-1)^{a+b+c+f+\delta}}{\sqrt{(2a+1)(2b+1)}}d_{\delta,\delta'}^{(f)}(\theta)\label{eq:6j-symbol-matrix-elem-rotgrp}
\end{equation}
where $a,\, b,\, c\gg f,\,\delta,\,\delta'$ and $\theta$ is the
angle between the edges $a$ and $b$ intersecting at $f$ and given
by the law of cosines
\[
cos\theta=\frac{a(a+1)+b(b+1)-c(c+1)}{2\sqrt{a(a+1)b(b+1)}}\,;\text{ for }0\leq\theta\leq\pi.
\]

Relation \prettyref{eq:6j-symbol-matrix-elem-rotgrp} together with
the Biedenharn-Elliott identity \prettyref{eq:Biedenharn-Elliott}
give a result%
\footnote{cf. Appendix C in \cite{ponzano1969semiclassical}.%
} suggesting the following general formula for the asymptotic form
of the 6j-symbols
\begin{equation}
\left\{ \begin{array}{ccc}
a & b & c\\
d & e & f
\end{array}\right\} \simeq\frac{1}{\sqrt{12\pi V}}cos\left(\sum_{h,k=1}^{4}j_{hk}\theta_{hk}+\frac{\pi}{4}\right),\label{eq:Asymptotic-formula-6j-symbols}
\end{equation}
where $j_{hk}=j_{kh},\, j_{hh}=0,\text{ and }\theta_{hk}=\theta_{kh}$
is given by
\[
cos\theta_{hk}=-\frac{9}{A_{h}A_{k}}\frac{\partial V^{2}}{\partial(j_{rs}^{2})}
\]
$\text{ for }h\neq k\neq r\neq s$, and $A_{h}$ is the area of the
face opposite to the vertex $h$, cf. Footnote \prettyref{fn:Area-Volume-relation-footnote}.
\begin{rem*}
Even if \prettyref{eq:Asymptotic-formula-6j-symbols} has no strict
formal derivation, there are several arguments for accepting this
relation:\end{rem*}
\begin{enumerate}
\item Numerical accuracy which improves as the values of the angular momenta
increase.
\item Invariance under the full symmetry group of the 6j-symbols. 
\item The formula satisfies asymptotically the Biedenharn-Elliott identity
as well as all identities which are enough to derive all properties
and numerical values of the 6j-symbols.
\end{enumerate}

\subsection{The 3nj-symbols and their relation to general relativity\label{sub:The-3nj-symbols}}

We consider now the generalizations of the $6j$-symbols for systems
with a higher number of edges, called the $3nj$-symbols. These are
best described by the use of diagrams that provide combinatorial information
on how angular momenta are coupled. Since the coupling process is
always the same%
\footnote{ Two angular momenta are coupled, then a third is coupled to the resulting
one, giving a new angular momenta which may be coupled to a forth
one and so on; for instance, $(\overrightarrow{j_{1}}+\overrightarrow{j_{2}})+\overrightarrow{j_{3}}=\overrightarrow{j_{12}}+\overrightarrow{j_{3}}=\overrightarrow{J}$.%
} in a given scheme, where there are more than three angular momenta
involved, the recoupling coefficients arising can be expressed in
terms of 6j-symbols. Moreover, it is possible to write these symbols
in terms of four 3j-symbols%
\footnote{The 3j-symbols where defined by Wigner as the recoupling coefficients
of three angular momenta to a zero one. In the diagrammatic notation,
they correspond to a trivalent vertex as in \cite{edmonds1996angular},
or its dual diagram, a triangle, which is easier to relate to the
vector coupling picture of three vectors adding up to a zero vector.%
} as follows (cf. \cite[Sec. 6.2]{edmonds1996angular})
\begin{eqnarray*}
\left\{ \begin{array}{ccc}
j_{1} & j_{2} & j_{3}\\
j_{4} & j_{5} & j_{6}
\end{array}\right\}  & = & \sum_{all\, m}\left(\begin{array}{ccc}
j_{1} & j_{2} & j_{3}\\
m_{1} & m_{2} & m_{3}
\end{array}\right)\left(\begin{array}{ccc}
j_{1} & j_{5} & j_{6}\\
m'_{1} & m_{5} & m'_{6}
\end{array}\right)\times\\
 & \times & \left(\begin{array}{ccc}
j_{4} & j_{2} & j_{6}\\
m'_{4} & m'_{2} & m_{6}
\end{array}\right)\left(\begin{array}{ccc}
j_{4} & j_{5} & j_{3}\\
m_{4} & m'_{5} & m'_{3}
\end{array}\right)\left(\begin{array}{c}
j_{1}\\
m_{1}\, m'_{1}
\end{array}\right)\left(\begin{array}{c}
j_{2}\\
m_{2}\, m'_{2}
\end{array}\right)\times\\
 & \times & \left(\begin{array}{c}
j_{3}\\
m_{3}\, m'_{3}
\end{array}\right)\left(\begin{array}{c}
j_{4}\\
m_{4}\, m'_{4}
\end{array}\right)\left(\begin{array}{c}
j_{5}\\
m_{5}\, m'_{5}
\end{array}\right)\left(\begin{array}{c}
j_{6}\\
m_{6}\, m'_{6}
\end{array}\right),
\end{eqnarray*}
 where
\[
\left(\begin{array}{c}
j\\
m\, m'
\end{array}\right)=(-1)^{j+m}\delta_{m,-m'}.
\]

Furthermore, it is always possible to evaluate a recoupling graph
of a compact semi-simple group $G$ as a sum of products of Racah
coefficients of $G$. This result is known by the name of \textbf{Decomposition
Theorem}, cf. \cite[sec. 4.2]{moussouris1983quantum} and section
\prettyref{sub:The-Decomposition-Theorem}. Hence, a diagram $D$
is essentially a clear notation for the expansion of a 3nj-symbol
$\left[D\right]$ in terms of 3j-symbols:
\[
\left[D\right]=\sum_{all\, m}\left(\begin{array}{ccc}
j_{1} & j_{2} & j_{3}\\
m_{1} & m_{2} & m_{3}
\end{array}\right)\dots\left(\begin{array}{ccc}
j_{p} & j_{q} & j_{r}\\
m_{p} & m_{q} & m_{r}
\end{array}\right)\left(\begin{array}{c}
j_{r}\\
m_{r}\, m'_{r}
\end{array}\right)\left(\begin{array}{ccc}
j_{r} & j_{s} & j_{t}\\
m'_{r} & m_{s} & m_{t}
\end{array}\right)\dots
\]

In fact, $D$ is a 2-dimensional simplicial complex with a 1-to-1
correspondence between 1- and 2-simplexes in $D$ and angular momenta
$j_{i}$ and 3j-symbols respectively in the r.h.s. of $\left[D\right]$.
Furthermore, the boundary of a face (i.e. a 2-simplex) is the sum
of the three edges in the corresponding 3j-symbol and we disregard
the orientation of the simplexes. From the general form of $\left[D\right]$,
observe that there are 2n faces and 3n edges in $D$.%
\footnote{The value of $n$ is related to the Euler characteristic:
\[
\chi=f-e+v=n+v
\]

for $f$ faces, $e$ edges and $v$ vertices.%
} Since there is a certain degree of arbitrariness regarding the position
of the points, regard diagrams as equivalent whenever they yield the
same symbol $\left[D\right]$. If only topological properties are
considered, we use the word diagram to denote the equivalent class.
On the other hand a configuration is an element of the equivalent
class, i.e. a diagram with additional information about the angles
and/or length of the edges, which removes the ambiguities. 

Now we will give a rough description on how to calculate the 3nj-symbol
$\left[D\right]$ for any diagram $D$. First introduce a diagram
$\mathfrak{D}(D)$ defined as a 3-dimensional simplicial complex with
$\partial\mathfrak{D}(D)=D$. The vertices, edges and faces of $\mathfrak{D}(D)$
are called external if they belong to the boundary, otherwise internal.
We denote the cells of $\mathfrak{D}(D)$, which are 3-simplexes,
by $T_{k}\,(k=1,2,\dots,p)$ and all internal edges are labelled by
$x_{i}\,(i=1,2,\dots,q)$ and external edges by $l_{j}\,(j=1,2,\dots,r)$.
To each cell $T_{k}$ we associate a function $\left[T_{k}\right]$
of the internal and, if the case is given, external edges, which constitute
the diagram $T_{k}$. Then, form the product of all symbols $\left[T_{k}\right]$,
the resulting function will be a function of all internal labels:
\[
A(x_{1},\dots,x_{q}):=\prod_{k=1}^{p}\left[T_{k}\right](-1)^{\chi}\prod_{i=1}^{q}(2x_{i}+1)
\]

where%
\footnote{\label{fn:fixing-chi}This formula for the Euler characteristic is
only true if we assume that $D$ and $\mathfrak{D}(D)$ are homeomorphic
to $S^{2}$ and the 3-ball respectively, which is assumed in the following.%
} $\chi=\sum_{j=1}^{q}(n_{j}-2)x_{j}+\chi_{0}$. The constant $\chi_{0}$
is a fixed phase to make $\chi$ an integer and $n_{j}$ is the number
of tetrahedra $T$ with the common edge $x_{j}$. 

Now, consider the sum over all internal variables
\[
S:=\sum_{x_{1},\dots,x_{q}}A(x_{1},\dots,x_{q})
\]
 If there are no internal vertices, a case which is always possible
to realize, then the sum is finite and $S=\left[D\right]$. On the
other hand, in the more general case where there are internal vertices,
the sum becomes infinite but in some simple cases it is possible
to renormalize%
\footnote{In fact, the general case in which there are internal vertices was
not discussed in \cite{ponzano1969semiclassical}, but only the simple
case of a tetrahedron with an internal vertex was shown to converge
with the renormalization method suggested by Regge and Ponzano and
described in this section. A counter-example in which the limit of
the renormalized partition function via this method does not converge
is given in \cite[Sec. 2.4]{barrett2009ponzano}. For an extensive
discussion of different methods of renormalization of the Regge-Ponzano
partition function the reader is referred to \cite{barrett2009ponzano}.%
} it via a method described below to obtain $\left[D\right]$. This
case is interesting and deserves our attention since it gives a result
which suggests a formal analogy with the Feynman path integral formalism
in connection with the theory of relativity via Regge calculus. Before
discussing briefly this result, which is the most important one of
this section, we conclude the description of the calculation of $\left[D\right]$
in the case where there are no internal vertices. 

If we replace $\left[T_{k}\right]$ with the asymptotic formula \prettyref{eq:Asymptotic-formula-6j-symbols}
and express the cosine in terms of Euler functions, then $A(x_{1},\dots,x_{q})$
will be the sum of $2^{p}$ pairwise conjugate terms. We express also
the factor $(-1)^{\chi}$ in terms of the exponential function%
\footnote{Even if this procedure is only correct for integer $x_{j}$, it is
possible to extended it to half-integers as well.%
} $e^{\pm i\pi\chi}$ such that $A$ takes following form when we denote
by $x_{ik}$ the edge $x_{i}$ belonging to the tetrahedron $T_{k}$:
\begin{gather*}
A(x_{1},\dots,x_{q})=\frac{e^{i\pi\chi_{0}}}{2^{p}(12\pi)^{p/2}(V_{1}\dots V_{p})^{1/2}}\left(\prod_{j=1}^{q}(2x_{j}+1)e^{\pm i\pi(n_{j}-2)x_{j}}\right)\times\\
\times\prod_{k=1}^{p}\left(e^{i\left[\sum x_{ik}\theta_{ik}+\pi/4\right]}+e^{-i\left[\sum x_{ik}\theta_{ik}+\pi/4\right]}\right)
\end{gather*}
which can be rearrange as a sum of terms%
\footnote{The summation in the exponential is over all tetrahedra $T_{k}$ with
the common internal edge $x_{j}$.%
} of the form proportional to
\[
\prod_{j=1}^{q}(2x_{j}+1)exp\left\{ i\left[\sum_{k=1}^{p_{j}}(\pm\theta_{kj}-\pi)+2\pi\right]x_{j}\right\} .
\]

Then $S$ will be the sum over the internal edges and over the $2^{p}$
terms similar to the above one. Hence we have
\[
S=\sum_{2^{p}\, terms}\sum_{x_{1}}\dots\sum_{x_{q}}\frac{C}{(V_{1}\dots V_{p})^{1/2}}\prod_{j=1}^{q}(2x_{j}+1)exp\left\{ iF(\theta_{kj})x_{j}\right\} ,
\]
 where $F(\theta_{kj})$ is a linear function of all dihedral angles
of the tetrahedra in $\mathfrak{D}(D)$. If we replace the summation
over all internal variables by integrals, one should expect that the
most important contributions arise from points where the phase is
stationary with respect to the $x_{j}$'s, i.e. $F(\theta_{kj})\overset{!}{=}0$,
thus we have
\[
\sum_{k=1}^{p_{j}}(\pi\mp\theta_{kj})\overset{!}{=}2\pi,
\]
 meaning that the sum of internal angles around $x_{j}$ is $2\pi$,
which implies the existence of a configuration embedded in a three-dimensional
Euclidean space. Since each solution%
\footnote{Since we assumed that there are no internal vertices, the internal
edges in $\mathfrak{D}(D)$ connect external vertices of $D$ and
they are sufficient to specify completely the configuration.%
} of these stationary conditions represent a specific configuration.
we obtain as final result the sum of contributions from each configuration.
Notice that we assumed already at the beginning of the discussion
for the calculation of $S$ that the diagram $D$ was embeddable in
the 2-sphere. We did this by fixing the form of the Euler characteristic,
cf. footnote \prettyref{fn:fixing-chi}, thus it is not surprising
that the solutions of the stationary conditions give configurations
embeddable in a 3-dim Euclidean space; $\chi$ contains information
about the combinatorial manifold. The main idea here is that through
solutions of these stationary conditions for situations with general
Euler characteristics we obtain configurations of $D$ corresponding,
in the limit, to differentiable manifolds. 

Finally, consider the sum $S$ in a simple case where there are internal
vertices. In this case the sum is infinite, \cite{ponzano1969semiclassical,moussouris1983quantum}.
However the form of the factor responsible for this infinity makes
it easy to renormalize the sum such that for general $\left[D\right]$
we have
\[
\left[D\right]=\lim_{R\rightarrow\infty}\frac{1}{\mathfrak{R}(R)^{P}}\sum_{x_{1}<R}\dots\sum_{x_{q}<R}A(x_{1},\dots,x_{q})
\]
 where $P$ is the number of internal vertices and $\mathfrak{R}(R)\simeq\frac{1}{\pi}\left(\frac{4\pi R^{3}}{3}\right)$
is the factor in $S$ responsible for the infinity when $R\rightarrow\infty$. 

Furthermore, assume that the number of vertices and edges in $D$
and $\mathfrak{D}(D)$ is very high such that the simplicial complex
approaches a differential manifold $M$ with boundary $D$. According
to Regge the sum $\sum_{j=1}^{q}\left[\sum_{k=1}^{p_{j}}(\pi\mp\theta_{kj})\right]x_{j}$
which, as stated previously, gives a possible configuration of the
embedding of $D$ in some manifold, approaches $\mathcal{S}(M)=\frac{1}{16\pi}\int_{M}RdV$
where $R$ is the scalar curvature of $M$. Therefore the positive
frequency part of $S$ has the following form
\begin{equation}
S^{+}=\frac{1}{\mathfrak{R}^{P}}\int_{D\, fixed}e^{i\mathcal{S}(M)}d\mu(M),\label{eq:path-integral-over-geometries}
\end{equation}
 where the summations over the internal variables $x_{j}$ were interpreted
as an integration over all manifolds with a given fixed boundary $\partial M=D$.
The above integral is not defined in any precise mathematical sense
since its derivation was heuristic in nature and its merely function
is to show the strong resemblance with a Feynman path integral with
the same Lagrangian density $\mathcal{L}$ as in a 3-dimensional Einstein
theory. 

The discussion above is the link between the concepts exposed in the
previous sections and it gives us an idea of the possible, not yet
fully understood, relation between the abstract idea of space and
basic quantum mechanical objects such as quantum angular momenta.
This suggests the possibility of describing the notion of space in
a physical manner, where a richer structure than the one given by
the theory of general relativity arises. In \prettyref{cha:Invariants-of-3-Manifolds}
we will discuss the more general concepts formalizing the discussion
above about decompositions of 3-manifolds and the invariants defined
out of their combinatorial properties, for instance, in the above
case the Feynman path integral over all differential 3-manifolds with
a given boundary.

\[
\]

\chapter{Mathematical Framework\label{cha:Mathematical-Framework}}

The mathematical framework encoded in the diagrammatic language of
spin networks is related to the category of representations of the
(deformed) quantum enveloping Hopf algebra of the $sl{}_{2}$ Lie
algebra, denoted $U{}_{q}(sl_{2})$. In this section the mathematical
concepts are introduced briefly. 

First, the concept of Hopf algebras (which are bialgebras with an
antipode obeying some identities) is described as in \cite{majid2000foundations}.
This structure is a generalization of the concept of a group with
a linear map, the antipode, sending $g$ to its inverse $g{}^{-1}$.
The quantum enveloping algebra $U(\mathfrak{g})$ is then a non-commutative
Hopf algebra generated by $1$ and the generators of the Lie algebra
$\mathfrak{g}$. This algebra can then be ``deformed'' by a non-zero
parameter $q$ which enters the commutation relations of the generators
of $U(\mathfrak{g})$ in such manner that if $q\rightarrow1$, the
deformed enveloping algebra reduces to the bialgebra $U(\mathfrak{g})$
. 

Second, the main ideas of category theory are introduced in order
to understand the relation between spin networks and the category
of representations of $U{}_{q}(sl_{2})$, which is a monoidal category
$\mathcal{C}$ with some extra structure. The monoidal structure of
the category is given by the tensor product of representations and
the morphisms (which are intertwiners between representations) can
be expressed in a graphical way. The extra structure is given by
the dual-functor $\star:\,\mathcal{C}^{op}\rightarrow\mathcal{C}$,
where $\mathcal{C}^{op}$ is the canonical dual of the category $\mathcal{C}$.
This structure makes the category into a monoidal category with duals
and allows it to have a pivotal structure, which is a morphism $\epsilon_{A}:\, e\rightarrow A\otimes A^{\star}$
with some axioms and compatibility with the counit and the tensor
product. The pivotal property is equivalent to the requirement of
planar isotopy of the diagrams in the framework of spin networks,
cf. \cite{barrett1996invariants,barrett1999spherical} and \cite{major1999spin}.
If a pivotal category is invariant under diffeomorphisms of $S{}^{2}$
it is called spherical. Ultimately, it is this type of categories
which give the general framework for spin networks and the category
of representations of $U{}_{q}(sl_{2})$ is a special case which give
rise to the Turaev-Viro invariant described in \prettyref{cha:Invariants-of-3-Manifolds}.

Third, the diagrammatic language is given by the Temperley-Lieb recoupling
theory which is a powerful tool to evaluate networks and gives the
identification of the tetrahedron with the recoupling coefficients
as well as basic identities between (quantum) $6j$-symbols such as
the Biedenharn-Elliott and orthogonality relations. In this context,
we also give the formal definition of many of the concepts already
seen in \prettyref{sec:Spin-Networks}, for instance, the projectors
or $n$-edges and its loop-value.

To finalize this chapter, we give a very brief and general account
of the topological quantum field theory which relates, in our context,
concepts of combinatorial manifolds with the algebraical data needed
to construct the invariants of 3-manifolds. We will encounter this
theory explicitly in \prettyref{cha:Invariants-of-3-Manifolds} where
the state sum invariant of $U{}_{q}(sl_{2})$ is described.

\section{Hopf algebras and Quantum Groups\label{sec:Hopf-alg-and-QGrps}}

\subsection{Algebras, Coalgebras and Bialgebras}

An algebra $A(\cdot,+;k)$ over a field $k$ is a ring $(A,\cdot,+)$
which is also a vector space $(A,+;k)$ and the action of the field
$k$ on $A$ is compatible with the product and addition. The associativity
of the product and the existence of a unit element (if required) are
expressed as commutative diagrams%
\footnote{For the definition cf. Sec.\ref{sub:Category-Theory}.%
}. In this language, an algebra is a vector space $A$ with a product
(and a unit) such that these diagrams commute. Here we will always
assume the existence of a unit. If $A,B$ are algebras, then $A\otimes B$
is an algebra which has a vector space given by the tensor product
of the vector spaces $A,B$ and a product defined by $(a\otimes b)(c\otimes d)\equiv(ac\otimes bd)$.

The dual notion of an algebra is a coalgebra $(C,+,\Delta,\epsilon;k)$
over a field $k$, which is a vector space $(C,+;k)$ and a linear
map $\vartriangle:\, C\rightarrow C\otimes C$, called the coproduct,
which is coassociative and for which there exist a linear map $\epsilon:\, C\rightarrow k$,
called the counit. The commutative diagrams for coassociativity (on
the left) and counit (the two on the right) are the following

\begin{center}
\includegraphics[scale=0.4]{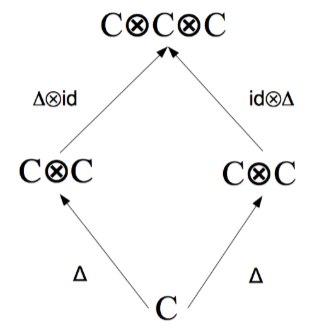} \includegraphics[scale=0.5]{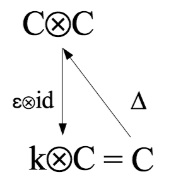}\includegraphics[scale=0.6]{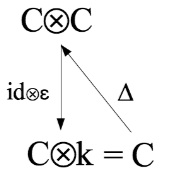}
\par\end{center}

Hence, since these diagrams commute, coassociativity means $\left(\Delta\otimes id\right)\circ\Delta=(id\otimes\Delta)\circ\Delta$
and the existence of a counit $\epsilon$ means $(\epsilon\otimes id)\circ\Delta=id=(id\otimes\epsilon)\circ\Delta$.
If we denote, for any $c\in C$,

\begin{center}
$\Delta(c)=\underset{i}{\sum}c_{i(1)}\otimes c_{i(2)}\in C\otimes C$
\par\end{center}

the coassociativity takes the explicit form

\begin{center}
$\underset{i}{\sum}c_{i(1)}\otimes c_{i(2)(1)}\otimes c_{i(2)(2)}=\underset{i}{\sum}c_{i(1)(1)}\otimes c_{i(1)(2)}\otimes c_{i(2)}$ 
\par\end{center}

and the counit means explicitly 

\begin{center}
$\underset{i}{\sum}\epsilon(c_{i(1)})\cdot c_{i(2)}=c$ 
\par\end{center}

by the isomorphism $k\otimes C\cong C$. From now on we will drop
the sum symbol and index and think of the sum implicitly in the notation
$\Delta(c)=c_{(1)}\otimes c_{(2)}$.

If a vector space $(H,+;k)$ is both an algebra and a coalgebra and
both structures are compatible, i.e. for all $g,\, h\in H$ we have

\begin{center}
$\Delta(hg)=\Delta(h)\Delta(g)$, $\Delta(1_{H})=1_{H}\otimes1_{H}$,
$\epsilon(hg)=\epsilon(h)\epsilon(g)$, $\epsilon(1_{H})=1_{k}$,
\par\end{center}

then $H$ is called a \textbf{bialgebra}. The compatibility of the
algebra and coalgebra structures is due to the fact that $\Delta,\,\epsilon$
are algebra maps and $\cdot:\, H\otimes H\rightarrow H,\text{ and}\,\eta:\, k\rightarrow H$
are coalgebra maps, i.e. respect the (co-)algebra structure. The map
$\cdot:\, H\otimes H\rightarrow H$ is the algebra multiplication
and the map $\eta:\, k\rightarrow H$ is a map used to express diagrammatically
the existence of a unit element in $H$ (cf. Sec. \prettyref{sub:Category-Theory}).

\subsection{Hopf algebras and some properties \label{sub:Hopf-algebras}}
\begin{defn}
A \textbf{Hopf algebra} $H$ is a bialgebra with a linear map $S:\, H\rightarrow H$
called \textbf{antipode}, which obeys $\cdot(S\otimes id)\circ\Delta=\cdot(id\otimes S)\circ\Delta=\eta\circ\epsilon$.
The axiom for the antipode is expressed as a commutative diagram in
the following way
\end{defn}
\begin{center}
\includegraphics[scale=0.5]{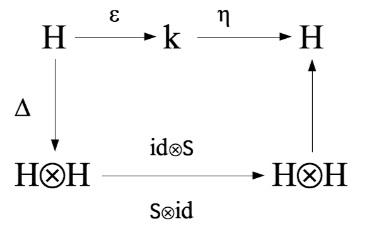}
\par\end{center}
\begin{rem*}
Although the antipode is the generalization of the concept of inverse
in a group, it is not required that $S{}^{2}=id$ or the existence
of $S^{-1}$, however, in the finite dimensional case $S^{-1}$ always
exists. Explicitly, we have from the counit and antipode axioms the
relations 
\[
h=h_{(1)}\epsilon(h_{(2)})=\epsilon(h_{(1)})h_{(2)}
\]
 and 
\[
h_{(1)}S(h_{(2)})=\epsilon(h)\cdot1_{H}=S(h_{(1)})h_{(2)};\ S(h)=\epsilon(h_{(1)})S(h_{(2)}),
\]
for any $h\in H$.
\end{rem*}
The uniqueness of the antipode $S$ of a given Hopf algebra follows
from the counit and antipode axioms and from the linearity of $S$.
The antipode is an anti-algebra and anti-coalgebra map, i.e. it transposes
the order of the factors; this means that, for any $g,\, h\in H$,
$S$ obeys $S(gh)=S(h)S(g)$, $S(1)=1$ and $(S\otimes S)\circ\Delta h=\tau\circ\Delta\circ Sh$,
$\epsilon Sh=\epsilon h$. 
\begin{example}
\label{exa:Weyl-algebra.}The \textbf{Weyl algebra}. Consider the
generators $1,\, X,\, g,\, g^{-1}$ of an algebra with relations $gg^{-1}=1=g^{-1}g,\; g^{-1}Xg=qX$,
where $q\in k$ is fixed and invertible. These become a Hopf algebra
with following relations
\[
\Delta X=X\otimes1+g\otimes X;\;\Delta g=g\otimes g,\;\Delta g^{-1}=g^{-1}\otimes g^{-1}
\]
\[
\epsilon(X)=0,\; S(X)=-g^{-\text{1}}X;\;\epsilon(g)=1=\epsilon(g^{-1}),\; S(g)=g^{-1},\; S(g^{-1})=g
\]
\end{example}
\begin{rem*}
\textit{i)} The defining operations are S and $\Delta$, since with
$\epsilon(h)=S(h_{(1)})h_{(2)}$, the counit follows directly from
those operators, e.g. $\epsilon(X)=S(X)1+S(g)X=-g^{-\text{1}}X+g^{-\text{1}}X=0$.
\textit{ii) }The comultiplication $\Delta$ and the antipode extend
as algebra and anti-algebra maps respectively, since $\Delta Xg=Xg\otimes g+g^{2}\otimes Xg=\Delta X\Delta g$
and they respect the defining relations, e.g. $\Delta Xg=q\Delta gX$.\end{rem*}
\begin{example}
\label{exa:universal-enveloping}The \textbf{universal enveloping
Hopf algebra} of $\mathfrak{g}$, denoted $U(\mathfrak{g})$:
\end{example}
Let $\mathfrak{g}$ be a Lie algebra over a field $k$, then the non-commutative
algebra $U(\mathfrak{g})$ generated by 1 and the generators of $\mathfrak{g}$
is a non-commutative Hopf algebra with coproduct, counit and antipode
given by $\Delta\xi=\xi\otimes1+1\otimes\xi,\:\epsilon(\xi)=0,\: S(\xi)=-\xi$
for any $\xi\in\mathfrak{g}$ extended as (anti-)algebra maps consistent
with $\Delta\left[\xi,\eta\right]=\left[\xi,\eta\right]\otimes1+1\otimes\left[\xi,\eta\right],\:\epsilon(\left[\xi,\eta\right])=0$
and $S(\left[\xi,\eta\right])=-\left[\xi,\eta\right]$.
\begin{rem*}
Since the notion of inverse in group theory is generalized as antipode
in the theory of Hopf algebras, the category of finite groups extends
to that of Hopf algebras. Similarly there is an association from the
category of Lie algebras to that of universal enveloping Hopf algebras
given by a forgetful functor, cf. Section \prettyref{sub:Category-Theory}. 
\end{rem*}
The action of a Hopf algebra $H$ on a vector space $V$ is defined
in the usual way. An algebra $A$ is an $H$-module algebra if $A$
is a left $H$-module and the action of $H$ for any $h\in H,\: a,\, b\in A$
is as follows

\[
h\vartriangleright(ab)=\sum(h_{(1)}\vartriangleright a)(h_{(2)}\vartriangleright b),\; h\vartriangleright1=\epsilon(h)\cdot1
\]

A coalgebra $C$ is a left $H$-module coalgebra if $C$ is an $H$-module
and for any $h\in H,\: c\in C$

\[
\Delta(h\vartriangleright c)=\sum h_{(1)}\vartriangleright c_{(1)}\otimes h_{(2)}\vartriangleright c_{(2)},\;\epsilon(h\vartriangleright c)=\epsilon(h)\epsilon(c)
\]

This means that the action $\vartriangleright:\, H\otimes C\rightarrow C$
is a coalgebra map, i.e. $\Delta(h\vartriangleright c)=(\Delta h)\vartriangleright(\Delta c)$.
\begin{example}
\label{exa:Tensor-product-of-modules}The adjoint action for $H=U(\mathfrak{g})$
is for all $\xi,\eta\in\mathfrak{g},\;\xi\vartriangleright\eta=\left[\xi,\,\eta\right]$.
The action of the universal enveloping Hopf algebra of $\mathfrak{g}$
on a $U(\mathfrak{g})$-module algebra $A$ is the usual notion (as
for Lie algebras), $\xi\vartriangleright\left(ab\right)=(\xi\vartriangleright a)b+a(\xi\vartriangleright b),\:\xi\vartriangleright1=0$
for any $\xi\in\mathfrak{g}$. 
\end{example}
If a Hopf algebra H acts on vector spaces V and W, then it also acts
on the tensor product space $V\otimes W$ by 
\[
h\vartriangleright(v\otimes w)=\sum h_{(1)}\vartriangleright v\otimes h_{(2)}\vartriangleright w,\;\,\forall h\in H,\, v\in V,\, w\in W
\]

Hence, the two actions have a tensor product and the tensor product
of $H$-modules is also an $H$-module. If a map $\phi:\, V\rightarrow W$
commutes with the corresponding actions, i.e. $\phi(h\vartriangleright v)=h\vartriangleright\phi(v)$
for all $v\in V$, then $\phi$ is called an intertwiner. This type
of maps constitute the morphisms in the category of Hopf algebra representations,
which are closely related to the theory of spin networks.

The dual notion of commutativity is cocommutativity, i.e. $\tau\circ\Delta=\Delta$,
where $\tau:\, H\otimes H\rightarrow H\otimes H$ is the transposition
map and can be weakened with the help of an element $\mathcal{R}\in H\otimes H$
called the quasitriangular structure. Then, the Hopf algebra H is
cocommutative only up to conjugation by $\mathcal{R}$.
\begin{defn}
A \textbf{quasitriangular bialgebra }is a pair $\left(H,\mathcal{R}\right)$
where $H$ is a bialgebra and $\mathcal{R}\in H\otimes H$ is invertible
and obeys
\begin{equation}
(\Delta\otimes id)\mathcal{R=R}_{13}\mathcal{R}_{23}\,,\;(id\otimes\Delta)\mathcal{R}=\mathcal{R}_{13}\mathcal{R}_{12}\label{eq:QTstructure}
\end{equation}
\begin{equation}
\tau\circ\Delta h=\mathcal{R}(\Delta h)\mathcal{R}^{-1},\;\forall h\in H
\end{equation}

With the notation $\mathcal{R}=\sum\mathcal{R}^{(1)}\otimes\mathcal{R}^{(2)}$,
$\mathcal{R}_{ij}\in H\otimes H\otimes...\otimes H$ is defined as
\[
\mathcal{R}_{ij}=\sum1\otimes\ldots\otimes\underset{i-th}{\mathcal{R}^{(1)}}\otimes\ldots\otimes\underset{j-th}{\mathcal{R}^{(2)}}\otimes\ldots\otimes1
\]

\end{defn}
To understand Eq. \eqref{eq:QTstructure} let us write the first equation
explicitly. On the l.h.s. we have
\[
(\Delta\otimes id)\mathcal{R}=\sum_{i}(\Delta\mathcal{R}_{i}^{(1)})\otimes\mathcal{R}_{i}^{(2)}=\sum_{i}(\sum_{j}\mathcal{R}_{(1)i,j}^{(1)}\otimes\mathcal{R}_{(2)i,j}^{(1)})\otimes\mathcal{R}_{i}^{(2)}
\]
and with $\mathcal{R}_{13}=\sum_{m}\mathcal{R}_{m}^{(1)}\otimes1\otimes\mathcal{R}_{m}^{(2)}$
and $\mathcal{R}_{23}=1\otimes\mathcal{R}=\sum_{k}1\otimes\mathcal{\tilde{R}}_{k}^{(1)}\otimes\mathcal{\tilde{R}}_{k}^{(2)}$
we have on the r.h.s.
\[
\mathcal{R}_{13}\mathcal{R}_{23}=(\sum_{m}\mathcal{R}_{m}^{(1)}\otimes1\otimes\mathcal{R}_{m}^{(2)})(\sum_{k}1\otimes\tilde{\mathcal{R}}_{k}^{(1)}\otimes\tilde{\mathcal{R}}_{k}^{(2)})=\sum_{m}\sum_{k}\mathcal{R}_{m}^{(1)}\otimes\tilde{\mathcal{R}}_{k}^{(1)}\otimes\mathcal{R}_{m}^{(2)}\tilde{\mathcal{R}}_{k}^{(2)}
\]

Comparing the r.h.s. and the l.h.s. we get
\[
\mathcal{R}_{(1)i,j}^{(1)}=\mathcal{R}_{m}^{(1)},\;\mathcal{R}_{(2)i,j}^{(1)}=\tilde{\mathcal{R}}_{k}^{(1)},\;\mathcal{R}_{i}^{(2)}=\mathcal{R}_{m}^{(2)}\tilde{\mathcal{R}}_{k}^{(2)}\;\, for\, some\, i,\, j,\, m,\, k.
\]

This means that the factors of the coproduct of $\mathcal{R}$ are
the factors of itself -in the first and second position- or a product
of its factors, in the third position. 

The quasitriangular structure $\mathcal{R}$  in a bialgebra $H$
obeys $\cdot(\epsilon\otimes id)\mathcal{R}=\cdot(id\otimes\epsilon)\mathcal{R}=1_{H}$.
If $H$ is a Hopf algebra then, it also obeys $(S\otimes id)\mathcal{R=R}^{-1}$
and $(id\otimes S)\mathcal{R}^{-1}=\mathcal{R}$. Hence it also satisfies
$(S\otimes S)\mathcal{R=R}$.

The defining characteristics of a quasitriangular bialgebra imply
the abstract quantum Yang-Baxter-Equation, 
\begin{equation}
\mathcal{R}_{12}\mathcal{R}_{13}\mathcal{R}_{23}=\mathcal{R}_{23}\mathcal{R}_{13}\mathcal{R}_{12}.\label{eq:Yang-Baxter-eq}
\end{equation}

Also, if $\left(H,\,\mathcal{R}\right)$ is a quasitriangular Hopf
algebra, the antipode $S$ is automatically invertible and satisfies
$S^{2}(h)=uhu^{\text{-1}},\,\forall h\in H$, where $u$ is an invertible
element in $H$ obeying
\[
u=\sum(S\mathcal{R}^{(2)})\mathcal{R}^{(1)},\; u^{-1}=\sum\mathcal{R}^{(2)}S^{2}\mathcal{R}^{(1)},\;\Delta u=Q^{-1}(u\otimes u)
\]
with $Q=\mathcal{R}_{21}\mathcal{R}$. Similarly, the antipode of
the previous element $v=Su$ satisfies $S^{-2}(h)=vhv^{\text{-1}},\,\forall h\in H$
and obeys 
\[
v=\sum\mathcal{R}^{(1)}S\mathcal{R}^{(2)},\; v^{-1}=\sum(S^{2}\mathcal{R}^{(1)})\mathcal{R}^{(2)},\;\Delta v=Q^{-1}(v\otimes v).
\]

Having $u,\, v\in H$, the element $uv=vu$ is central, i.e. commutes
with all elements%
\footnote{This follows from the fact that $\forall h\in H:\: S^{2}(h)=uhu^{-1}\text{ and }S^{-2}(h)=vhv^{-1}$,
thus, $h=(vu)h(vu)^{-1}$.%
} of $H$, and obeys the relation $\Delta(uv)=Q^{-2}(uv\otimes uv)$.
On the other hand, the element $w=uv^{-1}=v^{-1}u$ is group-like
and implements $S^{4}$ by conjugation%
\footnote{From $h=v^{-1}S^{-2}(h)v$, we have $S^{2}(S^{-2}(h))=v^{-1}S^{-2}(h)v$,
thus, $v^{-1}$ implements $S^{2}$ as well.%
}. If the element $uv$ has a central square root called the ribbon
element $\nu$, i.e. $\nu^{2}=uv$, satisfying $\Delta\nu=Q^{-1}(\nu\otimes\nu),\;\epsilon(\nu)=1,\; S(\nu)=\nu$,
then the quasitriangular Hopf algebra is called a \textbf{ribbon Hopf
algebra} (cf. Sec. \ref{sub:Spherical-categories}). 
\begin{example}
\label{exa:anyon-generating}The \textbf{anyon-generating quantum
group }is generated by $1,\, g,\, g^{-1}$ and the relation $g^{n}=1$
with the Hopf algebra structure $\Delta g=g\otimes g,\,\epsilon(g)=1,\, S(g)=g^{-1}$.
It has the following non-trivial quasitriangular structure
\[
\mathcal{R=}\frac{1}{n}\sum_{a,b=0}^{n-1}e^{-\frac{2\pi iab}{n}}g^{a}\otimes g^{b}
\]

The r.h.s. of the first equation in \prettyref{eq:QTstructure} is
\[
\mathcal{R}_{13}\mathcal{R}_{23}=\frac{1}{n^{2}}\sum_{a,b,c,d}e^{-\frac{2\pi i(ab+cd)}{n}}g^{a}\otimes g^{c}\otimes g^{b+d}
\]
which can be simplify with $b\text{\textasciiacute}=b+d$ and $n^{-1}\sum_{b=0}^{n-1}e^{-\frac{2\pi iab}{n}}=\delta_{a,0}$.
\end{example}
Since the notion of an algebra is dual to the one of a coalgebra,
the axioms of a Hopf algebra are self dual when interchanging $\Delta,\,\epsilon$
and $\cdot,\,\eta$. In terms of dual linear spaces, this symmetry
gives for every finite dimensional Hopf algebra $H$ a dual Hopf algebra
$H^{*}$ build on the vector space dual to $H$ (cf. Sec. 1.4 in
\cite{majid2000foundations}).

The axioms of a quasitriangular Hopf algebra are \textit{not self-dual
}but it is possible to extend the concept of a dual quasitriangular
structure as an ``invertible'' map $A\otimes A\rightarrow k$, where
$A$ is the Hopf algebra in the dual formulation of $H$; the invertibility
being defined in a suitable way (cf. Sec. 2.2 in \cite{majid2000foundations}).

The above description of Hopf algebras and quasitriangular structure
is intended to present briefly the quantum group $U{}_{q}(sl_{2})$
and it is not the intention to give a detailed description of any
of the concepts of this broad topic since this would go beyond the
scope of this dissertation.

\subsection{\label{sub:The-Quantum-Group}The Quantum Group $U_{q}(\mathfrak{sl}_{2})$}

Recall that the smallest simple Lie algebra is $\mathfrak{sl}_{2}$
with the following relations for the operators $H$ and $X_{+},\; X_{-}$
\[
\left[H,\, X_{\pm}\right]=\pm2X_{\pm},\;\left[X_{+},\, X_{-}\right]=H
\]

The deformation of the universal enveloping algebra of $\mathfrak{sl}_{2}$
(cf. Example \prettyref{exa:universal-enveloping}) with a parameter
$q\neq0$ defines the quantum group $U_{q}(\mathfrak{sl}_{2})$, which
is a non-commutative Hopf algebra generated by $1,\, X_{+},\, X_{-},\, q^{H/2},\, q^{-H/2}$
and the relations 
\begin{equation}
q^{\pm H/2}q^{\mp H/2}=1,\; q^{H/2}X_{\pm}q^{-H/2}=q^{\pm1}X_{\pm},\;\left[X_{+},\, X_{-}\right]=\frac{q^{H}-q^{-H}}{q-q^{-1}}\label{eq:CommutRelU(sl)}
\end{equation}
Similar to example \prettyref{exa:Weyl-algebra.}, this forms a Hopf
algebra with comultiplication, counit and antipode maps as follows
\[
\Delta q^{\pm H/2}=q^{\pm H/2}\otimes q^{\pm H/2},\;\Delta X_{\pm}=X_{\pm}\otimes q^{H/2}+q^{-H/2}\otimes X_{\pm}
\]
\[
\epsilon q^{\pm H/2}=1,\;\epsilon X_{\pm}=0;\; SX_{\pm}=-q^{\pm\text{1}}X_{\pm},\; Sq^{\pm H/2}=q^{\mp H/2}
\]

The relations between the generators are defined such that if $q\rightarrow1$
\prettyref{eq:CommutRelU(sl)} reduce to the ones defining the Lie
algebra $\mathfrak{sl}_{2}$. For example using the l'Hôpital rule
for the third equation of \prettyref{eq:CommutRelU(sl)} we have
\[
\left[X_{+},\, X_{-}\right]=\lim_{q\rightarrow1}\frac{q^{H}-q^{-H}}{q-q^{-1}}=\lim_{q\rightarrow1}\frac{H(q^{H-1}+q^{-H-1})}{1+q^{-2}}=H.
\]

The parameter $q$ is usually an element of $\mathbb{C}$, however,
if no calculations are needed it is useful to work over the field
of formal power series of a parameter t, denoted $\mathbb{C}\left[t\right]$
and we define $q=e^{t/2}$. If some calculations are needed it is
not possible to work over $\mathbb{C}\left[t\right]$ since $q$ may
not have a finite value (the formal power series might not converge)
and the precise value of $q$ is relevant for some aspects of the
theory, e.g. the case when $q$ is a root of unity, cf. \prettyref{sec:Some-More-Diagrammatics}
and \prettyref{cha:Invariants-of-3-Manifolds}.

The above defined Hopf algebra has a quasitriangular structure
\[
\mathcal{R}=q^{\frac{H\otimes H}{2}}\exp_{q^{-2}}\left\{ (1-q^{-2})q^{H/2}X_{+}\otimes q^{-H/2}X_{-}\right\} 
\]
where the $q$-exponential is defined as 
\begin{equation}
e_{q^{-2}}^{x}=\sum_{n=0}^{\infty}\frac{x^{n}}{\left[n;\, q^{-2}\right]!}\;;\;\left[n;\, q^{-2}\right]=\frac{1-q^{-2n}}{1-q^{-2}}\,,\;\left[n;\, q^{-2}\right]!=\left[n;\, q^{-2}\right]\dots\left[1;\, q^{-2}\right]\label{eq:q-exp-and-quantum-factorial}
\end{equation}
and follows similar properties as the usual exponential function with
operators (cf. \cite{majid2000foundations}, p. 86). 

The quantum group $U_{q}(\mathfrak{sl}_{2})$ acts on objects similar
to the Lie algebra $\mathfrak{sl}_{2}$-modules. In fact, for $q\text{\ensuremath{\in}}\mathbb{R}$
and each $j=0,\,\frac{1}{2},\,1,\dots$ the real form $U_{q}(\mathfrak{su}_{2})$
 has a $(2j+1)$-dimensional unitary irreducible representation $V_{j}=\left\{ |j,m>\,|\, m=-j,\dots,j\,\right\} $such
that
\begin{equation}
X_{\pm}|j,m>=\sqrt{\left[j\mp m\right]\left[j\pm m+1\right]}|j,m\pm1>;\; q^{H/2}|j,m>=q^{m}|j,m>\label{eq:reps-of-Uq(sl2)}
\end{equation}

If $q$ is a\textbf{ primitive root of unity}%
\footnote{Here, primitive means that $n>1$ and there is no smaller $n$ such
that $q^{n}=1$.%
} these formulas also apply, however, the number $j$ has to be replaced
by a ``quantum integer'' $\left[j\right]$ so the allowed range
of $j$ is restricted in a suitable way, cf. \prettyref{cha:Invariants-of-3-Manifolds}.
For now, we will work over $\mathbb{C}$ and denote the generators
by $K,\, K^{-1},\, X_{\pm}$ instead of $q^{H/2},\, q^{-H/2},\, X_{\pm}$.
Then the quotient of $U_{q}(\mathfrak{sl}_{2})$ denoted by $U_{q}^{(r)}(\mathfrak{sl}_{2})$
is the finite dimensional quasitriangular Hopf algebra generated by
$1,\, K,\, K^{-1},\, X_{\pm}$ with relations
\[
K^{\pm1}K^{\mp1}=1,\; KX_{\pm}K^{-1}=q^{\pm1}X_{\pm},\;\left[X_{+},\, X_{-}\right]=\frac{K^{2}-K^{-2}}{q-q^{-1}}\,;\underbrace{K^{4r}=1,\, X_{\pm}^{r}=0}_{relations\, which\, quotient\, U_{q}(\mathfrak{sl_{2})}}
\]

It has the operations inherited from $U_{q}(\mathfrak{sl}_{2})$ but
a different quasitriangular structure given by
\[
\mathcal{R=R}_{K}\sum_{m=0}^{r-1}(KX_{+})^{m}\otimes(K^{-1}X_{-})^{m}\frac{(1-q^{-2})^{m}}{\left[m;\, q^{-2}\right]!}\,;\;\mathcal{R}_{K}=\frac{1}{4r}\sum_{a,b=0}^{4r-1}q^{-ab/2}K^{a}\otimes K^{b}.
\]

The representations $V_{j}$ of $U_{q}^{(r)}(\mathfrak{sl}_{2})$
are the same as in \prettyref{eq:reps-of-Uq(sl2)}, however only for
spins in the range $j=0,\,\frac{1}{2},\,1,\dots,\,\frac{r-1}{2}$.
To see that, notice that if $q$ is a primitive root of unity with
$q^{r}=1$ the quantum integer is $\left[n\right]=\frac{q^{n}-q^{-n}}{q-q^{-1}}\neq0$
for $n=1,\dots,r-1$ but $\left[r\right]=0$, so the action of $X_{\pm}$
is well-defined if $j$ is in the allowed range $0\leq j-m\leq r-1;\, m=-j,\dots,j$,
cf. Section \prettyref{sub:Temperley-Lieb-Algebra}.

The above case is important since it allows the regularization of
the Ponzano-Regge partition function, which otherwise diverges for
some manifolds. The quantum group $U{}_{q}(sl_{2})$ at a primitive
root of unity causes the state sum to become finite, thus, it is possible
to define an invariant for general 3-manifolds, as described in \prettyref{cha:Invariants-of-3-Manifolds}.

\section{Category Theory\label{sub:Category-Theory}}

In this section, the main concepts of category theory will be introduced
briefly following \cite{MacLane1971} closely. This theory is the
general framework needed to understand the mathematics behind spin
networks and its relation with the representations of the quantum
group $U_{q}(\mathfrak{sl}_{2})$. Category theory, roughly speaking,
studies in an abstract way the properties of specific mathematical
concepts by gathering them in collections of objects and arrows (or
morphisms) which satisfy certain fundamental conditions. In other
words, it is the theory for dealing, in the most general and abstract
way, with concepts like sets, topological spaces, vector spaces, groups,
etc.

\subsection{Basic concepts of category theory}

Let us start by introducing one of the basic concepts of category
theory, the \textbf{commutative diagram}. Consider a diagram 

\begin{center}
\includegraphics[scale=0.5]{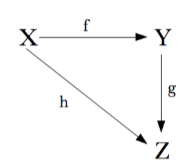}
\par\end{center}

where $f:\text{\,}x\mapsto fx$, $h:\, x\mapsto hx$ and $g:\, y\mapsto gy$
are morphisms. If $h=g\circ f$ the diagram is called commutative.
\begin{defn}
A \textbf{category }$\mathcal{C}$ is a collection of objects $a,b,c,...$
in $\mathcal{C}$ together with a collection $Hom_{\mathcal{C}}(a,\, b)$
of morphisms for any pair of objects $a,b$ in $\mathcal{C}$ with
a composition rule, such that for each $a,b,c$ in $\mathcal{C}$
and each pair of arrows 

\[
a\overset{f}{\rightarrow}b\;\text{and}\; b\overset{g}{\rightarrow c}
\]

there is a composite arrow $a\overset{g\circ f}{\longrightarrow}c$,
and for each object $a$ in $\mathcal{C}$ there is an identity arrow
$1_{a}:\, a\rightarrow a\,\in Hom_{\mathcal{C}}(a,\, a)$. The composition
rule and the identity arrow are subject to the following axioms:\end{defn}
\begin{itemize}
\item Associativity: given the arrows $f,\, g$ as above and $c\overset{h}{\rightarrow}d$,
the composition is always associative, i.e. $h\circ(g\circ f)=(h\circ g)\circ f$,
so the following diagram is commutative 
\end{itemize}
\begin{center}
\includegraphics[scale=0.4]{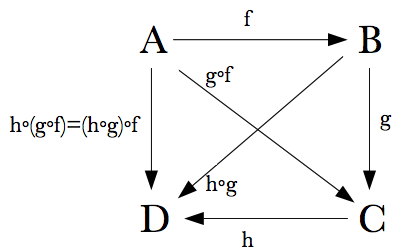}
\par\end{center}
\begin{itemize}
\item Unit law: for all arrows $a\overset{f}{\rightarrow}b\overset{g}{\rightarrow}c$
the composition with the unit arrow $1_{b}$ yields $1_{b}\circ f=f,\: g\circ1_{b}=g$,
i.e. the following diagram commutes
\end{itemize}
\begin{center}
\includegraphics[scale=0.35]{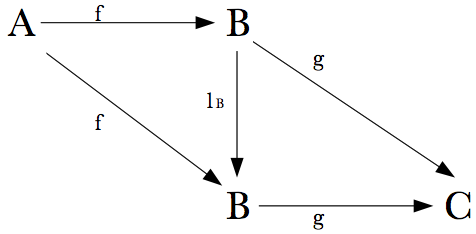}
\par\end{center}

The collections of morphisms are required to be pairwise disjunct,
i.e. $Hom_{\mathcal{C}}(a,\, b)$$\cap$$Hom_{\mathcal{C}}(c,\, d)=\emptyset$
for any objects $a,b,c,d$ with $a\neq c$ and $b\neq d$.
\begin{example*}
Typical examples of categories are:\end{example*}
\begin{itemize}
\item \textbf{Set}, with objects all sets and morphisms all functions between
them,
\item \textbf{Grp}, where the objects are groups and the arrows are group
homomorphisms,
\item \textbf{Top}, the category of all topological spaces and continuous
maps,
\item \textbf{Vect(k)}, with vector spaces over a field k and linear maps,
\item Let $A$ be a unital algebra and \textbf{$_{A}\mathcal{M}$ }the \textbf{category
of $A$-modules} with vector spaces, on which $A$ acts, as objects
and the morphisms are linear maps that commute with the action of
$A$, i.e. intertwines. \end{itemize}
\begin{defn}
A \textbf{functor} is a morphism between categories; more precisely,
for categories $\mathcal{C},\,\mathcal{B}$ a functor $T:\,\mathcal{C}\rightarrow\mathcal{B}$
consists of two related functions\end{defn}
\begin{enumerate}
\item \textbf{Object function} $T$: for any object $c$ in $\mathcal{C}$,
we have $c\mapsto Tc$, where $Tc$ is in $\mathcal{B}$.
\item \textbf{Arrow function} $T$: for any morphism $f:\, c\rightarrow c'$
of $\mathcal{C}$, we have $f\mapsto Tf$, where $Tf:\, Tc\rightarrow Tc'$
is an arrow of $\mathcal{B}$.
\end{enumerate}
such that $T(1_{c})=1_{Tc}$ and $T(g\circ f)=Tg\circ Tf$. 

A \textbf{covariant} functor $T:\,\mathcal{C}\rightarrow\mathcal{B}$
respects the structure of the categories, i.e. for any two morphisms
that can be composed $T$ sends $f\mapsto Tf$ and $g\mapsto Tg$
such that $T(g\circ f)=Tg\circ Tf$. On the other hand, a \textbf{contravariant}
functor $F:\,\mathcal{C}\rightarrow\mathcal{B}$ is defined as sending
$f:\, a\rightarrow b$ to $Ff:\, Fb\rightarrow Fa$ such that $F(g\circ f)=Ff\circ Fg$
for any two morphisms $f,g$ of $\mathcal{C}$ that can be composed,
\cite{majid2000foundations}.

An example of a functor is the forgetful functor, which, as the name
suggests, ``forgets'' some or all of the structure of an algebraic
object in a given category, e.g. the functor $U:\, Grp\rightarrow Set$
assigns to each group $G$ the set $UG$ of its elements and each
homomorphism the same function regarded as a function between sets.

Given some functors $\mathcal{C}\overset{T}{\rightarrow}\mathcal{B}\overset{S}{\rightarrow}\mathcal{A}$,
the composite functions on objects and arrows of $\mathcal{C}$ define
a functor, i.e. functors may be composed. Furthermore, the composition
is associative and for each category $\mathcal{C}$ there exists an
identity functor $I_{\mathcal{C}}:\,\mathcal{C}\rightarrow\mathcal{C}$
which acts as an identity for the composition of functors. Hence,
we can regard the collection of categories as a category itself where
the functors are the morphisms of this category.
\begin{defn*}
A functor $T:\,\mathcal{C}\rightarrow\mathcal{B}$ is an \textbf{isomorphism}
if and only if there is a functor $S:\,\mathcal{B}\rightarrow\mathcal{C}$
for which $S\circ T=I_{\mathcal{C}}$ and $T\circ S=I_{\mathcal{B}}$. 

A functor $T:\,\mathcal{C}\rightarrow\mathcal{B}$ is \textbf{faithful}
when for every pair $c,c'$ of objects in $\mathcal{C}$ and every
pair $f_{1},\, f_{2}:\, c\rightarrow c'$ of (parallel) arrows of
$\mathcal{C}$, $Tf_{1}=Tf_{2}$ implies $f_{1}=f_{2}$.

The functor $T:\,\mathcal{C}\rightarrow\mathcal{B}$ is called \textbf{full}
when for every pair $c,c'$ of objects in $\mathcal{C}$ and every
arrow $g:\, Tc\rightarrow Tc'$ of $\mathcal{B}$, there is an arrow
$f:\, c\rightarrow c'$ of $\mathcal{C}$ with $g=Tf$.\end{defn*}
\begin{defn}
A \textbf{natural transformation} is a morphism between functors.
Given two functors $S,\, T:\,\mathcal{C}\rightarrow\mathcal{B}$,
a natural transformation $\tau:\, S\rightarrow T$ is a function which
assigns to each object $c$ in $\mathcal{C}$ a morphism $\tau_{c}:\, Sc\rightarrow Tc$
of $\mathcal{B}$ such that every arrow $f:\, c\rightarrow c'$ of
$\mathcal{C}$ gives the following commutative diagram
\end{defn}
\begin{center}
\includegraphics[scale=0.5]{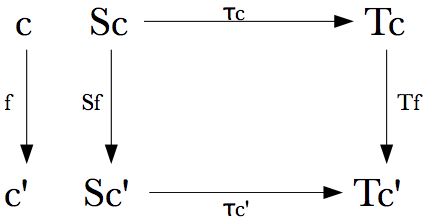}
\par\end{center}

In other words, if we regard $S$ as giving a picture of $\mathcal{C}$
in $\mathcal{B}$, then a natural transformation is a set of arrows
mapping the picture $S(\mathcal{C})$ in $\mathcal{B}$ to the picture
$T(\mathcal{C})$ in $\mathcal{B}$, i.e. the following diagram commutes

\begin{center}
\includegraphics[scale=0.5]{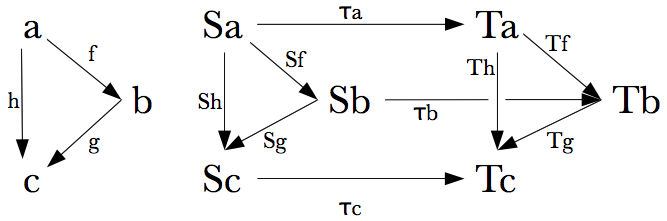}
\par\end{center}

The arrows $\tau_{a},\,\tau_{b},\,\tau_{c}$ of $\mathcal{B}$ are
called the components of $\tau$.

A \textbf{natural isomorphism} $\tau:\, S\cong T$ between functors
$S,\, T:\,\mathcal{C}\rightarrow\mathcal{B}$ is a natural transformation
$\tau$ such that every component $\tau_{c}$ is invertible in $\mathcal{B}$.
If this is the case, then $(\tau_{c})^{-1}$ are the components of
a natural isomorphism $\tau^{-1}:\, T\rightarrow S$. With the help
of this concept, one can define the equivalence between categories
$\mathcal{C}$ and $\mathcal{D}$. This is defined to be a pair of
functors $S:\,\mathcal{C}\rightarrow\mathcal{D}$, $T:\,\mathcal{D}\rightarrow\mathcal{C}$
with natural isomorphisms $I_{\mathcal{C}}\cong T\circ S$ and $I_{\mathcal{D}}\cong S\circ T$.
Notice that in this definition equality is \textbf{not} used since
$S$ and $T$ do not need to be isomorphisms.

There are some important cases where an object or arrows have specific
properties. For example, an arrow $e:\, a\rightarrow b$ is invertible
in $\mathcal{C}$ if there is an arrow $e':\, a\rightarrow b$ with
$e'e=1_{a}$ and $ee'=1_{b}$. If this arrow exists it is unique and
it is denoted $e'=e^{-1}$. In this case, the objects $a$ and $b$
are isomorphic in $\mathcal{C}$, $a\cong b$. If every arrow in a
category $\mathcal{G}$ is invertible, then $\mathcal{G}$ is called
a groupoid and each object $x$ in $\mathcal{G}$ determines a group
$Hom_{\mathcal{G}}(x,\, x)$. An arrow $f:\, x\rightarrow x'$ then
establishes a group isomorphism $Hom_{\mathcal{G}}(x,\, x)\cong Hom_{\mathcal{G}}(x',\, x')$
by conjugation, i.e. $\forall g\in Hom_{\mathcal{G}}(x,\, x):\: g\mapsto fgf^{-1}\in Hom_{\mathcal{G}}(x',\, x')$.

\subsection{Monoidal Categories}

The most important concept for us now is the one of monoidal categories,
which are basically a category supplied with a \textquotedbl{}product\textquotedbl{}
$\square$, e.g. the direct product $\times$, the direct sum $\oplus$
or the tensor product $\otimes$. There are to kinds of this category,
depending on the required strength of the associative law. A strict
monoidal category $\left\langle \mathcal{B},\,\square,\, e\right\rangle $
is a category $\mathcal{B}$ with a bifunctor $\square:\,\mathcal{B}\times\mathcal{B}\rightarrow\mathcal{B}$
which is \textbf{strictly associative}, this means that both functors
$\square\circ(\square\times1_{\mathcal{B}})$ and $\square\circ(1_{\mathcal{B}}\times\square)$
are \textbf{exactly equal}%
\footnote{Associativity of this bifunctor means that $\square$ is associative
for objects and for arrows.%
}, i.e. $(\mathcal{B}\times\mathcal{B})\times\mathcal{B}=\mathcal{B}\times(\mathcal{B}\times\mathcal{B})$.
This category also contains an object $e$ which is a left and right
\textbf{unit} for $\square$ such that $\square\circ(e\times1_{\mathcal{B}})=id_{\mathcal{B}}=\square\circ(1_{\mathcal{B}}\times e)$,
where $e\times1_{\mathcal{B}}:\,\mathcal{B}\rightarrow\mathcal{B}\times\mathcal{B}$
is the functor which sends an object $c$ in $\mathcal{B}$ to $\left\langle e,\, c\right\rangle $
\footnote{The unit law states for objects $e\square c=c=c\square e$ and for
arrows $1_{e}\square f=f=f\square1_{e}$.%
}. The product $\square$ assigns to each pair of objects $a,b$ in
$\mathcal{B}$ an object $a\square b$ in $\mathcal{B}$ and to each
pair of arrows $f:\, a\rightarrow a',\: g:\, b\rightarrow b'$ an
arrow $f\square g:\, a\square b\rightarrow a'\square b'$ such that
$1_{a}\square1_{b}=1_{a\square b}$ and $(f'\square g')\circ(f\square g)=(f'\circ f)\square(g'\circ g)$. 
\begin{defn}
A (relaxed) \textbf{monoidal category }$\left\langle \mathcal{B},\,\square,\, e,\,\alpha,\,\lambda,\,\rho\right\rangle $
is a category as the one described above, but with the associativity
law only up to a natural isomorphism $\alpha$ and the left and right
unit for $\square$ also up to natural isomorphisms $\lambda$ and
$\rho$ respectively, such that 
\[
\alpha_{a,b,c}:\: a\square(b\square c)\cong(a\square b)\square c\,;\:\lambda_{a}:\, e\square a\cong a\,,\,\rho_{a}:\, a\square e\cong a\;\forall a,b,c\:\text{in}\:\mathcal{B}.
\]
 \uline{and} the following diagrams commute for all $a,b,c,d$
in $\mathcal{B}$:
\end{defn}
\begin{center}
\includegraphics[scale=0.3]{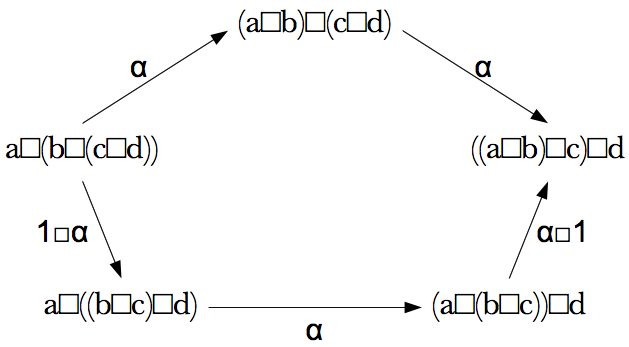} \includegraphics[scale=0.4]{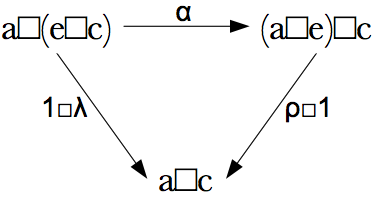}
\par\end{center}

The identification of multiple products of the category $\mathcal{B}$
is made in a coherent way with the help of the coherence theorem,
which states that all diagrams using the natural transformations $\alpha,\,\lambda,\,\rho$
commute (c.f. \cite{MacLane1971}, Sec. VII). This allows us to establish
an equivalence to a strict monoidal category and regard the associativity
law in a monoidal category as the usual one, but keeping always in
mind that we are dealing with equivalences rather than strict equalities.
We need this result since the category of representations of a Hopf
algebra is not strict, \cite{barrett1999spherical}. We can see this
by considering the more simple case of representations of the $SU(2)$
group. Recall that the coupling of three angular momenta $j_{1},\, j_{2},\, j_{3}$
depends on the way they are coupled, but the resulting vector spaces
are isomorphic to each other. We have then $(j_{1}\otimes j_{2})\otimes j_{3}\cong j_{1}\otimes(j_{2}\otimes j_{3})$.
\begin{defn}
A \textbf{category with duals} is a monoidal category $\left\langle \mathcal{C},\,\square,\, e,\,\alpha,\,\lambda,\,\rho\right\rangle $
with a functor $*:\,\mathcal{C}^{op}\rightarrow\mathcal{C}$ and the
following natural transformations\end{defn}
\begin{enumerate}
\item $\tau:\,1\rightarrow**$ i.e. $\phi^{**}\cong\phi$, where $\phi$
is an object or arrow of $\mathcal{C}$.
\item $\gamma:\,(*\times*)\circ\otimes\rightarrow*\circ\otimes^{op}$ i.e.
$\phi^{*}\otimes\psi^{*}\cong(\psi\otimes\phi)^{*}$ where $\phi,\,\psi$
are objects or arrows of $\mathcal{C}$.
\item $\nu:\, e\rightarrow e^{*}$, with $e^{*}\cong e$.
\end{enumerate}
where $\mathcal{C}^{op}$ is the category where the arrows of $\mathcal{C}$
are inverted and the objects are the same. For the natural transformations
in 1. and 2. is required each of their components to be isomorphisms.
\begin{rem*}
A category with \uline{strict} duals is a category in which $\tau,\,\gamma,$
and $\nu$ are the identity maps, \cite{barrett1999spherical}. 
\end{rem*}
Before further specification of the categories we are interested in,
we give the definition of a monoidal functor, which is a functor $F:\,\left\langle \mathcal{B},\,\square,\, e,\,\alpha,\,\lambda,\,\rho\right\rangle \rightarrow\left\langle \mathcal{V},\,\square',\, e',\,\alpha',\,\lambda',\,\rho'\right\rangle $
between monoidal categories which respects the monoidal product in
the sense that for all objects $a,b,c$ and arrows $f,g$ in $\mathcal{B}$
we have
\begin{eqnarray*}
F(a\square b)=Fa\square'Fb\,, & \: F(f\square g)=Ff\square'Fg\,, & Fe=e'\:;\\
F\alpha_{a,b,c}=\alpha'_{Fa,Fb,Fc}\,, & F\lambda_{a}=\lambda'_{Fa}\,, & F\rho_{a}=\rho'_{Fa}.
\end{eqnarray*}

\subsection{Spherical categories \label{sub:Spherical-categories}}

In this section we follow \cite{barrett1999spherical,barrett1996invariants}
to establish the relation between the graphs and their properties
described in \prettyref{sec:Spin-Networks} and a special type of
monoidal categories, the spherical categories, which are defined to
be a pivotal category satisfying the extra condition that the right
and left traces of the endomorphisms of all objects are equal. Hence,
in spherical categories closed planar graphs are equivalent under
isotopies of $S^{2},$ while in a pivotal category planar graphs are
equivalent under isotopies of a plane. Recall that a sphere can be
regarded as a plane with the infinity points identified by the stereographic
projection. It is then natural to regard closed graphs on a plane
as closed graphs on a sphere, however, the extra condition mentioned
is needed to allow us to pass an edge of a closed graph through the
point at infinity such that the resulting and the original closed
networks represent the same value.

In a more mathematical language, the representations of Hopf algebras
with a distinguished element satisfying some conditions, like involutory
or ribbon Hopf algebras (cf. Sec. \ref{sub:Hopf-algebras}) form in
the non-degenerate case spherical categories and the morphisms of
these categories are represented by planar graphs. If the category
of representations of a Hopf algebra is degenerate, one has to take
a non-degenerate quotient category. This quotient, however, is not
the category of representations of a finite dimensional Hopf algebra
since it is not possible to assign a dimension to each object, i.e.
a positive integer which is additive under direct sum and multiplicative
under tensor product, \cite{barrett1996invariants}. 
\begin{defn}
\label{def:pivotal-category}A \textbf{pivotal category} is a category
$\mathcal{C}$ with duals and a morphism $\epsilon(c):\, e\rightarrow c\otimes c^{*}$
for each object $c$ in $\mathcal{C}$ satisfying the following conditions\end{defn}
\begin{enumerate}
\item For all $f:\, a\rightarrow b$ in $\mathcal{C}$ following diagram
commutes:\\
\\
\includegraphics[scale=0.4]{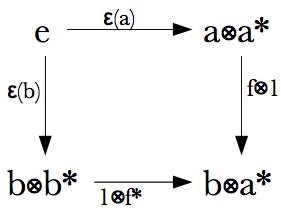}
\item For all objects $a$ in $\mathcal{C}$ following composition is $1_{a^{*}}$:\\
\\
\includegraphics[scale=0.4]{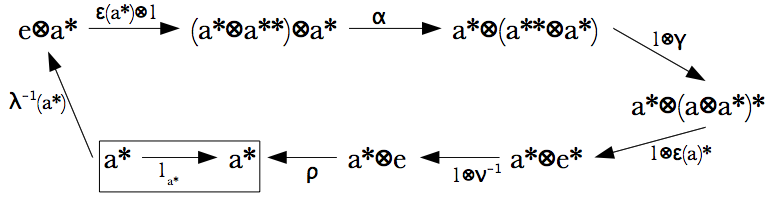}
\item For all objects $a,b$ in $\mathcal{C}$ following composite is required
to be $\epsilon(a\otimes b)$ :\\
\\
\includegraphics[scale=0.4]{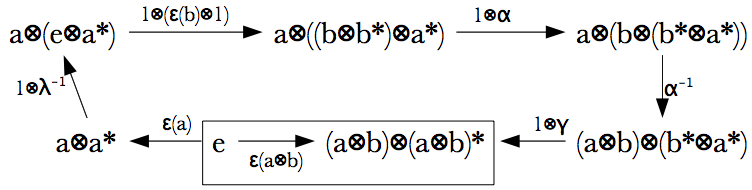}
\end{enumerate}
The morphism $\epsilon(c)$ in the above definition corresponds to
$^{c}\bigcup^{c^{*}}$. This correspondence together with the identification
of the identity $1_{a}$ with a straight line \textbf{$\biggl|_{a}^{a}$}
serves to understand the above commuting diagrams. Notice that the
dual $f^{*}$ of $f$, according to $f:\, a\rightarrow b$ in $\mathcal{C}$
and the functor $*:\,\mathcal{C}^{op}\rightarrow\mathcal{C}$, is
$f^{*}\equiv(f^{op})^{*}:\, b^{*}\rightarrow a^{*}$ since $f^{op}:\, b\rightarrow a$.
Hence we have for the diagrammatic expression of the first condition

\begin{center}
\includegraphics[scale=0.5]{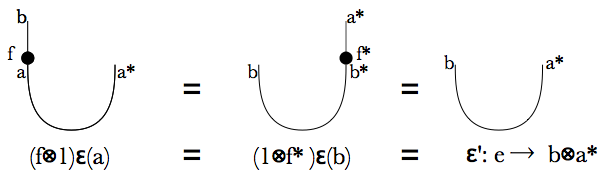}
\par\end{center}

where $\epsilon':\, e\rightarrow b\otimes a^{*}$ is a morphism in
$\mathcal{C}$. 

Now, notice that in the second condition of definition \ref{def:pivotal-category}
the dual of $\epsilon(c)$ is $\epsilon(c)^{*}:\, c^{**}\otimes c^{*}\rightarrow e^{*}$,
thus, the resulting diagram corresponding to this property of pivotal
categories is

\begin{center}
\includegraphics[scale=0.4]{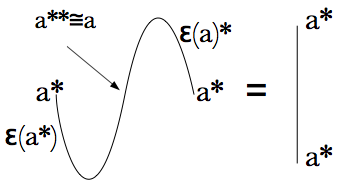}
\par\end{center}

The third condition expresses the compatibility of the pivotal structure
with the monoidal product $\otimes$ and the counit $e$:

\begin{center}
\includegraphics[scale=0.5]{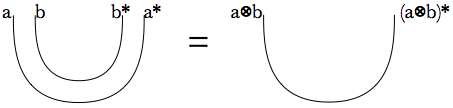}
\par\end{center}
\begin{rem*}
\textit{(i)} In a pivotal category the following composite is the
dual $f^{*}:\, b^{*}\rightarrow a^{*}$ of any morphism $f:\, a\rightarrow b$
of $\mathcal{C}$:

\includegraphics[scale=0.5]{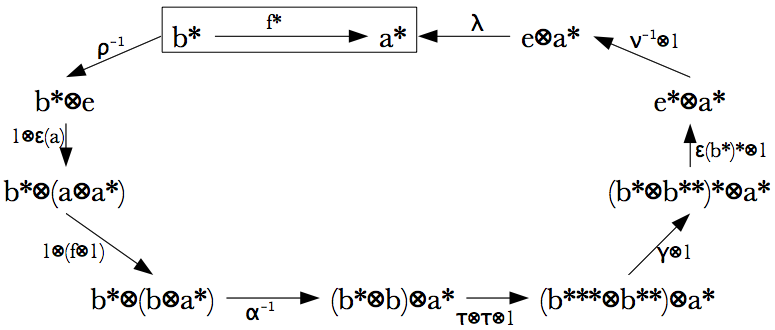}

This takes the following diagrammatic form

\includegraphics[scale=0.6]{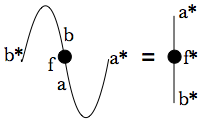}

\textit{(ii) }In this type of category for all $a,b,c$ in $\mathcal{C}$
there are natural isomorphisms such that
\begin{eqnarray*}
Hom(a\otimes b,\, c)\cong Hom(b,\, a^{*}\otimes c) & ; & Hom(a,\, b\otimes c)\cong Hom(a\otimes c^{*},\, b)\\
Hom(a\otimes b,\, c)\cong Hom(a,\, c\otimes b^{*}) & ; & Hom(a,\, b\otimes c)\cong Hom(b^{*}\otimes a,\, c)
\end{eqnarray*}

This natural isomorphisms can be described diagrammatically as following
in the case of a morphism $f\in Hom(a\otimes b,\, c)\cong Hom(a,\, c\otimes b^{*})\ni g$

\includegraphics[scale=0.5]{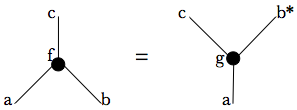}. 

This implies that given a trivalent planar graph with an orientation
of each edge, a distinguished edge at each vertex, a map from edges
to objects and a map from vertices to morphisms, then the graph can
be evaluated to obtain a morphism. The fact that this data corresponds
to a pivotal category implies that the resulting morphism is dependent
only on the isotopy class of the graph, \textit{if} the data is carried
along with the isotopy. In \prettyref{cha:Invariants-of-3-Manifolds}
we will see that this is used to construct invariants of 3-manifolds
which represent partition functions in analogy to the ones constructed
in \prettyref{sec:Connection-between-GR-SN}.

\textit{(iii)} Main examples of pivotal categories are categories
of representations of Hopf algebras which are in general not strict.
The difference between a pivotal and \textit{strict} pivotal category
is rather technical. Objects that are canonically isomorphic in a
pivotal category are equal in a strict pivotal category. We will only
consider the latter ones, however, every pivotal category is equivalent
to a strict one, hence, there is no loss of generality, \cite{barrett1996invariants,barrett1999spherical,MacLane1971}.
\end{rem*}
Now, in order to evaluate graphs one needs to define a way of mapping
a graph, which represents basically a morphism, into a scalar. More
general,%
\footnote{See also \prettyref{sec:Invariants-from-Spherical}.%
} one defines a so-called trace map as follows:
\begin{defn}
For any object $a$ in a pivotal category $\mathcal{C}$, the monoid%
\footnote{A monoid is an algebraic structure with an associative multiplication
and an identity element, e.g. a semi-group, a monoidal category with
one object, cf. \cite{MacLane1971}.%
} $End(a)$ has two \textbf{trace maps} $tr_{L},\, tr_{R}:\, End(a)\rightarrow End(e);\, f\mapsto tr_{L,R}(f)$,
which are defined to be the following composites respectively\smallskip{}

\includegraphics[scale=0.4]{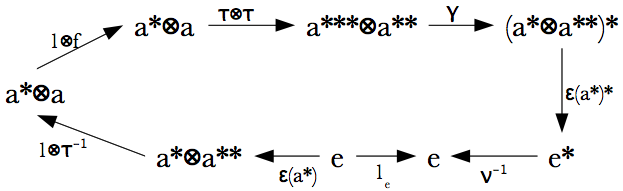}

and\smallskip{}

\includegraphics[scale=0.4]{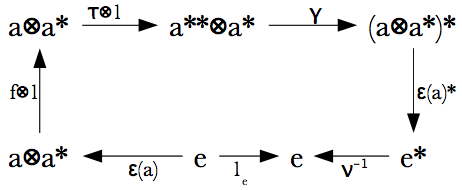}.
\end{defn}
In a strict pivotal category these definitions simplify to
\begin{eqnarray*}
tr_{L}(f) & = & \epsilon(a^{*})(1\otimes f)\epsilon(a^{*})^{*}\\
tr_{R}(f) & = & \epsilon(a)(f\otimes1)\epsilon(a)^{*}
\end{eqnarray*}
since $a^{**}=a,\: a\otimes a^{*}=(a\otimes a^{*})^{*}$.
\begin{defn}
A \textbf{spherical category} is a pivotal category for which both
trace maps coincide for all objects $a$ in $\mathcal{C}$ and all
morphisms $f\in End(a)$ of $\mathcal{C}$, i.e. 
\[
tr_{L}(f)=tr_{R}(f)\,;\,\forall a\,\text{in}\,\mathcal{C},\,\forall f\in End(a).
\]
\end{defn}
\begin{rem*}
This is equivalent to $tr_{L}(f)=tr_{L}(f^{*})\,,\,\forall f\in End(a)$.
Also, in a spherical category $tr_{L}(f\otimes g)=tr_{L}(f)\cdot tr_{L}(g)\,,\,\forall f\in End(a),\,\forall g\in End(b)$,
where $\cdot$ is the binary operation in $End(e)$.
\end{rem*}
With the definition of trace and the condition of spherical categories
we are able to define for each object $c$ in a spherical category
its quantum dimension by $dim_{q}(c)=tr_{L}(1_{c})$, cf. \prettyref{sec:Spin-Networks}
and \vref{eq:loop-value-in-TL-alg}. The spherical condition implies
then $dim_{q}(c)=dim_{q}(c^{*})$.

Now we consider categories of finitely generated modules of spherical
Hopf algebras which are additive%
\footnote{In this case ``additive'' means that all Hom-sets are finitely generated
abelian groups w.r.t. pointwise addition and that the spherical structure
is compatible with the additive structure. We assume as well that
$End(e)=\mathbb{F}$ is a field and that each Hom-set is a finite
dimensional vector space over $\mathbb{F}$, cf. \cite{barrett1999spherical}.%
} spherical categories. 
\begin{defn}
A \textbf{spherical Hopf algebra} over a ring%
\footnote{In any additive monoidal category $\mathbb{F\equiv}End(e)$ is a commutative
ring, cf. \cite{KellyLaplaza1980} as referred to in \cite{barrett1999spherical}.%
} $\mathbb{F}$ is a Hopf algebra $H$ with an antipode $S$ and an
additional structure given by an element $w\in H$ that satisfies
following conditions:\end{defn}
\begin{enumerate}
\item $S^{2}(a)=waw^{-1},\:\forall a\in H$
\item $\Delta w=w\otimes w$, i.e. $w$ is a group-like element.
\item $tr(\theta w^{-1})=tr(\theta w),\:\forall\theta\in End_{H}(V)$; where
$V$ is any finitely generated left $H$-module.
\end{enumerate}
Notice that a Hopf algebra with a group-like element that satisfies
the first condition above is spherical if either $w^{2}=1$ or all
modules are isomorphic to their dual, since the first case is trivial
and in the latter case we have from the spherical condition $tr_{L}(f)=tr_{L}(f^{*})$,
thus $tr(\theta w)=tr((\theta w)^{*})=tr(w^{*}\theta^{*})=tr(\theta w^{-1})$
by the cyclic property of the trace and the fact that the duals are
given by the antipode of $H$%
\footnote{One can see this by regarding the category with only one object, the
Hopf algebra itself, and noticing that the conditions in the definition
of a category with duals allow the correspondence between the dualisation
and the antipode. %
}. 

J. Barrett and B. Westbury proved in \cite{barrett1999spherical},
that if $H$ is a spherical Hopf algebra over $\mathbb{F}$, then
the category of left finitely generated $H$-modules, which are free%
\footnote{Here, ``free'' means in this context that the H-module is generated
by a finite linearly independent (over $\mathbb{F}$) set of elements
of the H-module. In other words the H-module has a basis, cf. \cite{Smith2009}.%
} as $\mathbb{F}$-modules, is a spherical category since an element
$w$ in a Hopf algebra satisfying the first and second condition above
determines a pivotal structure for the category of modules, which
as we saw in example \vref{exa:Tensor-product-of-modules} is a monoidal
category and the two trace maps are given by $tr_{L}(\theta)=tr(\theta w)$
and $tr_{R}(\theta)=tr(\theta w^{-1})$.

If we assign trivalent planar graphs to categories, it is the structure
just discussed which allows us to evaluate the graphs to give morphisms
depending only on the isotopy class of the graph. Furthermore, spherical
categories determine an invariant of isotopy classes of closed graphs
embedded on the sphere.

Finally, we describe some conditions needed in order to obtain a spherical
category from which one may construct an invariant of (closed) 3-manifolds.
These conditions allow to construct an invariant from a finite summation
over some objects of the spherical category with the trace of some
map as summands, cf. \prettyref{cha:Invariants-of-3-Manifolds}.

The first condition is non-degenerancy which allows the construction
of the second condition, which is that of a spherical category being
semisimple. This allows an isomorphism between a sum of tensor products
of vector spaces given by the Hom-sets and its corresponding Hom-set
with a vector space structure, see \prettyref{eq:semi-simple-condition-sph-cat}.

Following definitions are given in \cite{barrett1996invariants},
\begin{defn}
For any two objects $a,b$ in $\mathcal{C}$, there is a \textbf{bilinear
pairing}
\[
\Theta:\, Hom(a,b)\otimes Hom(b,a)\rightarrow\mathbb{F}
\]
defined by $\Theta(f,g)=tr_{L}(fg)=tr_{L}(gf)$. An additive spherical
category is non-degenerate if, for all objects $a$ and $b$, the
pairing $\Theta$ is non-degenerate in the usual sense.

A \textbf{semisimple} spherical category $\mathcal{C}$ is \textit{additive
and non-degenerate} such that there exists a set of inequivalent non-zero%
\footnote{An object $a$ is non-zero if $End(a)\neq0$.%
} objects of a set $J$, so that for any two objects $x,y$ in $\mathcal{C}$,
the natural map given by
\begin{equation}
\bigoplus_{a\in J}Hom(x,a)\otimes Hom(a,y)\rightarrow Hom(x,y)\label{eq:semi-simple-condition-sph-cat}
\end{equation}
is an isomorphism. An object is called simple if its endomorphism
ring is isomorphic to $\mathbb{F}$.
\end{defn}
The set $J$ is fixed by the category since there is an isomorphism
between any simple object in the category and its corresponding (unique)
element in $J$. Thus, every element in $J$ is simple. A semisimple
spherical category is called finite if $J$ is finite, i.e. if the
set of isomorphism classes of simple objects is finite. The dimension
of this finite category is then defined by $K=\sum_{a\in J}\text{dim}_{q}^{2}(a)$.

Next theorem ensures that it is always possible to construct a non-degenerate
additive spherical category from a category which is only additive
and spherical. The proof can be found in \cite{barrett1999spherical}.
\begin{thm}
\label{thm:non-degenerate-quotient}For a given additive spherical
category $\mathcal{C}$, the additive subcategory $\mathcal{J}$ defined
to have the same set of objects and the Hom-sets defined by
\[
Hom_{\mathcal{J}}(a,b)=\left\{ f\in Hom_{\mathcal{C}}(a,b)\,;\,\Theta(f,g)=tr_{L}(fg)=0\text{ for all }g\in Hom_{\mathcal{C}}(b,a)\right\} 
\]
gives a quotient $\mathcal{C}/\mathcal{J}$ which is a non-degenerate
additive spherical category.
\end{thm}
The above theorem is very important since this construction is the
general way of attaining the semisimple condition needed for the construction
of invariants of closed 3-manifolds. This result is proved in \cite{barrett1999spherical}
and states that, given a spherical Hopf algebra $H$ over a field
$\mathbb{F}$, the non-degenerate quotient of the spherical category
of finitely generated left $H$-modules is semisimple. 

We will see in \prettyref{cha:Invariants-of-3-Manifolds} that the
above gives rise to the generalized construction of the Turaev-Viro
invariant as a state sum over the quotient of the category of representations
of the deformed quantum enveloping Hopf algebra $U_{q}(\mathfrak{sl}_{2})$,
which can be made a spherical Hopf algebra. Hence, these categories
can be seen as the generalization of the objects introduced in \prettyref{sec:Spin-Networks}.

\section{Some More Diagrams: The Temperley-Lieb Recoupling Theory \label{sec:Some-More-Diagrammatics}}

In this section, the diagrams of the previous sections are formalized
and it is intended as a short introduction in this very broad area
of knot theory. We follow \cite{kauffman2001knots} for the definition
of the bracket polynomial and a discussion of the Jones polynomial,
as well as \cite{kauffman1994temperley} for the discussion on Temperley-Lieb
recoupling theory. 

First we introduce the bracket polynomial to associate a knot to an
invariant in order to evaluate the given diagram, which is defined
as an schematized picture of the given knot in a plane. The diagram
is composed of curves that cross in 4-valent vertices. Each of these
vertices are equipped with the extra structure given by an under-
or over-crossing:

\begin{center}
\begin{eqnarray*}
\xygraph{!{0;/r2.0pc/:}[u(0.5)][ll]!{\xunderv}} & \text{and} & \xygraph{!{0;/r2.0pc/:}[u(0.5)][r]!{\xoverv}}\\
Overcrossing &  & Undercrossing
\end{eqnarray*}

\par\end{center}

Now, let us consider a crossing in an unoriented diagram, then two
associated diagrams can be obtained by splicing the crossing in two
ways as
\[
\underset{B}{\xygraph{!{0;/r1.75pc/:}[u(0.5)]!{\xunderv->{B}|{A}<{A}}}}\rightarrow\begin{cases}
\xygraph{!{0;/r1.0pc/:}[ld(0.5)]!{\vcap}[u(1.5)]!{\vcap-}} & A\\
\xygraph{!{0;/r1.0pc/:}[lu(0.5)]!{\hcap}[r(1.5)]!{\hcap-}} & B
\end{cases}
\]
where $A$ or $B$ denote the type of splitting, i.e. an $A-/B-$split
joints the regions labelled $A/B$ at the crossing, \cite{kauffman2001knots}.
With this convention a split crossing labelled $A$ or $B$ in a diagram
can be reconstructed to form the original crossing. Hence, by keeping
track of the labelling one can reconstruct the original link from
\emph{any }of its descendants. The primitive descendants of a link
$K$, those which have no crossings left, are collections of Jordan
curves in the plane%
\footnote{Jordan curves are closed loops in the plane homeomorphic to $S^{1}$.%
} and are called the \textbf{states of }$K$. The labelling of the
above mentioned splitting process makes the algorithm to evaluate
a link unambiguous since each primitive descendant can be associated
in the same way to its ancestral link. From these states we are able
to construct invariants of knots by averaging over them.
\begin{defn}
The \textbf{bracket polynomial} is defined to be the following formula
\begin{equation}
\left\langle K\right\rangle =\sum_{\sigma}\left\langle K\,|\sigma\right\rangle d^{\left\Vert \sigma\right\Vert }\label{eq:bracket-polynomial}
\end{equation}
where $\sigma$ denotes a state of $K$, $\left\Vert \sigma\right\Vert =(\text{loops in }\sigma)-1$
and $\left\langle K\,|\sigma\right\rangle $ the product of (commutative)
labels of $\sigma$, i.e. the product of A's and B's labelling a given
state. Here $d$ and the labels are commuting algebraic variables.
Hence, in the case relevant for us where $B=A^{-1}$ and $d=-A^{2}-A^{-2}$
we have that $\left\langle K\right\rangle \in\mathbb{Z}\bigl[A,\, A^{-1}\bigr]$
is a polynomial in $A\text{ and }A^{-1}$ with integer coefficients.
\end{defn}
As defined, the bracket polynomial is not a topological invariant,
since in this form its value changes under Reidemeister moves. To
see this notice that for an over-crossing $\UseComputerModernTips\xygraph{!{0;/r1.0pc/:}[u(0.5)]!{\htwist}}$
we have%
\footnote{This is understood as regarding the over-crossing and the splits as
part of a bigger diagram, i.e. we consider the splitting process locally
and the rest of the diagram stays unchanged. For the proof of this
relation and the following ones see \cite[Part I, Sec. 3]{kauffman2001knots}. %
} 
\begin{equation}
\left\langle \UseComputerModernTips\xygraph{!{0;/r1.0pc/:}[u(0.5)]!{\htwist}}\right\rangle =A\left\langle \UseComputerModernTips\xygraph{!{0;/r1.0pc/:}[u(0.5)]!{\huntwist}}\right\rangle +B\left\langle \UseComputerModernTips\xygraph{!{0;/r1.0pc/:}[u(0.5)]!{\vuntwist}}\right\rangle \label{eq:splitting}
\end{equation}

The reiterated use of the above relation%
\footnote{Its analog one for the under-crossing is not needed since switching
the crossings exchanges the roles of $A$ and $B$.%
} is all is needed to compute the bracket. Now, it is easy to see that
the diagrams in the first and second Reidemeister moves give the following
bracket relations
\[
\begin{array}{ccc}
\left\langle \xygraph{!{0;/r1.0pc/:}[u]!{\vover}!{\vunder-}[r]!{\xcaph@(0)}[luu]!{\xcaph@(0)}[lll]!{\xcaph@(0)}[ldd]!{\xcaph@(0)}}\right\rangle  & = & AB\left\langle \xygraph{!{0;/r1.0pc/:}[u(0.5)]!{\vuntwist}}\right\rangle +AB\left\langle \xygraph{!{0;/r1.0pc/:}[u(0.5)]!{\huntwist}[d(0.5)]!{\vcap[-0.7]}!{\vcap[0.7]}}\right\rangle +(A^{2}+B^{2})\left\langle \xygraph{!{0;/r1.0pc/:}[u(0.5)]!{\huntwist}}\right\rangle \\
\left\langle \xygraph{!{0;/r1.0pc/:}[u(0.75)]!{\vover}!{\vcap-}[lu]!{\xcaph@(0)}[r]!{\xcaph@(0)}}\right\rangle  & = & (Ad+B)\left\langle \xygraph{!{0;/r1.0pc/:}[u]!{\xcaph[3]@(0)}}\right\rangle \\
\left\langle \xygraph{!{0;/r1.0pc/:}[u(0.75)]!{\vunder}!{\vcap-}[lu]!{\xcaph@(0)}[r]!{\xcaph@(0)}}\right\rangle  & = & (A+Bd)\left\langle \xygraph{!{0;/r1.0pc/:}[u]!{\xcaph[3]@(0)}}\right\rangle 
\end{array}
\]

So, if $B=A^{-1}$ and $d=-(A^{2}+A^{-2})$, then we obtain 

\[
\begin{array}{ccc}
\left\langle \xygraph{!{0;/r0.75pc/:}[u]!{\vover}!{\vunder-}[r]!{\xcaph@(0)}[luu]!{\xcaph@(0)}[lll]!{\xcaph@(0)}[ldd]!{\xcaph@(0)}}\right\rangle  & = & \left\langle \xygraph{!{0;/r1.8pc/:}[u(0.5)]!{\hcap}[r(1.2)]!{\hcap-}}\right\rangle \\
\left\langle \xygraph{!{0;/r1.0pc/:}[u(0.75)]!{\vover}!{\vcap-}[lu]!{\xcaph@(0)}[r]!{\xcaph@(0)}}\right\rangle  & = & (-A^{3})\left\langle \xygraph{!{0;/r1.0pc/:}[u]!{\xcaph[1.5]@(0)}}\right\rangle \\
\left\langle \xygraph{!{0;/r1.0pc/:}[u(0.75)]!{\vunder}!{\vcap-}[lu]!{\xcaph@(0)}[r]!{\xcaph@(0)}}\right\rangle  & = & (-A^{-3})\left\langle \xygraph{!{0;/r1.0pc/:}[u]!{\xcaph[1.5]@(0)}}\right\rangle 
\end{array}
\]

Note that from the definition of the bracket polynomial, it follows
that $\left\langle OK\right\rangle =d\left\langle K\right\rangle $
where $OK$ denotes the disjoint union of a closed loop to the diagram
$K$. From this, it follows that the bracket with the conditions $B=A^{-1},\, d=-(A^{2}+A^{-2})$
is an invariant of regular isotopy, i.e. under the moves II and III%
\footnote{The invariance under move III follows from the invariance under move
II, \cite{kauffman2001knots}.%
}. In order to have an invariant of ambient isotopy, i.e. under moves
I, II, III, we need to normalize the bracket polynomial. The normalized
bracket $\mathcal{L}_{K}$ for oriented links $K$ is defined by
\begin{equation}
\mathcal{L}_{K}=(-A^{3})^{-w(K)}\left\langle K\right\rangle \label{eq:normalized-BP}
\end{equation}
 where $w(K)$ is the writhe of $K$ defined by $w(K)=\sum_{p}\epsilon(p)$
where $p$ runs over all crossings in $K$ and $\epsilon(p)=\pm1$
is the sign of the over- and under-crossing respectively. Hence, \prettyref{eq:normalized-BP}
is an invariant of ambient isotopy since $w(K)$ is an invariant of
regular isotopy and
\[
\begin{array}{ccc}
w\Bigl(\xygraph{!{0;/r0.7pc/:}[u(0.75)]!{\vover}!{\vcap-}[lu]!{\xcaph@(0)}[r]!{\xcaph@(0)}}\Bigr) & = & 1+w(\xygraph{!{0;/r1.0pc/:}[u]!{\xcaph@(0)}})\\
w\Bigl(\xygraph{!{0;/r0.7pc/:}[u(0.75)]!{\vunder}!{\vcap-}[lu]!{\xcaph@(0)}[r]!{\xcaph@(0)}}\Bigr) & = & -1+w(\xygraph{!{0;/r1.0pc/:}[u]!{\xcaph@(0)}})
\end{array}
\]

\begin{rem*}
The mirror image $K^{*}$ of a link $K$ is the link resulting from
the interchange of over- and under-crossings. Let $K^{*}$ be the
mirror image of an oriented link $K$, then 
\[
\begin{array}{ccc}
\left\langle K^{*}\right\rangle (A) & = & \left\langle K\right\rangle (A^{-1})\\
\mathcal{L}_{K^{*}}(A) & = & \mathcal{L}_{K}(A^{-1})
\end{array}
\]

Hence, if $\mathcal{L}_{K}(A)\neq\mathcal{L}_{K}(A^{-1})$, then $K$
is not ambient isotopic to its mirror image $K^{*}$.
\end{rem*}
We now give the definition of the Jones polynomial as well as its
relation to the normalized bracket polynomial, which ensures the existence
and well-definiteness of the one-variable Jones polynomial.
\begin{defn}
The one-variable Jones polynomial $V_{K}(t)$ is a polynomial in $t^{1/2}$
with finitely many positive and negative powers of $t^{1/2}$ (i.e.
a Laurent polynomial) associated to an oriented link $K$ and which
satisfies the following properties:\end{defn}
\begin{enumerate}
\item If $K$ is ambient isotopic to $K'$, $V_{K}(t)=V_{K'}(t).$
\item $V_{O^{+}}(t)=1$, where $O^{+}$ denotes the loop having clockwise
orientation: $\UseComputerModernTips\xygraph{!{0;/r1.5pc/:}[u(0.5)]!{\hcap-}!{\hcap=<}}$ 
\item $\UseComputerModernTips\knottips{FT}t^{-1}V_{t}\Bigl[\xygraph{!{0;/r1.0pc/:}[u(0.5)]!{\htwist}}\Bigr]-tV_{t}\Bigl[\xygraph{!{0;/r1.0pc/:}[u(0.5)]!{\htwistneg}}\Bigr]=\Bigl(t^{1/2}-t^{-1/2}\Bigr)V_{t}\Bigl[\xygraph{!{0;/r1.0pc/:}[u(0.5)]!{\huntwist}}\Bigr]$,
where the brackets are a notation for links which differ from each
other only in the showed crossing.
\end{enumerate}
With this definition of $V_{K}(t)$ and \prettyref{eq:normalized-BP},
we have that the normalized bracket relates to the 1-variable Jones
polynomial as
\[
\mathcal{L}_{K}(t^{-1/4})=V_{K}(t).
\]
This shows that even if the definition above is not obviously well-defined
we can take it as given since $\mathcal{L}_{K}$ exists and is well-defined,
\cite{kauffman2001knots}. 

The Jones polynomial has been generalized in many ways. One way involves
choosing a compact Lie group $G$ and an irreducible representation.
Then a polynomial invariant of oriented links corresponding to these
irreducible representations is constructed by using solutions of the
Yang-Baxter equations \prettyref{eq:Yang-Baxter-eq}. With this method,
the original Jones polynomial corresponds to choosing $G=SU(2)$ with
its standard representation on $\mathbb{C}^{2}$, \cite{atiyah1990geometry}.
In fact, the quantum group $U_{q}(\mathfrak{sl}_{2})$ gives rise
to the Jones polynomial and the Yang-Baxter equations are the algebraical
form of Move III in \prettyref{sec:Spin-Networks}, \cite{sawin1996links}.

\subsection{Temperley-Lieb Algebra\label{sub:Temperley-Lieb-Algebra}}

Now we proceed with the introduction%
\footnote{The following can be found in \cite{kauffman1994temperley}.%
} of an algebra which allows us to construct the bracket polynomial
and the projectors given in definition \vref{def:The-weaving}. For
this, first define the elementary tangles $U_{0}=1_{n},\, U_{1},\, U_{2},\dots,\, U_{n-1}\in T_{n}$
of the $n$-strand algebra $T_{n}$. We can think of each of these
tangles as being $n$ strands fixed at their ends in a box having
$n$ fixing points at the upper margin (outputs) and $n$ fixing points
at the lower margin (inputs), such that in $U_{i}$ the $k^{th}$
output is connected for $k\neq i,\, i+1$ to the $k^{th}$ input and
the $i^{th}$ input and output are connected to the $(i+1)-th$ input
and output respectively. The multiplication of the elements is given
by attaching the output with the input, i.e. by stacking the boxes
one above the other. The algebra is given by the following defining
relations:
\begin{enumerate}
\item $U_{i}^{2}=dU_{i}$, where $d$ is a value assigned to a closed loop.
\item $U_{i}U_{i+1}U_{i}=U_{i}$
\item $U_{i}U_{j}=U_{j}U_{i}$ for $|i-j|>1.$\end{enumerate}
\begin{example*}
Consider the case $n=4$. The elementary tangles in $T_{4}$ are given
by 
\[
1_{4}\equiv\xygraph{!{0;/r2.0pc/:}[u(0.5)]!{\xcapv@(0)}[r(0.3)u]!{\xcapv@(0)}[r(0.3)u]!{\xcapv@(0)}[r(0.3)u]!{\xcapv@(0)}};\; U_{1}\equiv\xygraph{!{0;/r2.0pc/:}[d(0.5)]!{\vloop[0.3]}[u]!{\vloop[-0.3]}[r(0.5)]!{\xcapv@(0)}[r(0.3)u]!{\xcapv@(0)}};\; U_{2}\equiv\xygraph{!{0;/r2.0pc/:}[u(0.5)]!{\xcapv@(0)}[r(0.2)]!{\vloop[0.3]}[u]!{\vloop[-0.3]}[r(0.5)]!{\xcapv@(0)}};\; U_{3}\equiv\xygraph{!{0;/r2.0pc/:}[u(0.5)]!{\xcapv@(0)}[r(0.3)u]!{\xcapv@(0)}[r(0.25)]!{\vloop[0.3]}[u]!{\vloop[-0.3]}}
\]
 with the relations, say 
\[
\xygraph{!{0;/r1.7pc/:}[r(4)][d]!{\vloop[0.3]}[u]!{\vloop[-0.3]}[r(0.5)]!{\xcapv@(0)}[r(0.3)u]!{\xcapv@(0)}[u][l(0.8)]!{\vloop[0.3]}[u]!{\vloop[-0.3]}[r(0.5)]!{\xcapv@(0)}[r(0.3)u]!{\xcapv@(0)}}\cong\,\xygraph{!{0;/r2.0pc/:}[r]!{\vcap[-0.3]}!{\vcap[0.3]}[rd(0.5)]!{\vloop[0.3]}[u]!{\vloop[-0.3]}[r(0.5)]!{\xcapv@(0)}[r(0.3)u]!{\xcapv@(0)}}
\]
\[
\xygraph{!{0;/r2.0pc/:}[r][d(1.5)]!{\vloop[0.3]}[u]!{\vloop[-0.3]}[r(0.55)]!{\xcapv@(0)}[r(0.3)u]!{\xcapv@(0)}[uu][l(0.85)]!{\xcapv@(0)}[r(0.3)]!{\vloop[0.265]}[u]!{\vloop[-0.265]}[r(0.54)]!{\xcapv@(0)}[u][l(0.85)]!{\vloop[0.3]}[u]!{\vloop[-0.3]}[r(0.57)]!{\xcapv@(0)}[r(0.29)u]!{\xcapv@(0)}}\cong\;\xygraph{!{0;/r2.0pc/:}[d(0.5)]!{\vloop[0.3]}[u]!{\vloop[-0.3]}[r(0.5)]!{\xcapv@(0)}[r(0.3)u]!{\xcapv@(0)}}
\]

\[
\xygraph{!{0;/r2.0pc/:}!{\xcapv@(0)}[r(0.29)u]!{\xcapv@(0)}[r(0.25)]!{\vloop[0.3]}[u]!{\vloop[-0.3]}[l(0.55)]!{\vloop[0.283]}[u]!{\vloop[-0.283]}[r(0.53)]!{\xcapv@(0)}[r(0.3)u]!{\xcapv@(0)}}\cong\;\xygraph{!{0;/r2.0pc/:}[u]!{\xcapv@(0)}[r(0.29)u]!{\xcapv@(0)}[r(0.25)]!{\vloop[0.3]}[u]!{\vloop[-0.3]}[l(0.55)][dd]!{\vloop[0.283]}[u]!{\vloop[-0.283]}[r(0.53)]!{\xcapv@(0)}[r(0.3)u]!{\xcapv@(0)}}
\]

\end{example*}
Where $\simeq$ means that the tangles are equivalent. This is the
case when they are regularly isotopic relative to their (fixed) endpoints.
Every planar non-intersecting $n$-tangle is equivalent to a product
of the $n$ elementary tangles and two such products represent equivalent
tangles if and only if one product can be obtain from the other by
the relations above. 
\begin{defn}
The \textbf{Temperley-Lieb algebra} $T_{n}$ is a free additive algebra
over the set of rational functions with numerator and denominator
in $\mathbb{Z}\bigl[A,\, A^{-1}\bigr]$ and with multiplicative generators
$1_{n},\, U_{1},\dots,\, U_{n-1}$ and the relations given above.
The value of the loop is $d=-A^{2}-A^{-2}$ and since $A\text{ and }A^{-1}$
commute with all elements of $T_{n}$, thus $d$ too. 
\end{defn}
In order to evaluate an $n$-tangle $x$ we define a trace map $tr:\, T_{n}\rightarrow\mathbb{Z}\bigl[A,\, A^{-1}\bigr]$
defined by $tr(x)=\langle\bar{x}\rangle$ and $tr(x+y)=tr(x)+tr(y)$
where $\bar{x}$ denotes the standard closure of $ $$x$ obtained
by joining the $k^{th}$ input and output from outside the tangle,
such that each strand forms a loop going from the bottom of the n-tangle
to the top of the n-tangle. This trace map is defined also for the
generalization of $T_{n}$ to the tangle algebra generated multiplicatively
by $n$-strand tangles of general form, e.g. there can be crossings
inside the box representing the tangle. Note that the cyclic property
of the trace map, i.e. $tr(ab)=tr(ba)$, is a direct consequence of
the properties of the bracket polynomial%
\footnote{Recall that the coefficients $\left\langle \bar{ab}\left|K\right.\right\rangle $
corresponding to the states of $K$ are products of commutating factors.%
} and the standard closure:
\begin{gather*}
tr(ab)=\langle\overline{ab}\rangle=\sum_{\sigma}\langle\overline{ab}|\sigma\rangle d^{\|\sigma\|}=\sum_{\sigma}\left\langle \bar{a}|\sigma\right\rangle \left\langle \bar{b}|\sigma\right\rangle d^{\|\sigma\|}\\
=\sum_{\sigma}\left\langle \bar{b}|\sigma\right\rangle \left\langle \bar{a}|\sigma\right\rangle d^{\|\sigma\|}=\sum_{\sigma}\langle\overline{ba}|\sigma\rangle d^{\|\sigma\|}=tr(ba)
\end{gather*}
 
\begin{rem*}
In the case where $x\in T_{n}$, i.e. $x$ is a product of $U_{i}$'s,
we have that $\bar{x}$ is a disjoint union of Jordan curves so $tr(x)=d^{\|\bar{x}\|}$,
where $\|\bar{x}\|$ is the number of loops in the plane. Thus, each
product of $U_{i}$'s correspond to a single bracket state.
\end{rem*}
Consider now the Artin braid group $B_{n}$ which is one special case
of the above generalization and is useful to formalize the concept
of $n$-edge in \prettyref{sec:Spin-Networks}. As the name suggests
a braid in $B_{n}$ is a collection of $n$ strands woven into a single
$n$-tangle. Notice that weaving two strands, say the $i^{th}$ and
$(i+1)^{th}$, is nothing more than crossing one above the other.
There are, however, two ways of doing this, weaving the $i^{th}$
over the $(i+1)^{th}$ forming an over-crossing or vice versa, forming
an under-crossing. Denote these two ways, $\sigma_{i}$ and $\sigma_{i}^{-1}$
respectively. Denoting $\sigma_{i}^{-1}$ as the inverse of $\sigma_{i}$
makes sense since they are, in fact, inverse to each other:
\[
\sigma_{i}=\xygraph{!{0;/r2.0pc/:}[u(0.5)]!{\xcapv@(0)|{...}<{1}}[r(0.5)u]!{\xcapv@(0)}[r(0.25)u]!{\xunderv[0.8]>{i}}[r][u]!{\xcapv@(0)|{...}}[r(0.5)u]!{\xcapv@(0)<{n}}}\text{and }\quad\sigma_{i}^{-1}=\xygraph{!{0;/r2.0pc/:}[u(0.5)]!{\xcapv@(0)|{...}<{1}}[r(0.5)u]!{\xcapv@(0)}[r(0.25)u]!{\xoverv[0.9]<{i}}[r(1.1)][u]!{\xcapv@(0)|{...}}[r(0.5)u]!{\xcapv@(0)<{n}}}
\]
 such that
\[
\sigma_{i}\sigma_{i}^{-1}=\qquad\xygraph{!{0;/r2.0pc/:}[u]!{\xcapv[1.8]@(0)<{1}|{...}}[r(0.5)u]!{\xcapv[1.8]@(0)}[r(0.25)][u(0.9)]!{\xunderv[0.8]}[u(0.2)]!{\xoverv[0.8]}[r][u(2)]!{\xcapv[1.8]@(0)|{...}}[r(0.5)u]!{\xcapv[1.8]@(0)<{n}}}\cong1_{n}
\]
 due to the second Reidemeister move.

Now, considering \prettyref{eq:splitting} each $\sigma_{i}$ has
an horizontal type $A$ smoothering $H(\sigma_{i})=U_{i}$ and a vertical
type $A^{-1}$ smoothering $V(\sigma_{i})=1_{n}$ and for $\sigma_{i}^{-1}$
the opposite types. Hence, each bracket state of the closure $\bar{b}$
of a braid $b$ corresponds to the closure of an element of the Temperley-Lieb
algebra. This means that there exists a representation $\rho:\, B_{n}\rightarrow T_{n}$
given by 
\[
\begin{array}{ccc}
\rho(\sigma_{i}) & = & AU_{i}+A^{-1}1_{n}\\
\rho(\sigma_{i}^{-1}) & = & A^{-1}U_{i}+A1_{n}
\end{array}
\]
 which allows us to compute $\langle\bar{b}\rangle$ as a sum of trace
evaluations of elements of $T_{n}$, i.e. $tr(\rho(b))=\langle\bar{b}\rangle$.
Hence, the $n$-edges introduced in \prettyref{sec:Spin-Networks}
are related to the Temperley-Lieb algebra via this representation
and, moreover, its value defined as the value of the bracket polynomial
\prettyref{eq:bracket-polynomial} of its closure is well-defined
and computable via the above representation. 

Before discussing some properties of the projectors, or $n$-edges,
and the 3-vertex as defined in \prettyref{sec:Spin-Networks}, we
give a more general and formal definition of the projectors as a sum
of $n$-tangles.

First, consider the element $f_{i}\in T_{n}$ defined inductively
for $i=0,\,1,\dots,\, n-1$ by
\[
\begin{array}{ccc}
(i)\quad f_{0} & = & 1_{n}\\
(ii)\; f_{i+1} & = & f_{i}-\mu_{i+1}f_{i}U_{i+1}f_{i}
\end{array}
\]
 where $\mu_{1}=d^{-1},\,\mu_{i+1}=(d-\mu_{i})^{-1}$.

With this definition the elements $f_{i}$ have following properties
for $i\in\{0,\,1,\dots,\, n-1\}$, \cite{kauffman1994temperley}:
\begin{enumerate}
\item $f_{i}^{2}=f_{i}$ which corresponds to the projection property in
\prettyref{sec:Spin-Networks}
\item $f_{i}U_{j}=U_{j}f_{i}=0$ for $j\leq i$, corresponds to the irreducibility
of an $i$-edge as in \prettyref{sec:Spin-Networks}
\item $tr(f_{n-1})=\Delta_{n}\bigl(-A^{2}\bigr)$ and $\mu_{i+1}=\Delta_{i}/\Delta_{i+1}$
with $\Delta_{0}=1$ and 
\[
\Delta_{n}(x)=\frac{x^{n+1}-x^{-n-1}}{x-x^{-1}}
\]
 is called the $n^{th}$ Chebyschev polynomial. Furthermore, it holds
that $tr(f_{i-1})=\Delta_{i}$ when the trace is taken with respect
to $i$-tangles since $T_{1}<T_{2}<\dots<T_{n}<\dots$
\end{enumerate}
With respect to the algebra $T_{n}$ there is a unique non-zero element
given by $f_{n-1}\in T_{n}$ such that the second property holds for
all $j=1,\,\dots,\, n-1$. Hence, with the definition of the projector
$\xygraph{!{0;/r0.5pc/:}[u(2)]!{\xcapv@(0)}[d]n*\frm{-}[d]!{\xcapv@(0)}}$
given below, we have%
\footnote{Since any $n$-tangle can be seen as a sum of elements in the Temperley-Lieb
algebra we can regard $\xygraph{!{0;/r0.5pc/:}[u(2)]!{\xcapv@(0)}[d]n*\frm{-}[d]!{\xcapv@(0)}}$
as in $T_{n}$.%
} $f_{n-1}=\xygraph{!{0;/r0.5pc/:}[u(2)]!{\xcapv@(0)}[d]n*\frm{-}[d]!{\xcapv@(0)}}$. 
\begin{defn}
\label{def:Projector-general-def}For a given positive integer $n$,
define the $n$-tangle obtained from the sum of elements of the minimal
representations of the symmetric group $S_{n}$, as a product of transpositions,
as follows:
\[
\xygraph{!{0;/r0.5pc/:}[u(2)]!{\xcapv@(0)}[d]n*\frm{-}[d]!{\xcapv@(0)}}=\frac{1}{[n;A^{-4}]!}\sum_{\sigma\in S_{n}}\bigl(A^{-3}\bigr)^{\|\sigma\|}\xygraph{!{0;/r0.5pc/:}[u(2)][r(1.5)]!{\xcapv@(0)}[d]{\hat{\sigma}}*\frm{-}[d]!{\xcapv@(0)}}
\]
where $\|\sigma\|\in\mathbb{N}$ is the number of transpositions in
the minimal representation of $\sigma$ and $\left[n;A^{-4}\right]!=\sum_{\sigma\in S_{n}}\bigl(A^{-4}\bigr)^{\|\sigma\|}=\prod_{k=1}^{n}\frac{1-A^{-4k}}{1-A^{-4}}$
is the $q$-deformed factorial with $q=A^{2}$ and the property that
for $A=\pm1\;\text{we have }[n]!=n!$. Notice that $\sigma\in S_{n}$
is represented by an element $\hat{\sigma}\in B_{n}$. 
\end{defn}
This is a more general definition as the one given in \prettyref{sec:Spin-Networks}.
Here the braiding of the $n$ strands also correspond to permutations,
but the concept has now being formalized by identifying a transposition
$(i\: i+1)\in S_{n}$ with the generator $\sigma_{i}$ of the Artin
braid group $B_{n}$.
\begin{example*}
In the case where $n=2$ the explicit expansion of the 2-edge is as
follows: 
\begin{multline*}
\xygraph{!{0;/r0.5pc/:}[u(2)]!{\xcapv[1.5]@(0)}[ru]!{\xcapv[1.5]@(0)}[d][l(0.5)]{\hspace{1em}}*\frm{-}[l(0.5)][d(0.6)]!{\xcapv[1.5]@(0)}[ru]!{\xcapv[1.5]@(0)}}=\frac{1}{[2;A^{-4}]!}\biggl[\ \ \ \xygraph{!{0;/r1.0pc/:}[u(0.75)]!{\xcapv[1.5]@(0)}[r(0.5)][u]!{\xcapv[1.5]@(0)}[l]}+A^{-3}\ \xygraph{!{0;/r1.0pc/:}[u(0.5)][ll]!{\xunderv}}\biggr]\\
=\frac{1}{1+A^{-4}}\biggl[\ \ \ \xygraph{!{0;/r1.0pc/:}[u(0.75)]!{\xcapv[1.5]@(0)}[r(0.5)][u]!{\xcapv[1.5]@(0)}[l]}+A^{-3}\ \biggl[A\xygraph{!{0;/r1.0pc/:}[d(0.5)]!{\xcaph[1.5]}[u(1.25)][l]!{\xcaph[-1.5]}}+A^{-1}\xygraph{!{0;/r1.0pc/:}[u]!{\xcapv[1.5]}[ur]!{\xcapv[-1.5]}}\biggr]\biggr]\ \ \ \ \ \ \ \ \ \ \ \ \ \\
=\xygraph{!{0;/r1.0pc/:}[u(0.75)]!{\xcapv[1.5]@(0)}[r(0.5)][u]!{\xcapv[1.5]@(0)}[l]}-\frac{1}{d}\xygraph{!{0;/r1.0pc/:}[d(0.5)]!{\xcaph[1.5]}[u(1.25)][l]!{\xcaph[-1.5]}}\ \ \ =f_{1}\ \ \ \ \ \ \ \ \ \ \ \ \ \ \ \ \ \ \ \ \ \ \ \ \ \ \ \ \ \ \ \ \ \ \ \ \ \ \ \ \ \ \ \ \ \ \ \ \ \ \ \ \ \ \ \ \ \ \ 
\end{multline*}
where $d=-A^{2}-A^{-2}$ .
\end{example*}
With the diagrammatic notation, the recursion relation defining $f_{n-1}=\xygraph{!{0;/r0.5pc/:}[u(2)]!{\xcapv@(0)}[d]n*\frm{-}[d]!{\xcapv@(0)}}$
is given by
\begin{equation}
\xygraph{!{0;/r0.5pc/:}[u(2)]!{\xcapv[1.5]@(0)|{n+1}}[d]{\hspace{1em}}*\frm{-}[d(0.6)]!{\xcapv[1.5]@(0)}}\;=\quad\xygraph{!{0;/r0.5pc/:}[u(2)]!{\xcapv[1.5]@(0)|{n}}[d]{\hspace{1em}}*\frm{-}[d(0.6)]!{\xcapv[1.5]@(0)}[r(2)][u(3.5)]!{\xcapv[4]@(0)}}-\frac{\Delta_{n-1}}{\Delta_{n}}\,\xygraph{!{0;/r0.5pc/:}[ru(3.5)]!{\xcapv@(0)}[d(0.7)]{\hspace{1em}}*\frm{-}[d(0.7)]!{\xcapv[-1.6]@(0)|{n-1}}[u(3.4)][r(0.75)]!{\xcapv@(0)}[d(1.4)]!{\vcap-}[r][u(2.1)]!{\xcapv[2]@(0)}[d(3.2)][l(1.8)]{\hspace{1em}}*\frm{-}[d(0.6)]!{\xcapv@(0)}[u][r(0.75)]!{\xcapv@(0)}[u(2.1)]!{\vcap}[r]!{\xcapv[2]@(0)}}\label{eq:diagram-recursion-relation}
\end{equation}
 from which the recursion relation for the evaluation of its closure
follows, cf. Footnote \vref{fn:loop-value-recursion-relation}:
\[
\Delta_{n+1}=\Delta_{n}d-\Delta_{n-1}
\]

Both relations above are very useful for calculations of closed diagrams,
for instance, the theta-value of a trivalent vertex discussed at the
end of this section. 

Finally we discuss some properties of the projectors and 3-vertices
which are important for the general discussion of the recoupling theory
and the $q$-deformed $6j$-symbols. In this case we impose $A^{2}=q$
where $q$ is a $2r$-th primitive root of unity such that $q^{r}=-1$,
$ $$q=e^{i\pi/r}$ and $d=-q-q^{-1}$. Thus, the trace evaluation
of the element $f_{r-2}$, given by 
\[
tr(f_{r-2})=\Delta_{r-1}=(-1)^{r-1}\frac{q^{r}-q^{-r}}{q-q^{-1}},
\]

vanishes. In fact, the stronger%
\footnote{Stronger in the sense that if any closed network contains $\xygraph{!{0;/r0.5pc/:}[u(2)]!{\xcapv[1.5]@(0)|{r-1}}[d]{\hspace{1em}}*\frm{-}[d(0.6)]!{\xcapv[1.5]@(0)}}$,
then its evaluation vanishes.%
} identity $\xygraph{!{0;/r0.5pc/:}[u(2)]!{\xcapv[1.5]@(0)|{r-1}}[d]{\hspace{1em}}*\frm{-}[d(0.6)]!{\xcapv[1.5]@(0)}}=0$
holds, \cite{kauffman1994temperley}. Furthermore, since $q=e^{i\pi/r}$
the trace of $\xygraph{!{0;/r0.5pc/:}[u(2)]!{\xcapv[1.5]@(0)|{l}}[d]{\hspace{1em}}*\frm{-}[d(0.6)]!{\xcapv[1.5]@(0)}}$
does not vanish for $0\leq l\leq r-2$ and in this range we have
\begin{equation}
\Delta_{l}=\hspace{1em}\xygraph{!{0;/r1.5pc/:}[r]!{\vcap>{l}}[d(0.2)]{\hspace{1em}}*\frm{-}[d(0.2)]!{\vcap-}[r][u(0.4)]!{\xcapv[0.4]@(0)}}=(-1)^{l}\frac{\sin(\pi(l+1)/r)}{\sin(\pi/r)}.\label{eq:loop-value-in-TL-alg}
\end{equation}
Using l'Hôpital's rule we see that the above expression delivers the
loop value as in \prettyref{sec:Spin-Networks}, where $r\rightarrow\infty$
and the $n$-edges are representations of $SU(2)$. 

For generic $d$ we remark that the coefficients of left-right symmetric
terms in the expansion of the projectors are equal. In addition,
one useful expression of the projection is given by $\xygraph{!{0;/r0.5pc/:}[u(2)]!{\xcapv[1.5]@(0)|{n}}[d]{\hspace{1em}}*\frm{-}[d(0.6)]!{\xcapv[1.5]@(0)}}=1_{n}+\mathcal{U}_{n}$
where $\mathcal{U}_{n}$ is a sum of products of the generators of
$T_{n}$, hence, it has strands turning back. From this and the irreducibility
of projectors we obtain the following identity
\[
\xygraph{!{0;/r1pc/:}[u]!{\xcapv[-1]@(0)|{m}}[d(0.25)]{\hspace{1em}\hspace{1em}}*\frm{-}[d(0.25)]!{\xcapv[0.5]@(0)}[u(0.25)]{\;}*\frm{-}[d(0.25)]!{\xcapv[-0.5]@(0)|{m}}[r(0.5)][u(3.5)]!{\xcapv@(0)|{n}}[d(0.5)]!{\xcapv[1.5]@(0)}}=\hspace{1em}\xygraph{!{0;/r1pc/:}[u]!{\xcapv[-1]@(0)|{m}}[d(0.25)]{\hspace{1em}\hspace{1em}}*\frm{-}[d(0.25)]!{\xcapv[1.5]@(0)}[r(0.5)][u(2.5)]!{\xcapv@(0)|{n}}[d(0.5)]!{\xcapv[1.5]@(0)}}
\]
 Now, consider two trivalent vertices $(a,\, d,\, c)$ and $(c,\, d,\, b)$
joint together at the edges $c$ and $d$. An important property of
this object -useful to prove the orthogonality and the Biedenharn-Elliott
identity of $q-6j$-symbols and their relation to tetrahedra- is the
following relation to the projector. Assume that $\Delta_{a}\neq0$
and denote the theta-evaluation of the 3-vertex $(a,\, c,\, d)$ by
$\xygraph{!{0;/r1.5pc/:}!{\vcap|{a}}!{\vcap-|{^{d}}}*{\bullet}!{\xcaph@(0)|{c}*{\bullet}}}$,
then 
\begin{equation}
\xygraph{!{0;/r1.5pc/:}[u(0.5)]!{\hcap|{d}}!{\hcap-|{c}}[u]!{\xcapv@(0)|{a}*{\bullet}}[d]*{\bullet}!{\xcapv@(0)|{b}}}=\delta_{b}^{a}\left(\frac{\xygraph{!{0;/r1.5pc/:}!{\vcap|{b}}!{\vcap-|{^{d}}}*{\bullet}!{\xcaph@(0)|{c}*{\bullet}}}}{\xygraph{!{0;/r1pc/:}[r]!{\vcap}[d(0.2)]{\hspace{1em}}*\frm{-}[d(0.2)]!{\vcap-}[r][u(0.4)]!{\xcapv[0.4]@(0)|{b}}}}\right)\xygraph{!{0;/r1.5pc/:}[u(2)]!{\xcapv[1.75]@(0)|{b}}[d]{\hspace{1em}\hspace{1em}}*\frm{-}[d(0.25)]!{\xcapv[1.75]@(0)}}\label{eq:Edge-with-loop}
\end{equation}
 Hence, it vanishes whenever $a\neq b$.

To conclude with this section, we mention shortly%
\footnote{For the derivation of the following formulas and a more detailed discussion
of the theta-evaluations see \cite{kauffman1994temperley}.%
} the evaluation of the theta-net $\theta(a,\, b,\, c)=\xygraph{!{0;/r1.5pc/:}!{\vcap|{a}}!{\vcap-|{^{c}}}*{\bullet}!{\xcaph@(0)|{b}*{\bullet}}}$
in the Temperley-Lieb algebra for the case $q=e^{i\pi/r}$. First,
notice that the theta-function can be written in terms of three projectors
as follows
\begin{equation}
\theta(a,\, b,\, c)=\xygraph{!{0;/r1.5pc/:}[u(0.25)]!{\vcap[1.2]<{m}}[d(0.2)]{\hspace{1em}}*\frm{-}[d(0.2)]!{\vcap[-1.2]}[r(1.25)][u(0.2)]{\hspace{1em}\;}*\frm{-}[r(0.25)][u(0.2)]!{\vcap[1.2]>{n}}[d(0.4)]!{\vcap[-1.2]}[u(0.2)][r(1.2)]{\hspace{1em}}*\frm{-}[l(2.85)][d(0.2)]!{\vcap[-3]<{^{p}}}[u(0.4)]!{\vcap[3]}}\label{eq:theta-function}
\end{equation}
where 
\[
m=\frac{a+b-c}{2},\; n=\frac{b+c-a}{2},\; p=\frac{c+a-b}{2}.
\]
 Hence, with the help of relation \prettyref{eq:diagram-recursion-relation}
one can find, after a tedious calculation, a general formula for the
evaluation of a trivalent vertex given by
\[
\theta(a,\, b,\, c)=(-1)^{m+n+p}\frac{[m+n+p+1]![n]![m]![p]!}{[n+p]![m+p]![m+n]!}
\]
with $m,n,p$ given above and $[n]=(-1)^{n-1}\Delta_{n-1},\:[n]!=[1][2]\dots[n]$.

An immediate consequence of the above formula is that the evaluation
of $\xygraph{!{0;/r1.5pc/:}[u]!{\xcapv@(0)>{a}}*{\bullet}[lr(0.1)]!{\sbendv@(0)|{^{c}}}[ll]!{\sbendh-@(0)|{b}}}$
where $a,b,c\leq r-1$ but $a+b+c\geq2(r-1)$ vanishes since $[m+n+p+1]=(-1)^{m+n+p}\Delta_{m+n+p}=0$
whenever $m+n+p=r-1$ and hence the quantum factorial vanishes for
higher values of this sum.

\subsection{\label{sub:Recoupling-Theory}Recoupling Theory }

In this section we discuss one of the most important theorems presented
here which allows the proof of the orthogonality relation and the
Biedenharn-Elliott identity of the ($q$-deformed) $6j$-symbols.
We discuss the case where $q=e^{i\pi/r}$.
\begin{defn}
A triple of non-negative integers $(a,b,c)$ is called $q$-admissible
if 
\begin{equation}
\begin{array}{cc}
(i) & a+b+c\equiv0\, mod2\\
(ii) & b+c-a\geq0,\, a+b-c\geq0,\, c+a-b\geq0\\
(iii) & a+b+c\leq2r-4
\end{array}\label{eq:q-admissibility}
\end{equation}

The set of $q$-admissible triples is denoted by $ADM_{q}$.
\end{defn}
Consider the set $\mathcal{T}\Bigl[\begin{array}{cc}
a & b\\
c & d
\end{array}\Bigr]$ of all tangles $T$ of the form
\[
\xygraph{!{0;/r1.5pc/:}!{\sbendv@(0)<{b}}[dr(0.4)]{T}*\frm<0.8pc>{o}[rd(0.4)]!{\sbendv-@(0)<{d}}[l(2.75)]!{\zbendv@(0)<{a}}[r(0.75)][u(0.75)]!{\zbendv-@(0)<{c}}}
\]

We can regard this tangle as a functional on tangles $T'$ dual to
it, cf. \prettyref{sec:Spin-Networks}. In doing so, we obtain an
inner-product $\langle\text{·},\text{·}\rangle:\,\mathcal{T}\Bigl[\begin{array}{cc}
a & b\\
c & d
\end{array}\Bigr]\times\mathcal{T}'\Bigl[\begin{array}{cc}
a & b\\
c & d
\end{array}\Bigr]\rightarrow\mathbb{C}$, which allows us to define the concept of equality of elements in
$\mathcal{T}\Bigl[\begin{array}{cc}
a & b\\
c & d
\end{array}\Bigr]$, i.e. two tangles are equal if they are equal as functionals on the
dual tangles in $\mathcal{T}'\Bigl[\begin{array}{cc}
a & b\\
c & d
\end{array}\Bigr]$, in other words, if the inner-product gives the same result for all
$T'\in\mathcal{T}'\Bigl[\begin{array}{cc}
a & b\\
c & d
\end{array}\Bigr]$.

Now, the addition of tangles makes $\mathcal{T}\Bigl[\begin{array}{cc}
a & b\\
c & d
\end{array}\Bigr]$ into a vector space over $\mathbb{C}$ where the set of tangles of
the form
\[
T_{j}=\xygraph{!{0;/r2pc/:}[d]!{\sbendh-@(0)>{a}}[ul]!{\sbendv@(0)<{b}}*{\bullet}!{\xcaph[1.5]@(0)|{j}}[r(0.5)]*{\bullet}!{\sbendh@(0)>{c}}[dl]!{\sbendv-@(0)<{d}}}\text{ with }(a,b,j),\,(c,d,j)\in ADM_{q}
\]
 is a basis%
\footnote{The linear independence comes from using equation \prettyref{eq:Edge-with-loop}
twice.%
} for $\mathcal{T}\Bigl[\begin{array}{cc}
a & b\\
c & d
\end{array}\Bigr]$. 

Similarly, the tangles of the form
\[
\tilde{T_{i}}=\xygraph{!{0;/r2pc/:}[u(1.55)]!{\sbendv@(0)<{b}}*{\bullet}!{\xcapv[1.5]@(0)|{i}}[u]!{\sbendh@(0)>{c}}[d(3.5)][l(2)]!{\sbendh-@(0)>{a}}*{\bullet}!{\sbendv-@(0)<{d}}}\text{ with }(b,c,i),\,(a,d,i)\in ADM_{q}
\]
form also a basis for $\mathcal{T}\Bigl[\begin{array}{cc}
a & b\\
c & d
\end{array}\Bigr]$.

The recoupling theorem is then a statement about the change of basis:
\begin{thm}
The Recoupling Theorem:

Let $(a,b,j),\,(c,d,j)\in ADM_{q}$, then there exists \uline{unique}
real numbers $\alpha_{i},\,(0\leq i\leq r-2)$, such that
\begin{equation}
\xygraph{!{0;/r2pc/:}[d]!{\sbendh-@(0)>{a}}[ul]!{\sbendv@(0)<{b}}*{\bullet}!{\xcaph[1.5]@(0)|{j}}[r(0.5)]*{\bullet}!{\sbendh@(0)>{c}}[dl]!{\sbendv-@(0)<{d}}}=\sum_{i}\hspace{1em}\alpha_{i}\xygraph{!{0;/r2pc/:}[l][u(1.55)]!{\sbendv@(0)<{b}}*{\bullet}!{\xcapv[1.5]@(0)|{i}}[u]!{\sbendh@(0)>{c}}[d(3.5)][l(2)]!{\sbendh-@(0)>{a}}*{\bullet}!{\sbendv-@(0)<{d}}}\label{eq:Recoupling-Thm}
\end{equation}
 where the sum goes over all non-negative integers $i$ such that
$(a,d,i),\,(b,c,i)\in ADM_{q}$ and all networks are evaluated at
$q=e^{i\pi/r}$. Furthermore, the coefficients $\alpha_{i}=\Bigl\{\begin{array}{ccc}
a & b & i\\
c & d & j
\end{array}\Bigr\}_{q}$ are the quantum q-6j-symbols, \cite{kauffman1994temperley}.  \end{thm}
\begin{rem*}
\textit{(i)} Note that the meaning of this equality is that both sides
are interchangeable in all bracket evaluations having these sums inside
larger networks.\textit{ (ii)} For general value of $q$ the summation
is over all admissible triples, i.e. the third condition in \prettyref{eq:q-admissibility}
is not needed.
\end{rem*}
Next, we give the diagrammatic form of the $q-6j$-symbols in order
to appreciate finally its connection with the tetrahedron which, as
seen in \prettyref{sec:Connection-between-GR-SN}, is an important
object to build invariants such as the combinatorial analogous of
the ``path integral'' \prettyref{eq:path-integral-over-geometries}
over geometries with the exponential of the Hilbert-Einstein action
as integrand, cf. \prettyref{cha:Invariants-of-3-Manifolds}. This
connection is a direct consequence of the recoupling theorem and the
relation \prettyref{eq:Edge-with-loop}. The formula holds under conditions
\textit{(i) and (ii)} in \prettyref{eq:q-admissibility} for generic
$q$. The $q-6j$-symbols defined via the above theorem are given%
\footnote{Notice that Moussouris definition of the symbols in \cite{moussouris1983quantum}
is different. For this reason in section \prettyref{sub:The-Decomposition-Theorem}
the $6j$-symbols are defined without the loop value and with the
theta-net value set to 1, hence, the identity given by the Recoupling
theorem is slightly different, namely, the recoupling coefficients
are given by the loop value and the evaluation of the tetrahedral
network. Because of the different definitions in the literature, extra
care in the use of identities is needed when evaluating spin networks.%
} by
\begin{equation}
\Bigl\{\begin{array}{ccc}
a & b & i\\
c & d & j
\end{array}\Bigr\}_{q}=\frac{\Bigl[\xygraph{!{0;/r1pc/:}[r]!{\vcap>{i}}[d(0.2)]{\hspace{1em}}*\frm{-}[d(0.2)]!{\vcap-}[r][u(0.4)]!{\xcapv[0.4]@(0)}}\Bigr]_{q}}{\Bigl[\xygraph{!{0;/r1.5pc/:}!{\vcap|{b}}!{\vcap-|{^{i}}}*{\bullet}!{\xcaph@(0)|{c}*{\bullet}}}\Bigr]_{q}\Bigl[\xygraph{!{0;/r1.5pc/:}!{\vcap|{i}}!{\vcap-|{^{a}}}*{\bullet}!{\xcaph@(0)|{d}*{\bullet}}}\Bigr]_{q}}\Biggl[\qquad\xygraph{!{0;/r1.5pc/:}!{\sbendh@(0)|{^{b}}}*{\bullet}!{\sbendv@(0)|{^{c}}}[l(2)]*{\bullet}!{\xcaph[2]@(0)|{j}}[r]*{\bullet}[ll]!{\sbendv@(0)|{a}}*{\bullet}!{\sbendh@(0)|{d}}[lu]!{\vcap[1.5]}[r(1.5)]!{\xcapv[2]@(0)|{i}}[d][l(1.5)]!{\vcap[-1.5]}}\Biggr]_{q}\label{eq:q6j-symbol-and-Tetrahedron}
\end{equation}
where $[\dots]_{q}$ denotes the evaluation of the diagram at a given
$q$.

Now, regard the labels $a,b,c,d,j$ as parameters of a function $\bigl\{\dots\bigr\}_{q}:\:\bigl\{0,1,\dots,r-2\bigr\}^{6}\rightarrow\mathbb{R}$
given by the coefficients above, then from a double use of the recoupling
theorem we obtain the orthogonality relation
\begin{equation}
\sum_{i=0}^{r-2}\Bigl\{\begin{array}{ccc}
a & b & i\\
c & d & j
\end{array}\Bigr\}_{q}\Bigl\{\begin{array}{ccc}
d & a & k\\
b & c & i
\end{array}\Bigr\}_{q}=\delta_{j}^{k}.\label{eq:Orthogonality-Relation}
\end{equation}

Finally we present an important identity called the Biedenharn-Elliott
identity, sometimes also called the pentagon identity. Consider the
following diagram
\[
\xygraph{!{0;/r2pc/:}!{\xcapv@(0)<{a}}[u]*{\bullet}[lu]!{\sbendv@(0)|{i}}!{\sbendh@(0)|{j}}[ll][ul]!{\sbendv@(0)<{b}}[ul(0.25)]*{\bullet}!{\sbendh@(0)>{c}}[dr(1.25)]!{\sbendh@(0)>{e}}[dl(0.8)]*{\bullet}[ul]!{\sbendv-@(0)>{d}}}
\]
expressed in terms of 
\[
\xygraph{!{0;/r2pc/:}!{\xcapv@(0)<{a}}[u]*{\bullet}[lu]!{\sbendv@(0)}!{\sbendh@(0)|{k}}[ll][ul]!{\sbendv@(0)<{b}}[rr]!{\sbendh@(0)>{e}}[dl]*{\bullet}[ul(0.75)]!{\sbendv[0.75]@(0)|{l}}[lu]*{\bullet}!{\sbendh[0.75]@(0)>{d}}[l(1.75)][d(0.25)]!{\sbendv[0.75]@(0)<{c}}}
\]
There are two ways of doing this, one by two consecutive applications
of the recoupling theorem, the other one by three applications. Since
both ways must give the same result, we obtain the following relation
\begin{equation}
\sum_{m=0}^{r-2}\Bigl\{\begin{array}{ccc}
a & i & m\\
d & e & j
\end{array}\Bigr\}_{q}\Bigl\{\begin{array}{ccc}
b & c & l\\
d & m & i
\end{array}\Bigr\}_{q}\Bigl\{\begin{array}{ccc}
b & l & k\\
e & a & m
\end{array}\Bigr\}_{q}=\Bigl\{\begin{array}{ccc}
b & c & k\\
j & a & i
\end{array}\Bigr\}_{q}\Bigl\{\begin{array}{ccc}
k & e & l\\
d & c & j
\end{array}\Bigr\}_{q}.\label{eq:Biedenharn-Elliott}
\end{equation}
 These two last properties are very important for the 3-manifold invariants
discussed in \prettyref{cha:Invariants-of-3-Manifolds} and have several
consequences as seen earlier. Before coming to the mentioned invariants,
we will first discuss briefly other concepts related to the invariants
of 3-manifolds, namely, the topological quantum field theory.

\section{Atiyah's Axiomatic Topological Quantum Field Theory\label{sec:TQFT}}

In \cite{atiyah1988topological} M. Atiyah gave an axiomatic approach
to topological quantum field theory (TQFT), which will be discussed
briefly in this section. The TQFT described in this section is only
defined for manifolds with fixed dimension using the concept of cobordisms%
\footnote{A cobordism $W=(M;F_{+},F_{-};i_{+},i_{-})$ between $d$-dimensional
manifolds $F_{+}$ and $F_{-}$ is a $(d+1)$-dimensional compact
manifold $M$, such that $i_{+}:\, F_{+}\rightarrow\partial M$ and
$i_{\lyxmathsym{\textminus}}:\, F_{\lyxmathsym{\textminus}}\rightarrow\partial M$
are embeddings with $\partial M=i_{+}(F_{+})\cup i_{-}(F_{-})$ and
$i_{+}(F_{+})\cap i_{\lyxmathsym{\textminus}}(F_{\lyxmathsym{\textminus}})=\emptyset$.
If we regard closed manifolds as objects in a category, then the cobordisms
can be considered as the morphisms of this category, the \textbf{category
of cobordisms}, \cite{turaev1992state}. %
} as morphisms ``propagating'' a manifold to another manifold of
the same dimension but possibly with a different topology. For instance,
the figure below shows a cobordism $W=(M;i_{+},i_{-})$ with $i_{+}:\, S^{1}\rightarrow\partial M$
and $i_{-}:\, S^{1}\cup S^{1}\rightarrow\partial M$, where $\cup$
denotes the disjoint union of two copies of $S^{1}$.

\selectlanguage{british}%
\bigskip{}

\begin{center}
\xy 
(0,-6)*\ellipse(3,1){.};
(3,3)*\ellipse(3,1){-};
(0,-6)*\ellipse(3,1)__,=:a(180){-};
(-3,3)*\ellipse(3,1){-};              %ellipses one at the bottom two at the top
(-3,6)*{}="1";                        %connection between the points of the ellipses on the top and bottom
(3,6)*{}="2";
"1";"2" **\crv{(-3,-3) & (3,-3)};
(-9,6)*{}="A2";
(9,6)*{}="B2";
(-3,-6)*{}="A"; 
(3,-6)*{}="B";
(-3,-12)*{}="A1";
(3,-12)*{}="B1";
"A";"A1" **\dir{-};
"B";"B1" **\dir{-};
"B2";"B" **\crv{(6,-5) & (3,-3)};
"A2";"A" **\crv{(-6,-5) & (-3,-3)}; 
\endxy
\par\end{center}

\bigskip{}

\selectlanguage{english}%
\begin{defn}
A \textbf{topological quantum field theory} in dimension $d$ defined
over a ring $\Lambda$ consists of the following data:\end{defn}
\begin{itemize}
\item To each oriented closed smooth $d$-dimensional manifold $\Sigma$
we associate a finitely generated $\Lambda$-module $Z(\Sigma)$.
\item To each oriented smooth $(d+1)$-dimensional manifold $M$ with boundary
we associate an element $Z(M)\in Z(\partial M)$.
\end{itemize}
subject to the following axioms%
\footnote{These axioms, excluding the first, are taken as in \cite{atiyah1990geometry}.%
}
\begin{enumerate}
\item $Z$ is a functor from the category of compact oriented smooth manifolds,
with orientation preserving diffeomorphisms as arrows, to the category
of $\Lambda$-modules. Another way of defining this functor is from
the category of cobordisms to the category of $\Lambda$-modules,
as in \prettyref{cha:Invariants-of-3-Manifolds}.
\item Involutory: $Z(\Sigma^{*})=Z(\Sigma)^{*}$ , where $\Sigma^{*}$ denotes
$\Sigma$ with opposite orientation and $Z(\Sigma)^{*}$ is the dual
space.
\item Multiplicativity: $Z(\Sigma_{1}\cup\Sigma_{2})=Z(\Sigma_{1})\otimes{}Z(\Sigma_{2})$
where $\cup$ is the disjoint union.
\item Associativity: For a composite cobordism $M=M_{1}\cup{}_{\Sigma_{3}}M_{2}$,
see figure below, we have 
\begin{equation}
Z(M)=Z(M_{2})Z(M_{1})\in{}Hom(Z(\Sigma_{1}),Z(\Sigma_{2}))
\end{equation}
 where $\cup_{\Sigma_{3}}$ denotes the union of two manifolds with
a common component $\Sigma_{3}$ of their boundary, i.e. $\Sigma_{3}\subseteq\partial M_{1}$
and $\Sigma_{3}^{*}\subseteq\partial M_{2}$.
\item Non-triviality axioms%
\footnote{Note that if $\Sigma=\emptyset$, then the vector space associated
to it is idempotent, i.e. $Z(\emptyset)=Z(\emptyset\cup\emptyset)=Z(\emptyset)\otimes Z(\emptyset)$,
thus it is zero or canonically isomorphic to $\Lambda$. For similar
reasons, if $M=\emptyset$ we have for the $(d+1)$-dimensional empty
manifold $Z(\emptyset_{d+1})=1$.%
}: $Z(\emptyset)=\Lambda$ and $Z(\Sigma\times{}\mathbb{I})=id_{Z(\Sigma)}$
is the identity endomorphism of $Z(\Sigma)$.
\end{enumerate}
The first axiom states that if $f:\,\Sigma\rightarrow\Sigma'$ is
an orientation preserving diffeomorphism, i.e. $f\in Diff^{+}(\Sigma,\Sigma')$,
then $f$ induces an isomorphism $Z(f):\, Z(\Sigma)\rightarrow Z(\Sigma')$
and for $g:\,\Sigma'\rightarrow\Sigma'',\: Z(g\circ f)=Z(g)\circ Z(f)$.
Moreover, if $f$ extends to an orientation preserving diffeomorphism
$f^{*}$ from $M$ to $M'$, then $Z(f^{*}):\, Z(M)\mapsto Z(M')$.
Notice that $Z(f^{*})$ maps the element $Z(M)$ associated to $M$
to an element $Z(M')$ associated to $M'$. Another way of looking
at this is by regarding the category of cobordisms with objects closed
manifolds and morphisms cobordisms and the category of $\Lambda$-modules
with homomorphisms. In this case the cobordism $M$ is associated
to a homomorphism $Z(M)\in Z(\partial M)$. Notice that if $\partial M=\Sigma_{1}\cup\Sigma_{2}$
the morphism is, in fact, a homomorphism from $Z(\Sigma_{1})$ to
$Z(\Sigma_{2})$. For instance, if $\Sigma_{1}=\Sigma_{2}$, then
$Z(M)\in Z(\Sigma_{1})$ and $Z(M)$ is an endomorphism by the action
of this element on the module.

When $\Lambda$ is a field, $Z(\Sigma)$ and $Z(\Sigma)^{*}$ are
dual vector spaces. This case is the most important for physical examples
with $\Lambda=\mathbb{C},\,\mathbb{R}$ and we will assume this from
now on.

Now, the third axiom states that if $\partial M_{1}=\Sigma_{1}\cup\Sigma_{3},\,\partial M_{2}=\Sigma_{2}\cup\Sigma_{3}^{*}$
and $M=M_{1}\cup_{\Sigma_{3}}M_{2}$, as shown below, then we require
the natural pairing 
\[
\langle\text{·},\text{·}\rangle:\, Z(\Sigma_{1})\otimes Z(\Sigma_{3})\otimes Z(\Sigma_{3})^{*}\otimes Z(\Sigma_{2})\rightarrow Z(\Sigma_{1})\otimes Z(\Sigma_{2})
\]
 to be defined by 

\begin{equation}
Z(M)=\langle Z(M_{1}),\, Z(M_{2})\rangle\label{eq:Pairing-in-TQFT}
\end{equation}
where $Z(M_{1})\in Z(\Sigma_{1})\otimes Z(\Sigma_{3})\text{ and }Z(M_{2})\in Z(\Sigma_{3})^{*}\otimes Z(\Sigma_{2})$,
hence, $Z(M)\in Z(\Sigma_{1})\otimes Z(\Sigma_{2})$.

\medskip{}

\begin{center}
\xy
(-9,4.5)*\ellipse(9,3){-};
(9,4.5)*\ellipse(9,3){-};
(0,-5.125)*\ellipse(1.5,7.75){.};                %Boundaries
(0,-2.4)*{}="B1";
(0,-18)*{}="B2";
"B1";"B2" **\crv{(-1.95,-2.4)&(-1.95,-18)};
(-9,9)*{}="X1"; 
(9,9)*{}="X2"; 
"X1";"X2" **\crv{(-9,-6) & (9,-6)};
(-27,9)*{}="X1";
(27,9)*{}="X2";
"X1";"X2" **\crv{(-27,-27) & (27,-27)};             %Points on the boundary to connect
(0,-20.5)*{\Sigma_3};
(-18,14)*{\Sigma_1};
(18,14)*{\Sigma_2};
(-15,0)*{M_1};
(15,0)*{M_2};
\endxy
\par\end{center}

\bigskip{}

Thus if $\Sigma_{1}=\Sigma_{2}=\emptyset$, i.e. the $(d+1)$-dimensional
manifold $M$ is closed, the pairing gives an element of $\Lambda$
which is independent of the choice of $\Sigma_{3}$. This means that
the numerical invariants of closed manifolds are independent of their
decomposition $M=M_{1}\cup_{\Sigma_{3}}M_{2}$ and can be computed
in term of this decomposition via the above relation. Note that when
$\Sigma_{3}=\emptyset$, i.e. $M$ is the disjoint union of $M_{1}$
and $M_{2}$, the pairing \prettyref{eq:Pairing-in-TQFT} reduces
to 
\[
Z(M)=Z(M_{1})\otimes Z(M_{2})
\]

This means that disjoint unions of $(d+1)$-manifolds are translated
into tensor products of $\Lambda$-modules respecting the associations
made for the distinguished elements $Z(M_{i})\,(i=1,2)$ to each component
and extending it naturally to their tensor product. We have been working
with this concept from the beginning on, associating a point in the
plane to representations of $SU(2)$, in \prettyref{sec:Spin-Networks}
or Hopf algebras in Section \prettyref{sub:Spherical-categories}.
In fact, the second and third axioms are used to view $Z(M_{1})$
and $Z(M_{2})$ as homomorphisms $Z(\Sigma_{1})\rightarrow Z(\Sigma_{3})$
and $Z(\Sigma_{3})\rightarrow Z(\Sigma_{2})$ respectively, \cite{atiyah1990geometry},
for instance, by $Z(M_{1})\triangleright Z(\Sigma_{1})=\Lambda\otimes Z(\Sigma_{3})\cong Z(\Sigma_{3})$
and $Z(M_{2})\triangleright Z(\Sigma_{3})\cong Z(\Sigma_{2})$. Hence,
\prettyref{eq:Pairing-in-TQFT} means that $Z(M):\, Z(\Sigma_{1})\rightarrow Z(\Sigma_{2})$
is transitive when cobordisms are composed. This corresponds to the
previous assignment of the 3-valent vertex to a homomorphism between
representations of (Hopf) algebras.

Now, consider $f\in Diff^{+}(\Sigma,\Sigma)$ and identify opposite
ends of $\Sigma\times\mathbb{I}$ by $f$, such that we obtain a manifold
$\Sigma_{f}$ with 
\[
Z(\Sigma_{f})=Tr(Z_{f})
\]
where $Z_{f}:\, Z(\Sigma)\rightarrow Z(\Sigma)$ is an induced automorphism.
For example, we can construct $\Sigma\times S^{1}$ by identifying
the opposite ends of $\Sigma\times\mathbb{I}$ and we obtain, \cite{atiyah1988topological},
\[
Z(\Sigma\times S^{1})=Tr(id_{Z(\Sigma\times\mathbb{I})})=\text{dim}Z(\Sigma\times\mathbb{I}).
\]

Compare this result with the loop value in \prettyref{sec:Spin-Networks}.

To finalize this section we describe shortly the physical interpretation
of this theory. It is important to note, however, that there is no
relation between the invariants $Z(M)$ and $Z(M^{*})$ for closed
$(d+1)$-manifolds given by the axioms. We can, however, consider
the additional assumption that the vector spaces $Z(\Sigma)$ posses
a non-degenerate Hermitian structure relative to some conjugation
on $\Lambda$, which gives an isomorphism $Z(\Sigma^{*})\rightarrow\overline{Z(\Sigma)}$,
where $\overline{Z(\Sigma)}$ denotes $Z(\Sigma)$ with the conjugate
action of $\Lambda$. This structure lets us consider a further Hermitian
axiom, 
\[
Z(M^{*})=\overline{Z(M)}
\]
which means that $Z(M^{*})$, regarded as a linear transformation
between Hermitian vector spaces, is the adjoint of $Z(M)$. Hence,
the numerical invariants of a closed manifold are sensible to changes
of orientation, unless their value is real. Furthermore, with the
Hermitian structure is possible to form a closed manifold $M\cup_{\Sigma}M^{*}$
from a manifold $M$ with $\partial M=\Sigma$, such that
\[
Z(M\cup_{\Sigma}M^{*})=|Z(M)|^{2}
\]
 where the r.h.s. is the norm in the Hermitian metric, \cite{atiyah1988topological}.

\subsection*{Physical interpretation of the axioms}

In these axioms $\Sigma$ is meant to indicate the physical space
and the extra dimension in $\Sigma\times\mathbb{I}$ is the ``imaginary''
time. Then one can think of the space $Z(\Sigma)$ as the Hilbert
space of the theory on $\Sigma$. The endomorphism $End(Z(\Sigma))$
given by $Z(\Sigma\times\mathbb{I})$ should be the imaginary time
evolution operator $e^{tH}\text{ where }t\in{}\mathbb{I}$, but the
second non-triviality axiom does not allow any dynamics, since $Z(\Sigma\times\mathbb{I})=id_{Z(\Sigma)}$
implies $H=0$. There is, however, a ``topological propagation''
across a non-trivial cobordism $M$ which changes the topology of
$\Sigma$. Then, for a closed $(d+1)$-manifold $M$, the invariant
$Z(M)$ is the partition function%
\footnote{If $\partial M=\Sigma$, the distinguished vector $Z(M)\in Z(\Sigma)$
is interpreted as the vacuum state defined by the topology of $M$,
\cite{atiyah1988topological}. %
} given by some Feynman integral, i.e. with a special Lagrangian that
gives rise to a topological invariant partition function. Relativistic
invariance assures that the numerical invariants $Z(M)$ are independent
of the decomposition of $M$, i.e. of the time variable chosen to
slice the cobordism, \cite{atiyah1990geometry}. 

The importance of this broad theory will be seen in the next chapter,
where invariants of 3-manifolds and their calculation via $q-6j$-symbols
are described.

\chapter{Invariants of 3-Manifolds\label{cha:Invariants-of-3-Manifolds} }

In this chapter we discuss the Turaev-Viro invariants of 3-manifolds
and their relation to the concept of spherical categories. Most of
the ideas are taken from \cite{turaev1992state,barrett1996invariants,kauffman1994temperley}.

In the first section, the invariant of a manifold is defined as a
state sum based on quantum $6j$-symbols, which are associated with
the quantized universal enveloping algebra $U_{q}(\mathfrak{sl}(2,\mathbb{C}))$.
To define a state sum on a triangulation $X$ of a compact 3-manifold
$M$ assume that there are colorings of $X$ associating elements
of the set of colors $\{0,1/2,1,...,(r-2)/2\}$ with edges of the
triangulation. This naturally leads to a one-to-one association of
colored 3-simplexes of $X$ with $q-6j$-symbols, which are multiplied%
\footnote{More precisely, the \textquotedbl{}multiplication\textquotedbl{} is
in fact a tensor contraction.%
} over all simplexes of the triangulation. The resulting weighted products
are then summed over all colorings of $X$ which are, in a sense defined
below, admissible. These concepts lead to a 3-dimensional non-oriented
topological quantum field theory where each closed surface $F$ is
associated with a finite-dimensional vector space $Z(F)$ over $\mathbb{C}$,
as in the previous section. However, to defi{}ne this vector space
we have to fi{}x a triangulation of $F$ and show a posteriori that
$Z(F)$ does not depend on the choice of triangulation, \cite{turaev1992state}. 

Although the state sums are computed on a triangulation of the manifold,
they are independent of the choice of triangulation since some transformations
of polyhedra, called Alexander moves, allow us to relate combinatorial
equivalent triangulations leaving the evaluation of the state sum
invariant. The number of transformations is infinite, however, in
the case of triangulations of manifolds one can pass to the dual complex,
called the cell subdivision. This dualisation, described in the second
section, transforms the Alexander moves into certain operations on
cell complexes, which can be presented as compositions of certain
finite set of local moves. In a 3-manifold there are three such moves
called the Matveev-Piergallini moves. This fact simplifies the task
of checking the invariance of the state sums since there are only
three identities to be verified. It turns out that these identities
follow directly from the basic properties of the $q-6j$-symbols.

After redefining the state sum for the simple 2-polyhedra forming
the cell subdivision, we give an informal identification of the constituent
terms of the state sum with the diagrammatic language presented in
\prettyref{sec:Some-More-Diagrammatics}. This identification allows
us then to give an explicit expression of the invariant for the case
when the objects used to construct it are representations of $U_{q}(\mathfrak{sl}_{2})$
with $q$ a root of unity.

Finally, a more general invariant of 3-manifolds is given briefly
in \prettyref{sec:Invariants-from-Spherical}, where the only structure
assumed is the one described in \prettyref{sub:Spherical-categories}.
Hence, the Turaev-Viro invariant defined in the next section is a
special case which, in fact, satisfies two extra conditions. First,
without going into detail, this invariant is defined for unoriented
manifolds and second, that there exist a TQFT associated to the invariant
given by the fact that each self-dual simple object of the category
involved in the construction of the invariant is orthogonal%
\footnote{An self-dual simple object $a$ is called orthogonal, if for its isomorphism
$\phi:\, a\rightarrow\hat{a}$ we have $\phi=\hat{\phi}$.%
}, \cite{barrett1996invariants}.

\section{State Sum Invariants\label{sec:State-Sum-Invariants}}

First, the initial data and the conditions on it needed to define
an invariant of 3-manifolds are given. Then, we proceed with the definition
of the state sum models for closed 3-manifolds and its relation to
topological quantum field theory is discussed.

For the initial data consider a commutative ring $K$ with unity and
denote by $K^{*}$ the group of invertible elements of $K$. The data
consists of five objects besides the ring $K$:
\begin{itemize}
\item A finite%
\footnote{Notice here the importance of $U_{q}(\mathfrak{sl}_{2})$ at a root
of unity to make this color set invariant.%
} set $I$ of ``colors''.
\item A function $f:\, I\rightarrow K^{*};\: i\mapsto f(i)=w_{i}$.
\item A distinguished element $w\in K^{*}$. 
\item A set $adm$ of unordered triples of elements of $I$, $adm\subset I^{3}$,
for which there are no further conditions imposed. The triples belonging
to $adm$ are called admissible.
\item A set of ordered 6-tuples $(i,j,k,l,m,n)\in I^{6}$ which are admissible,
meaning that the unordered triples $(i,j,k),\,(k,l,m),\,(m,n,i),\,(j,l,n)$
are admissible. Furthermore, we assume that each of these 6-tuple
is associated with an element of $K$ called the symbol and denoted
by 
\[
\left|\begin{array}{ccc}
i & j & k\\
l & m & n
\end{array}\right|\in K
\]
These symbols are assumed to have the same symmetries as the usual
$6j$-symbols in previous sections. From these symmetries we conclude
that by permutation and interchange of the upper and lower arguments
of any two columns respectively one can obtain different 6-tuples
which correspond to symbols with the same value. Denote this common
value of the symbols by $|T|$.
\end{itemize}
The initial data is assumed to follow four conditions. The first two
of them axiomatise the orthogonality and the Biedenharn-Elliot identities
for $q-6j$-symbols.

The data satisfies condition (I) if for any six elements $j_{1},j_{2},j_{3},j_{4},j_{5},j_{6}$
in $I$ such that $(j_{1},j_{3},j_{4})$, $(j_{1},j_{3},j_{6})$,
\foreignlanguage{british}{$(j_{2},j_{4},j_{5})$} and $(j_{2},j_{5},j_{6})$
are admissible we have
\begin{equation}
\sum_{j}w_{j}^{2}w_{j_{4}}^{2}\left|\begin{array}{ccc}
j_{2} & j_{1} & j\\
j_{3} & j_{5} & j_{4}
\end{array}\right|\left|\begin{array}{ccc}
j_{3} & j_{1} & j_{6}\\
j_{2} & j_{5} & j
\end{array}\right|=\delta_{j_{4},j_{6}}.\label{eq:Cond-I}
\end{equation}
where we sum up over $j$ such that the symbols involved in the sum
are defined, i.e. the 6-tuples in the sum are admissible.

The data satisfies the condition (II) if for any pair of admissible
6-tuples $(j_{23},a,e,j_{1},f,b)$ and $(j_{3},j_{2},j_{23},b,f,c)$
the following relation holds
\begin{equation}
\sum_{j}w_{j}^{2}\left|\begin{array}{ccc}
j_{2} & a & j\\
j_{1} & c & b
\end{array}\right|\left|\begin{array}{ccc}
j_{3} & j & e\\
j_{1} & f & c
\end{array}\right|\left|\begin{array}{ccc}
j_{3} & j_{2} & j_{23}\\
a & e & j
\end{array}\right|=\left|\begin{array}{ccc}
j_{23} & a & e\\
j_{1} & f & b
\end{array}\right|\left|\begin{array}{ccc}
j_{3} & j_{2} & j_{23}\\
b & f & c
\end{array}\right|.\label{eq:Cond-II}
\end{equation}

The condition (III) is satisfied if for any $j\in I$ we have
\begin{equation}
w^{2}=w_{j}^{-2}\sum_{k,l:(j,k,l)\in adm}w_{k}^{2}w_{l}^{2}.\label{eq:Cond-III}
\end{equation}
Finally, the initial data is said to be irreducible, if for any $j,k$
in $I$ there exists a sequence $l_{1},l_{2},\dots,l_{n}$ with $l_{1}=j,l_{n}=k$
such that the triple $(l_{i},l_{i+1},l_{i+2})$ is admissible for
any $i=1,...,n-2$. If we have irreducible initial data satisfying
condition (I), then the r.h.s. of \prettyref{eq:Cond-III} is independent
of $j\in I$, i.e. condition (III) is automatically satisfied. 

Now, consider a tetrahedron with edges labelled by elements of the
set $I$. Such a 3-simplex will be called a colored tetrahedron and
is said to be admissible if for any of its 2-simplexes $A$ the labels,
or colors, of the three edges in $A$ are in $adm$. From this we
can understand geometrically the notion of an admissible 6-tuple.
As mentioned before, the symbols corresponding to the admissible 6-tuples
in the initial data are regarded from a geometrical point of view
as colored tetrahedra. Thus, in this case, admissibility means the
condition for the existence of a tetrahedron with positive volume,
as in \prettyref{sec:Connection-between-GR-SN}. We stress here the
fact that each admissible colored tetrahedron corresponds to \textit{a
set} of admissible 6-tuples. There are 24 admissible 6-tuples for
a given tetrahedron $T$, which may be obtained from each other by
the obvious action of the symmetry group $S_{4}$ of $T$.

We now proceed to discuss the state model for \textbf{closed} 3-manifolds,
which leads to an invariant of the manifold with respect to triangulations.
Consider a closed triangulated 3-manifold $M$. Let $M$ have $b$
edges denoted by $E_{1},E_{2},\dots,E_{b}$. A coloring of $M$ is
defined to be an arbitrary mapping $\phi:\,\{E_{1},E_{2},\dots,E_{b}\}\rightarrow I$,
i.e. we label all edges of a given triangulation of $M$. Denote the
admissible colorings%
\footnote{As before, admissible means that for any 2-simplex $A$ of $M$ the
colors of its three edges form an admissible triple.%
} of $M$ by $adm(M)$. It is obvious that each $\phi\in adm(M)$ induces
an admissible coloring of each tetrahedra $T_{i}$ of $M$, denoted
by $T_{i}^{\phi}$.

As we saw in section \prettyref{sub:The-3nj-symbols} the $3n-j$-symbols
can be expressed as a product of $6j$-symbols. The definition of
a state $|M|_{\phi}$ of the manifold $M$ is defined in the same
fashion. For a given coloring $\phi\in adm(M)$ we set
\begin{equation}
|M|_{\phi}=w^{-2a}\prod_{r=1}^{b}w_{\phi(E_{r})}^{2}\prod_{t=1}^{d}|T_{t}^{\phi}|\,;\;|M|_{\phi}\in K\label{eq:state-of-M}
\end{equation}
 where $a$ is the number of vertices and $d$ the number of tetrahedra
in $M$. Then, the invariant $|M|$ of $M$ is defined as the sum
over all admissible colorings for a given triangulation:
\begin{equation}
|M|=\sum_{\phi\in adm(M)}|M|_{\phi}.\label{eq:Def-of-invariant-of-M}
\end{equation}

Next we present a theorem proved by Turaev and Viro, \cite{turaev1992state},
giving a scheme to define topological invariants of 3-manifolds. In
principle it is defined as \prettyref{eq:Def-of-invariant-of-M},
however, to realize the invariant one needs concrete initial data.
\begin{thm}
\label{thm:Invariance-wrt-triangulations}If the initial data satisfies
the conditions (I), (II) and (III), then |M| does not depend on the
choice of triangulation of M.
\end{thm}
The proof of this theorem can be found in \cite[Sec. 5]{turaev1992state}.

Now, let us consider the more general case where $M$ is a \textbf{compact}
triangulated 3-manifold. Suppose that $e$ of the $a$ vertices and
the first $f$ of the $b$ edges of $M$ lay on the boundary $\partial M$.
All the same concepts as above apply here as well, so that a coloring
of $\partial M$ means an arbitrary mapping $\alpha:\,\{E_{1},E_{2},\dots,E_{f}\}\rightarrow I$.
The formula for a state $|M|_{\phi}$, however, has to be modified
as follows to account for the boundary. For any $\phi\in adm(M)$
define
\begin{equation}
|M|_{\phi}=w^{-2a+e}\prod_{r=1}^{f}w_{\phi(E_{r})}\prod_{s=f+1}^{b}w_{\phi(E_{s})}^{2}\prod_{t=1}^{d}|T_{t}^{\phi}|\in K.\label{eq:phi-summand-of-invariant}
\end{equation}

For $\alpha\in adm(\partial M)$, denote by $adm(\alpha,M)\subseteq adm(M)$
the set of all colorings $\phi\in adm(M)$ of $M$ which extend $\alpha$,
i.e. which have $\alpha$ as a restriction of $\phi$ on $\partial M$.
Define
\[
\Omega_{M}(\alpha)=\sum_{\phi\in adm(\alpha,M)}|M|_{\phi}.
\]
 Hence, for an admissible coloring $\alpha$ of $\partial M$ the
invariant is $\lyxmathsym{\textgreek{W}}(\alpha)$ and it is dependent
on the coloring $\alpha$. 
\begin{thm}
\label{thm:generalization-invariance}If the initial data satisfies
the conditions (I), (II) and (III), then for any compact 3-manifold
$M$ with triangulated boundary and any admissible coloring $\alpha$
of $\partial M$, all extensions of the triangulation of $\partial M$
to $M$ yield the same $\lyxmathsym{\textgreek{W}}_{M}(\alpha)$,
\cite{turaev1992state}.
\end{thm}
This is a generalization of \prettyref{thm:Invariance-wrt-triangulations}
and means that for a given triangulation and coloring of the boundary
of a compact 3-manifold there is a state sum which is invariant under
Alexander moves on the extensions of the triangulation of the boundary
to all the manifold, i.e. on simplexes not lying on the boundary.
In this case, the state sums are called relative invariants. 

As mentioned above, this type of initial data relates to a topological
quantum field theory (TQFT). In what follows we discuss these relations
by describing the role that the invariants defined above take in the
theory and how the modules associated with the boundary of a cobordism
arise. The construction of a functor, which defines the TQFT is also
discussed briefly following \cite{turaev1992state}.

Consider a triangulated closed surface $F$. Since there is an element
$j$ of the set $I$ attached to each edge of the triangulation and
a function $f:\, I\rightarrow K^{*}$  with $j\mapsto w_{j}$ for
each $j\in I$, each admissible coloring gives a set of elements
of the group $K^{*}\subseteq K$ which corresponds by multiplication
to an element of $K$. Hence, each triangulated closed surface $F$
defines a $K$-module $C(F)$, which is the module freely generated
over $K$ by admissible colorings of $F$, i.e. each coloring gives
an element of the set generating $C(F)$. If we equip $C(F)$ with
a scalar product $C(F)\times C(F)\rightarrow K$ we can make the set
of admissible colorings an orthonormal basis of $C(F)$. According
to the convention that there exists exactly one map $\emptyset\rightarrow I$,
we set $C(F)=K$ if $F=\emptyset$, cf. \prettyref{sec:TQFT}.

Consider a cobordism $W=(M;F_{+},F_{-};i_{+},i_{-})$ between triangulated
surfaces $F_{+}$ and $F_{-}$ and define a homomorphism $\Phi_{W}:\, C(F_{+})\rightarrow C(F_{-})$
by
\begin{equation}
\Phi_{W}(\alpha)=\sum_{\beta\in adm(F_{-})}\Omega_{M}(i_{+}(\alpha)\cup i_{-}(\beta))\beta\label{eq:homomorph-in-cobordism}
\end{equation}
 where $\alpha\in adm(F_{+})$ and $i_{+}(\alpha)\cup i_{-}(\beta)\in adm(\partial M)$
is the coloring determined by $\alpha$ and $\beta$. From the above
discussion regarding the construction of $C(F)$, we can regard $\Phi_{W}$
as a homomorphism having as matrix elements $\Omega_{M}(i_{+}(\alpha)\cup i_{-}(\beta))$
with respect to the natural bases of $C(F_{\pm})$. For a closed%
\footnote{Closed manifolds can be considered as cobordisms between empty manifolds.%
} $M$, $\Phi_{W}$ acts in $K$ as multiplication by $|M|$, while
for a compact $M$, where the cobordism is $W=(M;id:\partial M\rightarrow\partial M,\emptyset\rightarrow\partial M)$,
the $\Omega_{M}(\alpha)$ are the matrix elements for $\Phi_{W}$.
Hence, the invariants defined above give the homomorphisms defined
in \prettyref{eq:homomorph-in-cobordism}.

Thus, from theorem \prettyref{thm:generalization-invariance} one
can conclude that for any cobordism $W=(M;i_{+},i_{-})$ between triangulated
surfaces, the homomorphism $\Phi_{W}$ does not depend on the extension
of triangulations of the surfaces to the triangulation of $M$.

Since each cobordism $W$ between surfaces $F_{+}$ and $F_{-}$ is
a morphism $F_{+}\rightarrow F_{-}$ of a category with objects closed
manifolds, the composition of cobordisms $W_{1}=(M_{1};i_{1}:F_{1}\rightarrow\partial M,i_{2}:F_{2}\rightarrow\partial M)$
and $W_{2}=(M_{2};j_{2}:F_{2}\rightarrow\partial M,j_{3}:F_{3}\rightarrow\partial M)$
is again a cobordism $W_{2}\circ W_{1}=(M_{1}\cup_{F_{2}}M_{2};i_{1},j_{3})$
obtained by gluing $M_{1}$ and $M_{2}$ along $F_{2}$. From this,
it is straightforward%
\footnote{For a detailed discussion see \cite{turaev1992state}.%
} to conclude that $\Phi_{W_{2}\circ W_{1}}=\Phi_{W_{2}}\circ\Phi_{W_{1}}$
which describes the multiplicativity of invariants.

We have now the ingredients to construct the mentioned topological
3-dimensional QFT related to the initial data. For this, notice that
the association $F\mapsto C(F),W\mapsto\Phi_{W}$ is not a functor
since the induced homomorphism for the unit cobordism, which is simply
the cylinder $F\times[0,1]$, is not always the identity. In \cite{turaev1992state}
Turaev and Viro constructed a functor by building the quotient $Q(F)=C(F)/Ker\Phi_{id_{F}}$,
where%
\footnote{Here $i_{t}:F\rightarrow\partial(F\times[0,1])$ is defined by $i_{t}(x)=(x,t)$.%
} $id_{F}=(F\times[0,1];i_{0},i_{1})$. As a consequence of the multiplicativity
of the invariants, $\Phi_{W}:C(F_{+})\rightarrow C(F_{-})$ induces
a $K$-linear homomorphism $\Psi_{W}:Q(F_{+})\rightarrow Q(F_{-})$
which is also multiplicative and satisfies $\Psi_{id_{F}}=id_{Q(F)}$.
Hence, $F\mapsto Q(F),W\mapsto\Psi_{W}$ is a functor from the category
of cobordisms of triangulated 2-manifolds to the category of $K$-modules. 

Since for any two triangulations of $F$ there exists a triangulation
of the cylinder $F\times[0,1]$ which coincides on $F\times0$ and
$F\times1$ with these given triangulations, an isomorphism between
the spaces $Q(F)$ defined via the triangulations of $F$ is determined
completely. Hence, $Q(F)$ does not depend on the triangulations up
to isomorphism. In other words, each triangulation of $F$ gives a
$K$-module isomorphic to another module $Q'(F)$ given by a different
triangulation of $F$. Moreover, the isomorphism does not depend on
the triangulation of $F\times[0,1]$ either, so all modules $Q(F)$,
for a given surface $F$, can be identified by this isomorphism. From
this, we can conclude that the theory discussed here can be generalized
even further to a functor from the category of cobordisms of topological%
\footnote{The distinction here from the last sentence in the previous paragraph
regarding the functoriality of $\Psi$ is that the surfaces are non-triangulated. %
} surfaces to the category of $K$-modules. This theory is then called
$(2+1)$-dimensional TQFT. 

In the next section the Matveev-Piergallini moves are introduced as
well as the diagrammatic correspondence up to a normalization factor
to the invariants described in this section.

\section{Moves on Triangulations, Simple 2-Polyhedra and TL-Recoupling Theory\label{sec:Moves-on-Triangulations}}

To understand the independence of the constructions in the previous
section we have to study some concepts on simplicial complexes and
transformations of triangulations. We start with some basic definitions,
and continue then with the dualisation of the triangulations of 3-manifolds
and the Alexander moves to obtain simple 2-polyhedra and the Matveev-Piergallini
moves. Finally, we give the diagrammatic form of the invariants, which
is best understood in the dual version of the theory.
\begin{defn}
The \textbf{join} $X*Y$ of spaces $X$ and $Y$ is the quotient space
of $X\times Y\times[0,1]$ obtained by contraction of subsets $pt\times Y\times0$
and $X\times pt\times1$. 
\end{defn}
This is a formal definition but a join $X*Y$ can be regarded as the
union of segments joining $X$ and $Y$ such that any pair of segments
intersect at most at their end points. The join of two simplexes is
again a simplex, thus, triangulations of the spaces $X$ and $Y$
define a triangulation of their join, \cite{viro1992moves}.
\begin{defn}
The \textbf{link} $lk(\sigma,K)$ of a simplex $\sigma$ in a simplicial
complex $K$ is defined%
\footnote{An equivalent definition would be that the link of a simplex $\sigma$
is the union of all closed simplexes contained in its star which do
not intersect $\sigma$, cf. \cite{viro1992moves}.%
} by $lk(\sigma,K)=\{\tau\in K:\sigma*\tau\in K\}$, \cite{lickorish1999simplicial}.
\end{defn}
From the above definitions we are able now to define the concept of
star of $\sigma$ in $K$, which is the union of all closed simplexes
containing $\sigma$ and it is denoted $St(\sigma)$. More formal,
\begin{defn}
The \textbf{star} $St(\sigma)$ of a simplex $\sigma$ is defined
by $St(\sigma)=lk(\sigma)*\sigma$.
\end{defn}
From this, the boundary of the star is $\partial St(\sigma)=lk(\sigma)*\partial\sigma$,
\cite{viro1992moves}. 

To describe the transformations needed in this section, consider a
link of a $p$-simplex $\sigma_{p}$ in a triangulated $n$-dimensional
manifold isomorphic to the boundary of a $q$-simplex, $lk(\sigma_{p})\cong\partial\sigma_{q}$,
where $n=p+q$. For example, in the case of a triangulated surface
take a 1-simplex $\sigma_{1}$. In this case, the link is the set
of the two vertices, denoted by $*$, opposite to $\sigma_{1}$ belonging
to both 2-simplexes containing $\sigma_{1}$. Thus, the $q$-simplex
mentioned above is, in this example, a 1-simplex containing these
two points: 
\[
\xy/r1.5pc/:+(0,.5)="A",+(0,-1)="B",+(0,.2)="Z","A",{\xypolygon3"C"{~>{.}}},"B",{\xypolygon3"D"{~:{(-1,0):}~>{.}}},"C1"*{*},"D1"*{*},"C2";"C3"**@{-},"Z"*{\sigma_{1}}\endxy
\]

Then, the boundary of the star of $\sigma_{p}$, $\partial St(\sigma_{p})=lk(\sigma_{p})*\partial\sigma_{p}$,
is isomorphic to the join of the boundaries of a $q$-dimensional
and a $p$-dimensional simplexes, i.e. $\partial St(\sigma_{p})\cong\partial\sigma_{q}*\partial\sigma_{p}$.
Now, the transformation, called simplex move of index $p$, is the
replacement%
\footnote{An extensive description of this type of moves can be also found in
\cite{lickorish1999simplicial}, where the transformations are called
bistellar moves.%
} 
\[
St(\sigma_{p})\mapsto\sigma_{q}*\partial\sigma_{p}.
\]
 In our example, the transformation results in two 2-simplexes too,
but with the $q$-simplex as the common edge: \foreignlanguage{british}{
\[
\xy/r1.5pc/:+(0,.5)="A",+(0,-1)="B","A",{\xypolygon3"C"{}},"B",{\xypolygon3"D"{~:{(-1,0):}}}\endxy\quad\rightleftarrows\quad\xy/r1.5pc/:+(0,.5)="A",+(0,-1)="B","A",{\xypolygon3"C"{~>{}}},"B",{\xypolygon3"D"{~:{(-1,0):}~>{}}},"C1";"D1"**@{-},"C1";"D2"**@{-},"C1";"D3"**@{-},"C2";"D1"**@{-},"C3";"D1"**@{-}\endxy.
\]
}The inverse transformation to the general one given above is the
transformation of index $q=n-p$.

In the case where $p=n$ the transformation is a star subdivision
centered at the given $n$-simplex $\sigma_{n}$, \cite{viro1992moves}.
A star subdivision replaces the star $St(\sigma_{n})$ of a simplex
$\sigma_{n}$ in a triangulation $T$ by the cone%
\footnote{A cone of $K$ is defined to be the join of a simplicial complex $K$
and a single point, \cite{ito1993encyclopedic}.%
} of $\partial St(\sigma_{n})$ centered in a point $b\in\sigma_{n}$
leaving the rest of the triangulation $T\backslash St(\sigma_{n})$
unchanged. For instance, 
\[
\xy/r2.0pc/:\xypolygon3{}\endxy\quad\rightleftarrows\quad\xy/r2.0pc/:\xypolygon3{~<{-}}\endxy
\]
The above transformations are also called Alexander moves%
\footnote{These moves are the dual version of the operations involved in the
evaluation of spin networks in terms of $6j$-symbols via Moussouris'
algorithm, cf. Section \prettyref{sub:The-Decomposition-Theorem}.%
}. J. W. Alexander showed that for any dimensionally homogeneous%
\footnote{This means that $P$ is a union of closed simplexes of the same dimension.%
} polyhedron $P$ any of its triangulations can be transformed to any
other by a finite sequence of Alexander moves, \cite{turaev1992state}.

However, the number of Alexander moves is infinite%
\footnote{They are determined by the combinatorics of the star of the simplex
$\sigma_{p}$ in a triangulated space, e.g. by the number of simplexes
containing $\sigma_{p}$.%
} and, in the case of triangulations, one can not factorize them into
a finite number of elementary ones. Thus, in order to verify the invariance
of the above state sums under this type of transformations one has
to dualise the moves in the sense described next. Each triangulation
of a manifold $M$ induces a cell subdivision dual to that triangulation,
which can be constructed with help of the notion of barycenter of
the simplexes involved as follows. Take a strictly increasing sequence
$\sigma_{0}\subset\sigma_{1}\subset\dots\subset\sigma_{m}$ of simplexes
of the triangulation of a manifold $M$ and associate an $m$-dimensional
simplex $[\beta_{0},\beta_{1},\dots,\beta_{m}]\subset M$ whose vertices
are the barycenters of $\sigma_{0},\dots,\sigma_{m}$. For a simplex
$\sigma$ of $M$, the union $\sigma^{*}$ of all simplexes $[\beta_{0},\dots,\beta_{m}]$
where $\beta_{0}$ is the barycenter of $\sigma$, is a combinatorial
cell of dimension $\text{dim}M-\text{dim}\sigma$ called the barycentric
star of $\sigma$. The cells $\{\sigma^{*}\}_{\sigma}$, where $\sigma$
goes over all simplexes of $M$ form a cell subdivision of $M$, cf.
\cite{turaev1992state}. In a less technical way and in the here
relevant 3-dimensional case, the dual cell complex to a tetrahedron
is a collection of six 2-dimensional cells sharing a single vertex,
which is the barycenter of the tetrahedron. Each of the edges of the
tetrahedron intersects a 2-cell at exactly one point as in the figure
below.
\[
\xy/r5pc/:="A",+(0,1.5)="B","A",{\xypolygon3"K"{~:{(1.25,0):(-.2,-.5)::}~<>{;"B"**@{.}}~>{.}}},%tetrahedron
"A",{\xypolygon3"X"{~:{(.68,.06):(.1,.45)::}~>{}}},"A";"X1"+(.15,0)**@{--},"A";"X2"+(.025,.025)**@{-},"A";"X3"**@{-},%
"X3",+(-.25,.5),{\xypolygon3"Y"{~:{(.9,0):(.1,1.4)::}~={-4}~>{}}},"B";"K1"**{}?<>(.43)="Y4","Y4";"Y0"**@{-},"B";"K3"**{}?<>(.5)="Y5","Y5";"Y0"**@{-},"X3";"Y0"**@{-},"X2",+(.2,.6),{\xypolygon3"Z"{~:{(-.55,0):(.3,2.2)::}~={3}~>{}}},"B";"K1"**{}?<>(.43)="Z4","Z4";"Z0"**@{-},"B";"K2"**{}?<>(.5)="Z5","Z5";"Z0"**@{-},"X2"+(.025,.025);"Z0"**@{-},%
"X0"+(0.1,.7),{\xypolygon3"L"{~:{(-1.25,0):(0.06,-.65)::}~>{}}},"B";"K2"**{}?<>(.5)="L4","L4";"L0"**@{-},"B";"K3"**{}?<>(.5)="L5","L5";"L0"**@{-},"X1"+(.15,0);"L0"**@{--},"K2";"K3"**{}?<>(.5)="K4","K4"+(-.152,.05)*@{*};"X0"**@{-},"K4"+(-.152,.05);"Y0"**@{-},"K4"+(-.152,.05);"Z0"**@{-},"K4"+(-.152,.05);"L0"**@{}?!{"Y4";"Y0"}="I","I";"L0"**@{-},"I";"K4"+(-.152,.05)**@{.}\endxy
\]
In this way, a triangulated 3-manifold gives rise globally to a dual
cell complex with 3-cells homeomorphic to balls, called special spine,
\cite{kauffman1994temperley}.

With the help of this dualisation it is possible to factorize the
dual form of the star subdivisions. This is achieved by local modifications
of the special spines called Matveev-Piergallini moves. There are
three transformations of this kind, the bubble move%
\footnote{In the case of three incident planes defining an edge, cf. (2) in
def. \prettyref{def:2-dimensional-polyhedron}, this move is called
edge dilation. In the presence of $\mathcal{L}$- and $\mathcal{M}$-moves
both versions of the $\mathcal{B}$-move are equivalent, \cite[Ch. 10]{kauffman1994temperley}.%
} denoted by $\mathcal{B}$, the lune move denoted by $\mathcal{L}$
and the Matveev move denoted by $\mathcal{M}$:

\bigskip{}

\selectlanguage{british}%
$\mathcal{B}\triangleq$
\xy/r2pc/:\xypolygon4{~:{(3,0):(.3,.3)::}}\endxy
$\quad\rightleftarrows\quad$
\xy/r2pc/:{\xypolygon4"B"{~:{(3,0):(.3,.3)::}~>{}}},
"B1";"B4"**@{-},"B2";"B3"**@{-},
"B3";"B4"**@{-},(.5,0)*\ellipse(1,.2){.},(.5,0)*\ellipse(1,.2)_=:a(180){-},
(.5,0)*\ellipse(1,1)^=:a(180){-},
(2,0);(1,2.5)**{}?!{"B2";"B1"}="I1",
(0,0);(1,2.5)**{}?!{"B2";"B1"}="I2",
"B1";"I1"**@{-},"B2";"I2"**@{-},"I2";"I1"**@{--}
\endxy

\selectlanguage{english}%
\medskip{}

\selectlanguage{british}%
$\mathcal{L}\triangleq$
\xy/r2pc/:{\xypolygon4"L"{~:{(3,0):(.3,.3)::}~>{}}},
"L1";"L4"**@{-},"L2";"L3"**@{-},
"L3";"L4"**@{-},
"L1";"L2"**{}?<>(.33)="U1",
"L1";"L2"**{}?<>(.6)="U2",
"L1";"L2"**{}?<>(.9)="U3",
"L4";"L3"**{}?<>(.6)="D2",
"L4";"L3"**{}?<>(.33)="D1",
"L1";"U1"**@{-},"U1";"U2"**@{-},"U1";"D1"**@{-},"U2";"D2"**@{-},"U3"+(0,.2);"D2"**@{-},
"U3";"L2"**@{-},"U2"+(0,1.2)="V1";"U3"+(0,.2)**@{-},"V1";"U2"**@{-},"U3";"U2"**@{.},
"D1";"D1"+(0,-1.25)="W1"**@{-},"W1";"U1"+(0,-1.25)="W2"**@{-},"W2";"U1"**@{.}
\endxy
$\quad\rightleftarrows\quad$
\xy/r2pc/:{\xypolygon4"L"{~:{(3,0):(.3,.3)::}~>{}}},
"L1";"L4"**@{-},"L2";"L3"**@{-},
"L3";"L4"**@{-},
"L1";"L2"**{}?<>(.33)="U1",
"L1";"L2"**{}?<>(.6)="U2",
"L1";"L2"**{}?<>(.9)="U3",
"L4";"L3"**{}?<>(.6)="D2",
"L4";"L3"**{}?<>(.33)="D1",
"D1";"D1"+(0,-1.25)="W1"**@{-},"W1";"U1"+(0,-1.25)="W2"**@{-},"W2";"U1"**@{.},
"L1";"U1"**@{-},"U1";"D1"**@{-},"U3"+(0,.2);"D2"**@{-},
"U3";"L2"**@{-},"U2"+(0,1.2)="V1";"U3"+(0,.2)**@{-},
"V1";"U2"**{}?<>(.4)="V2","U2";"V2"**@{.},"V1";"V2"**@{-},
"U3";"U2"**@{.},
(2,0)="C";"D2"**\crv{(2,-.3)&(0,0)},
"C";"V1"+(-.1,-.2)**\crv{(2,.3)&"V1"+(-.1,-.6)}?!{"U2";"U1"}="I2","U1";"I2"**@{-},"U2";"I2"**@{.},
"C";"U2"**\crv{~*=<3pt>{.}(2,.3)&"U2"-(0,.3)}
\endxy

\medskip{}

$\mathcal{M}\triangleq$
\xy/r2pc/:{\xypolygon4"L"{~:{(3,0):(.3,.3)::}~>{}}},
"L1";"L4"**@{-},"L2";"L3"**@{-},"L3";"L4"**@{-},
"L1";"L2"**{}?<>(.33)="U1",
"L1";"L2"**{}?<>(.6)="U2",
"L1";"L2"**{}?<>(.9)="U3",
"L4";"L3"**{}?<>(.6)="D2",
"L4";"L3"**{}?<>(.33)="D1",
"L1";"U1"**@{-},"U1";"U2"**@{-},"U2";"D2"**@{-},"U3"+(0,.2);"D2"**@{-},
"U3";"L2"**@{-},"U2"+(0,1.2)="V1";"U3"+(0,.2)**@{-},"V1";"U2"**@{-},"U3";"U2"**@{.},
"L0";"L1"**@{-},"L0";"L4"**@{-},
"L2";"L3"**{}?<>(.5)="L5","L5";"L0"**@{}?!{"U3";"D2"}="L6","L6";"L0"**@{}?!{"U2";"D2"}="L7",
"L5";"L6"**@{-},"L6";"L7"**@{.},"L0";"L7"**@{-},
"L0"-(0,1.5),{\xypolygon4"B"{~:{(3,0):(.3,.3)::}~>{}}},
"B0";"B4"**@{-},"B2";"B3"**{}?<>(.5)="B5","B5";"B0"**@{-},"L4";"B4"**@{-},"L1";"B1"**@{-},
"L5";"B5"**{}?!{"L3";"L4"}="M1","L0";"B0"**{}?!{"L3";"L4"}="M2","L5";"M1"**@{.},"L0";"M2"**@{.},"B5";"M1"**@{-},"B0";"M2"**@{-},
"B1";"B0"**{}?!{"B4";"L4"}="M3","B0";"M3"**@{.},"B1";"M3"**@{-},
\endxy
$\quad\rightleftarrows\quad$
\xy/r2pc/:{\xypolygon4"L"{~:{(3,0):(.3,.3)::}~>{}}},
"L1";"L4"**@{-},"L2";"L3"**@{-},"L3";"L4"**@{-},"L0"*{\bullet},
"L1";"L2"**{}?<>(.33)="U1",
"L1";"L2"**{}?<>(.6)="U2",
"L1";"L2"**{}?<>(.9)="U3",
"L4";"L3"**{}?<>(.6)="D2",
"L4";"L3"**{}?<>(.33)="D1",
"L2";"L3"**{}?<>(.5)="L5","L5";"L0"**@{}?!{"U3";"D2"}="L6",
"L5";"L6"**@{-},"L6";"L0"**@{.},
"L0"-(0,1.5),{\xypolygon4"B"{~:{(3,0):(.3,.3)::}~>{}}},
"B0";"B4"**@{-},"B2";"B3"**{}?<>(.5)="B5","B5";"B0"**@{-},"L4";"B4"**@{-},"L1";"B1"**@{-},
"L5";"B5"**{}?!{"L3";"L4"}="M1","L0";"B0"**{}?!{"L3";"L4"}="M2","L5";"M1"**@{.},"L0";"M2"**@{.},"B5";"M1"**@{-},"B0";"M2"**@{-},
"B1";"B0"**{}?!{"B4";"L4"}="M3","B0";"M3"**@{.},"B1";"M3"**@{-},
"L1";"U1"**@{-},"U3"+(0,.2);"D2"**@{-},
"U3";"L2"**@{-},"U2"+(0,1.2)="V1";"U3"+(0,.2)**@{-},
"V1";"U2"**{}?<>(.4)="V2","U2";"V2"**@{.},"V1";"V2"**@{-},
"U3";"U2"**@{.},
(2,0)="C";"D2"**\crv{(2,-.3)&(0,0)}?!{"L0";"L4"}="L7","L7";"L4"**@{-},"L0";"L7"**@{.},
"C";"V1"+(-.1,-.2)**\crv{(2,.3)&"V1"+(-.1,-.6)}?!{"U2";"U1"}="I2"?!{"L0";"L1"}="L8","U1";"I2"**@{-},"U2";"I2"**@{.},
"L8";"L1"**@{-},"L0";"L8"**@{.},
"C";"U2"**\crv{~*=<3pt>{.}(2,.3)&"U2"-(0,.3)}
\endxy

\bigskip{}

\selectlanguage{english}%
The Matveev-Piergallini moves, however, do not act on the class of
barycentric star subdivisions of triangulations\textbf{} so there
is the need to enlarge the class of objects on which the state sums
are defined. These objects, called simple 2-polyhedra, appear in a
natural way as 2-skeletons%
\footnote{Recall that given an inductive definition of a n-dimensional simplicial
complex $K$, the m-skeleton of $K$ is obtained by stopping at the
m-th step. %
} of the cell subdivisions of compact 3-manifolds dual to triangulations,
\cite{turaev1992state}. They also have a correspondence with objects
of the recoupling theory defined in \prettyref{sec:Some-More-Diagrammatics}
which will be given after a short discussion about simple 2-polyhedra
and the state sum defined for them. For an extensive discussion of
the topic see the original papers by Turaev and Viro, \cite{turaev1992state}
and \cite{viro1992moves}.
\begin{defn}
\label{def:2-dimensional-polyhedron}A 2-dimensional polyhedron (with
boundary) $X$ is \textbf{simple }if the neighborhood of each point
of $X$ is homeomorphic to either of the next spaces:\end{defn}
\begin{enumerate}
\item $\mathbb{R}^{2}$ 
\item The union of three half-planes meeting in their common boundary line\foreignlanguage{british}{
\[
\xy/r2pc/:{\xypolygon4"L"{~:{(3,0):(.3,.5)::}~>{}}},"L1";"L2"**@{-},"L2";"L3"**{}?<>(.5)="L5","L2";"L5"**@{--},"L1";"L4"**{}?<>(.5)="L6","L6";"L5"**@{--},"L6";"L1"**@{--},"L1"+(0,1.3)="V1";"L1"**@{--},"V1";"L2"+(0,1.3)="V2"**@{--},"L2";"V2"**@{--},"L1"+(.65,-.75)="D1";"L1"**@{--},"L2"+(.65,-.75)="D2";"L2"**@{.},"D1";"D2"**{}?!{"L1";"L6"}="D3";"D1"**@{--},"D3";"D2"**@{.}\endxy
\]
}
\item The cone over the one-skeleton of the tetrahedron%
\footnote{This is homeomorphic to six 2-dimensional cells sharing a single vertex,
cf. \cite[Ch. 10]{kauffman1994temperley}.%
} 
\[
\xy/r5pc/:(1,-.35)="A",+(0,1.5)="B","A",{\xypolygon3"K"{~:{(1.25,0):(-.2,-.5)::}~<>{;"B"**@{}}~>{}}},"A",{\xypolygon3"X"{~:{(.68,.06):(.1,.45)::}~>{}}},"A";"X1"+(.15,0)**@{--},"A";"X2"+(.025,.025)**@{-},"A";"X3"**@{-},"X3",+(-.25,.5),{\xypolygon3"Y"{~:{(.9,0):(.1,1.4)::}~={-4}~>{}}},"B";"K1"**{}?<>(.43)="Y4","Y4";"Y0"**@{-},"B";"K3"**{}?<>(.5)="Y5","Y5";"Y0"**@{-},"X3";"Y0"**@{-},"X2",+(.2,.6),{\xypolygon3"Z"{~:{(-.55,0):(.3,2.2)::}~={3}~>{}}},"B";"K1"**{}?<>(.43)="Z4","Z4";"Z0"**@{-},"B";"K2"**{}?<>(.5)="Z5","Z5";"Z0"**@{-},"X2"+(.025,.025);"Z0"**@{-},"X0"+(0.1,.7),{\xypolygon3"L"{~:{(-1.25,0):(0.06,-.65)::}~>{}}},"B";"K2"**{}?<>(.5)="L4","L4";"L0"**@{-},"B";"K3"**{}?<>(.5)="L5","L5";"L0"**@{-},"X1"+(.15,0);"L0"**@{--},"K2";"K3"**{}?<>(.5)="K4","K4"+(-.152,.05)*{\bullet};"X0"**@{-},"K4"+(-.152,.05);"Y0"**@{-},"K4"+(-.152,.05);"Z0"**@{-},"K4"+(-.152,.05);"L0"**@{}?!{"Y4";"Y0"}="I","I";"L0"**@{-},"I";"K4"+(-.152,.05)**@{.}\endxy\cong\quad\xy/r2pc/:{\xypolygon4"L"{~:{(3,0):(.3,.3)::}~>{}}},"L1";"L4"**@{-},"L2";"L3"**@{-},"L3";"L4"**@{-},"L1";"L2"**{}?<>(.33)="U1","L1";"L2"**{}?<>(.6)="U2","L1";"L2"**{}?<>(.9)="U3","L4";"L3"**{}?<>(.6)="D2","L4";"L3"**{}?<>(.33)="D1","L1";"U1"**@{-},"U1";"U2"**@{-},"U2";"D2"**@{-},"U3"+(0,.2);"D2"**@{-},"U3";"L2"**@{-},"U2"+(0,1.2)="V1";"U3"+(0,.2)**@{-},"V1";"U2"**@{-},"U3";"U2"**@{.},"L2";"L3"**{}?<>(.5)="L5","L5";"L0"**@{}?!{"U3";"D2"}="L6","L6";"L0"**@{}?!{"U2";"D2"}="L7"*{\bullet},"L5";"L6"**@{-},"L6";"L7"**@{.},"L0";"L7"**@{-},"L1";"L4"**{}?<>(.5)="L8";"L0"**@{-},"L0"-(0,1.5),{\xypolygon4"B"{~:{(3,0):(.3,.3)::}~>{}}},"B2";"B3"**{}?<>(.5)="B5","B5";"B0"**@{-},"L5";"B5"**{}?!{"L3";"L4"}="M1","L5";"M1"**@{.},"B5";"M1"**@{-},"B1";"B4"**{}?<>(.5)="B6","B6";"B0"**@{-},"B6";"L8"**@{-}\endxy
\]

\item The half-plane $\mathbb{R}_{+}^{2}$, or
\item the union of three quadrants $\{(x,y)\in\mathbb{R}^{2}:x\geq0,y\geq0\}$
meeting in the half-line $x=0$,
\[
\xy/r2.5pc/:{\xypolygon4"L"{~:{(1,0):(.3,.5)::}~>{}}},"L1";"L2"**@{-},"L2";"L3"**@{--},"L1";"L4"**@{-},"L3";"L4"**@{--},"L1"+(0,.75)="V1";"L1"**@{-},"V1";"L2"+(0,.75)="V2"**@{--},"L2";"V2"**@{--},"L1"+(.45,-.6)="D1";"L1"**@{-},"L2"+(.45,-.6)="D2";"L2"**@{.},"D1";"D2"**{}?!{"L1";"L4"}="D3";"D1"**@{--},"D3";"D2"**@{.}\endxy
\]

\end{enumerate}
The points of $X$ with neighborhoods homeomorphic to the last two
ones above belong to the boundary $\partial X$, which is a simple
graph%
\footnote{A simple graph is a finite 1-dimensional CW-complex such that its
0-cells are trivalent vertices and its 1-cells are homeomorphic to
$\mathbb{R}$, called edges, or to $S^{1}$ called loops.%
}. 

The simple 2-dimensional polyhedra are naturally stratified. The $k$-strata
$(k=2,1,0)$ being $k$-dimensional connected components of the set
of internal points in $X$ with neighborhood $\mathbb{R}^{2}$ , (2.)
and (3.) respectively. For the stratification of the boundary, the
1-strata are the edges and loops of the simple graph and the 0-strata
are the 3-vertices of the simple graph.

Now we have the ingredients to redefine the previous state sum with
the simple 2-polyhedra. The same concepts, like admissibility, apply
but with the subtlety that the coloring is now with respect to the
2-strata of $X$%
\footnote{Admissibility is then w.r.t. the edge formed by three labeled 2-strata
as in (2.) in the definition \prettyref{def:2-dimensional-polyhedron}.%
}. Any coloring of $X$ induces naturally a coloring of $\partial X$
since a 1-stratum of the boundary acquires the color assigned to the
2-stratum of $X$ in which this 1-stratum is contained. We denote
the map defined by this construction by $\partial:\, adm(X)\rightarrow adm(\partial X)$.
As mentioned before, the tetrahedra in Eq. \prettyref{eq:state-of-M}
are associated to vertices, which carry a 6-tuple labelling the 2-strata
that meet at the vertex. The association is such that if $x\in X\backslash\partial X$,
the 1-skeleton of its corresponding tetrahedron $T_{x}$ is the polyhedral
link of $x$ in $X$, i.e. the edges correspond to lines in the 2-strata
and the vertices correspond to points in the 1-strata where the lines
of the three 2-strata defining this 1-strata meet. This association
labels automatically the edges of $T_{x}$. The figure below shows
the six 2-strata intersecting in one vertex in the center of a tetrahedron
made out of the germs of 2-strata defining the polyhedral link%
\footnote{Notice that the figure below is, in fact, the cone over the 1-skeleton
of the tetrahedron. Hence, the polyhedral link is the 1-skeleton defining
this cone.%
}: 
\[
\xy/r2pc/:{\xypolygon3"T"{~:{(1,-.2):(-.1,1)::}}},"T0"+(1.25,-.1)="R1","T0"+(.2,0)="I","I";"T1"**@{-},"I";"T2"**@{-},"I";"T3"**@{-},"R1";"T1"**@{-},"R1";"T2"**@{.},"R1";"T3"**@{-},"R1";"I"**{}?!{"T1";"T3"}="K";"R1";"K"**@{-},"K";"I"**@{.}\endxy
\]

The edges of the dual tetrahedron $\hat{T}_{x}$ correspond to edges
of $T_{x}$, thus, they obtain the same labelling. This can also be
seen, if one constructs the dual tetrahedron directly from the configuration
of six 2-cells defining the vertex $x$, as in \prettyref{def:2-dimensional-polyhedron}.
From this construction it is clear that each admissible coloring of
$X$ induces an admissible coloring of $\hat{T}_{x}$. Then for $\phi\in adm(X)$
define
\begin{equation}
|X|_{\phi}=w^{-2\chi(X)+\chi(\partial X)}\prod_{s=1}^{f}w_{\partial\phi(E_{s})}^{\chi(E_{s})}\prod_{r=1}^{b}w_{\phi(\Gamma_{r})}^{2\chi(\Gamma_{r})}\prod_{t=1}^{d}|\hat{T}_{x_{t}}^{\phi}|\;\in K\label{eq:X-phi}
\end{equation}
where $\phi(\Gamma_{i})$ is the color of the edge $i$ of $\hat{T}_{x}$
corresponding to the 2-stratum $\Gamma_{i}$ of $X$ and $\partial\phi(E_{s})$
is the color of the edge $E_{s}$ of $\partial X$%
\footnote{Recall that there are $f$ edges $E_{s}$ in the boundary of $M$,
cf. Eq. \prettyref{eq:phi-summand-of-invariant}.%
}. If $E_{s}$ is homeomorphic to $\mathbb{R}$ then $\chi(E_{s})=-1$
and if it is homeomorphic to $S^{1}$ then $\chi(E_{s})=0$. The corresponding
invariant is 
\[
|X|=\sum_{\phi\in adm(X)}|X|_{\phi},
\]
 and for any admissible coloring $\alpha\in adm(\partial X)$ we have
\[
\Omega_{X}(\alpha)=\sum_{\phi:\partial\phi=\alpha}|X|_{\phi}.
\]

From the conditions (I), (II) and (III) follows the invariance of
$|X|$ and $\Omega_{X}(\alpha)$%
\footnote{The moves $\mathcal{L}$, $\mathcal{M}$ and $\mathcal{B}$ preserve
the boundary, \cite{turaev1992state}.%
} under $\mathcal{L}$-, $\mathcal{M}$- and $\mathcal{B}$-moves respectively.
A detailed proof of the invariance can be found in \cite{turaev1992state}
and \cite{kauffman1994temperley}. In the latter reference, the proof
occurs in the framework of Temperley-Lieb recoupling theory in which
the conditions (I) and (II) correspond to the orthogonality and Biedenharn-Elliott
identity for $q-6j$-symbols respectively. 

In the rest of this section we will discuss, following \cite{kauffman1994temperley},
the correspondence between the approach just discussed and another,
which we will call Kauffman-Lins approach, involving the recoupling
theory. The aim of this comparison is to understand the connections
between these two frameworks. This is done%
\footnote{The correspondence here is informal in nature. The formal proof is
by S. Piunikhin, \cite{piunikhin1992turaev}.%
} first by assigning weights to the vertices, edges and faces given
a coloring of the special spine of a 3-manifold corresponding to $X$,
cf. Eq. \prettyref{eq:q6j-symbol-and-Tetrahedron}. Then, a partition
function involving these weights is defined and finally the identification
of the factors is made. 

The first and most obvious association is, as explained before, that
of a vertex $x$ with a colored tetrahedron, 
\[
|\hat{T}_{\sigma(x)}|\triangleq\xygraph{!{0;/r1.pc/:}!{\sbendh@(0)|{^{b}}}*{\bullet}!{\sbendv@(0)|{^{c}}}[l(2)]*{\bullet}!{\xcaph[2]@(0)|{j}}[r]*{\bullet}[ll]!{\sbendv@(0)|{a}}*{\bullet}!{\sbendh@(0)|{d}}[lu]!{\vcap[1.5]}[r(1.5)]!{\xcapv[2]@(0)|{i}}[d][l(1.5)]!{\vcap[-1.5]}}.
\]

The weight of an edge in a special spine is the value $\theta(a,b,c)$
associated to it, where the labels are those of the three 2-cells
incident to the edge, i.e. we have
\[
\theta(a,\, b,\, c)=\xygraph{!{0;/r1.5pc/:}!{\vcap|{a}}!{\vcap-|{^{c}}}*{\bullet}!{\xcaph@(0)|{b}*{\bullet}}}\triangleq\xy/r1pc/:(0,-.9)="A","A",{\xypolygon4"L"{~:{(3,0):(.3,.5)::}~>{}}},"L0"-(0,.3)*{b},"L0"+(0,1.8)*{a},"L0"+(3.6,.5)*{c},"L1";"L2"**@{-},"L2";"L3"**{}?<>(.5)="L5","L2";"L5"**@{--},"L1";"L4"**{}?<>(.5)="L6","L6";"L5"**@{--},"L6";"L1"**@{--},"L1"+(0,1.3)="V1";"L1"**@{--},"V1";"L2"+(0,1.3)="V2"**@{--},"L2";"V2"**@{--},"L1"+(.65,-.75)="D1";"L1"**@{--},"L2"+(.65,-.75)="D2";"L2"**@{.},"D1";"D2"**{}?!{"L1";"L6"}="D3";"D1"**@{--},"D3";"D2"**@{.}\endxy
\]

As before, the 2-cells correspond in the dual sense to edges, hence,
they carry only one color $i$. Therefore, the weight associated with
it is the quantum integer $\Delta_{i}$, 
\[
\Delta_{i}\triangleq\xygraph{!{0;/r1.5pc/:}[d(0.5)]!{\sbendh[0.75]@(0)}[dl(0.25)]!{\xcaph@(0)}[dl(0.75)]!{\sbendh[0.75]@(0)}[dll]!{\xcaph@(0)<{i}}}.
\]

\begin{defn}
The partition function $TV_{M^{3}}$ for a 3-manifold $M$ in the
Temperley-Lieb recoupling theory is defined by
\[
TV_{M^{3}}=\sum_{\sigma}\prod_{v,e,f}\theta(\sigma(e))^{\chi(e)}\Delta_{\sigma(f)}^{\chi(f)}|\hat{T}_{\sigma(x)}|
\]
where $\chi(f)$ and $\chi(e)$ are the Euler characteristic of the
2-cell $f$ and the edge $e$. If $e$ has graphical nodes then $\chi(e)=-1$
and if $e$ is a loop without nodes then $\chi(e)=0$, cf. \prettyref{eq:X-phi}.
The coloring $\sigma$ involved in this definition of the partition
function is over a finite color set $\{0,1,\dots,r-2\}$ as well,
and admissible.
\end{defn}
The behavior of $TV_{M^{3}}$ under the bubble move deserves more
attention since it helps us identify the rest of the factors in the
definition of the Turaev-Viro invariant $|X|$. The result of performing
a bubble move on a face of the special spine of $M^{3}$ turns out
to be a global factor of $\Delta_{a}^{-1}\sum_{i,j}\Delta_{i}\Delta_{j}$,
where $a$ is the color of the face where the bubble move was realized
and the sum is over all admissible triples $(a,i,j)$. This factor
can be explained as follows. Consider the bubble move on the face
colored with $a$. Notice that the rest of the special spine is unaffected.
We have then
\[
\xygraph{!{0;/r1.5pc/:}[u]!{\xcapv[2]@(0)|{a}}[u]!{\xcaph[2.2]@(0)}[d(2)][l]!{\xcaph[2.2]@(0)}[u(2)][r(1.2)]!{\xcapv[2]@(0)}}\Longleftrightarrow\qquad\quad\xygraph{!{0;/r1.5pc/:}[u]!{\xcapv[2]@(0)|{a}}[r(0.7)]!{\vcap[0.75]|{j}}!{\vcap[-0.75]|{^{i}}}[u][r(1.5)]!{\xcapv[2]@(0)}[l(2.2)][u]!{\xcaph[2.2]@(0)}[d(2)][l]!{\xcaph[2.2]@(0)}}
\]
 where the colors $i,j$ correspond to the surfaces inside the circle
and to the hemisphere respectively, if we imagine the circle as being
the intersection of the surface of an hemisphere $j$ with the surfaces
$i$ and $a$. Since the cell colored by $a$ obtains a hole in the
process, its Euler characteristic is reduced by one. This explains
the term $\Delta_{a}^{-1}$. The factor involving the sum appears
as a natural consequence of the summation over all possible admissible
colorings keeping $a$ fixed. There is no $\theta$-value involved
in this term since the edge defined by the 2-cells $a,i,j$ is a loop
without nodes.

There is in fact a relation between the weights $\Delta_{j}$ of the
2-cells for $q$-admissible triples $(a,i,j)$ and $(b,i,j)$ where
$q=e^{i\pi/r}$ which strongly resembles condition (III) of the initial
data. For $a,b,i,j\in\{0,1,\dots,r-2\}$, we have 
\[
\Delta_{a}^{-1}\sum_{i,j}\Delta_{i}\Delta_{j}=\Delta_{b}^{-1}\sum_{i,j}\Delta_{i}\Delta_{j}
\]
where the sum is over the above $q$-admissible triples. Hence for
any $a\in\{0,1,\dots,r-2\}$ and $(a,i,j)$ $q$-admissible we have
\[
\tau_{q}=\Delta_{a}^{-1}\sum_{i,j}\Delta_{i}\Delta_{j}=-\frac{2r}{(q-q^{-1})^{2}}.
\]

The above relation looks just like the relation \prettyref{eq:Cond-III}
for the elements of the commutative ring $K$. Thus we can identify
$w^{2}=\tau_{q}$ and for all $j\in\{0,1,\dots,r-2\}$ we have $w_{j}^{2}=\Delta_{j}$,
cf. \cite{turaev1992state}.

Now, notice that the $\theta$-values are in fact values of a vertex
on the boundary of the 3-manifold with admissible coloring induced
by the coloring of the special spine. These values correspond to products
of terms like $w_{\partial\phi(E_{s})}$ in the evaluation of $|X|_{\phi}$.
This is readily seen if one recalls that the value of the theta-nets
are combinatorial products of $\Delta_{i}$ in which admissible coloring
is involved.

Finally, to obtain a topological invariant of $M^{3}$ from the recoupling
theory one has to normalize the above defined state summation in order
to take care of its change by the factor $\tau_{q}$ when considering
bubble moves. Thus, the invariant of 3-manifolds in \cite{kauffman1994temperley}
is defined by
\[
I_{M^{3},q}=\tau_{q}^{-(t-1)}TV_{M^{3}}
\]
where $t$ is the number of 3-cells%
\footnote{Recall that 3-cells correspond to vertices in the dual form of the
combinatorial manifold $X$. Thus, $t$ is the number of vertices
in $X$.%
} in the decomposition of $M^{3}$. Thus, the factor $\tau_{q}^{-(t-1)}$
corresponds to the factor involving $w^{2}\in K$ in $|X|$. 

As mentioned before, the two invariants defined in this section coincide
when $q$ is a root of unity. The Kauffman-Lins approach gives the
tools to calculate the invariant in a direct way while the Turaev-Viro
approach gives us a broader insight about the theoretical structure
giving rise to the invariant and its link to the topological QFT by
the fact that a simple 2-polyhedron $X$ is, in fact, a cobordism
between simple graphs.

\section{Invariants from Spherical Categories\label{sec:Invariants-from-Spherical}}

In this section a more general version of the above discussed theory
is presented following \cite{barrett1996invariants}, where an algebraic
framework for constructing invariants of closed oriented 3-manifolds
is presented. We use the previously learned concepts in section \prettyref{sub:Spherical-categories}
since the data for the construction of the invariant is a spherical
category, for instance, the representations of the quantized enveloping
algebra of $\mathfrak{sl}_{2}$ give the Turaev-Viro invariant defined
in the previous section. In \cite{barrett1996invariants} the invariance
from a finite set of moves on triangulations is also proved without
going to the dual form of the transformations. Here we will only give
a general account of the results obtained by J. W. Barrett and B.
W. Westbury.

For the rest of this section by spherical category we mean an \textbf{additive}
(strict) spherical category and we assume some conditions on it. First,
the ring $\mathbb{F}\cong End(e)$, which is commutative in any additive
monoidal category, is assumed to be a field. Second, each set of morphisms
in our spherical category is a finite dimensional vector space (over
$\mathbb{F}$). 

In this framework the data for the state sum consists of the set of
labels $I=J$, which is the set%
\footnote{More precisely, it is a set of representatives of each isomorphism
class of simple objects.%
} of simple objects in the category, a set of state spaces and a set
of partition functions for each tetrahedron. Denote by $D(a,b,c)$
the standard oriented triangle $+(012)$ labelled by $\partial_{0}D\mapsto a,\partial_{1}D\mapsto b,\partial_{2}D\mapsto c$,
where $\partial_{i}$ is the map%
\footnote{Note that these maps satisfy $\partial_{i}\partial_{j}\sigma=\partial_{j-1}\partial_{i}\sigma$
for $i<j$, cf. \cite{barrett1996invariants}.%
} sending any $n$-simplex $\sigma$ in a simplicial complex to one
of its $(n-1)$-faces obtained by omitting the $i^{th}$ vertex of
$\sigma$. The state space for $D(a,b,c)$ is a vector space $H(D_{a,b,c})=H(a,b,c)=Hom(b,a\otimes c)$
over a field $\mathbb{F}$, whereas for the opposite oriented triangle
$-D(a,b,c)$ it is defined to be the dual vector space $H^{*}(a,b,c)$,
cf. Sec. \prettyref{sub:Spherical-categories}. Consider once again
the standard oriented tetrahedron $T=+(0123)$ with edges $\partial_{i}\partial_{j}T$
labelled by $e_{ij}$. The partition function corresponding to this
tetrahedron is defined to be the linear map
\[
\left\{ \begin{array}{ccc}
e_{01} & e_{02} & e_{12}\\
e_{23} & e_{13} & e_{03}
\end{array}\right\} _{+}:\, H(e_{23},e_{03},e_{02})\otimes H(e_{12},e_{02},e_{01})\rightarrow H(e_{23},e_{13},e_{12})\otimes H(e_{13},e_{03},e_{01})
\]
 and accordingly the partition function of the opposite oriented tetrahedron
$-T$ is defined by
\[
\left\{ \begin{array}{ccc}
e_{01} & e_{02} & e_{12}\\
e_{23} & e_{13} & e_{03}
\end{array}\right\} _{-}:\, H(e_{23},e_{13},e_{12})\otimes H(e_{13},e_{03},e_{01})\rightarrow H(e_{23},e_{03},e_{02})\otimes H(e_{12},e_{02},e_{01}).
\]
 Hence, once again, we observe the correspondence between a topological
object and a morphism, where the factors in the tensor products are
associated to each one of the four faces of the tetrahedron%
\footnote{Recall that the labelled trivalent vertex can be regarded in a certain
sense as the dual of a labelled triangle.%
}. 

This data determines an element $Z(M)\in\mathbb{F}$ for each labelled
simplicial closed manifold $M$, the simplicial invariant of $M$
obtained as follows. Let
\[
V(M)\equiv\bigotimes_{D\in K(M)}H(D)
\]
 where $K(M)$ is the simplicial complex triangulating $M$, $D$
are the triangles in this triangulation and $H(D)$ are the state
spaces of the corresponding triangles. Consider the tensor product
over the set of partition functions corresponding to each tetrahedron
in $K(M)$. The resulting morphism%
\footnote{Recall the construction of an edge in \prettyref{sec:Some-More-Diagrammatics}
where the representations of the quantum group were woven into edges
by permuting them in the order of their tensor product, cf. \cite{kauffman1994temperley}.%
} is a linear map $V(M)\rightarrow V^{\pi}(M)$, where the tensor product
$V^{\pi}(M)$ is defined as $V(M)$ but with factors permuted by some
permutation $\pi$. Furthermore, the iteration of the standard twist
$P:x\otimes y\mapsto y\otimes x$ gives a unique linear map sending
$V^{\pi}(M)$ back to $V(M)$. The composition of these two maps gives
a linear map $V(M)\rightarrow V(M)$ and its trace defines the invariant
$Z(M)$%
\footnote{Note that the notation given here and in \prettyref{sec:TQFT} is
no coincidence.%
}. Then a state sum invariant of a closed manifold is then obtained
by a weighted sum of these simplicial invariants over the class of
(admissible) labelling.

As mentioned in section \prettyref{sub:Spherical-categories} the
(additive) spherical category which defines the state sum model is
the non-degenerate quotient category constructed from the category
of representations of a spherical Hopf algebra. This construction
is important, and always possible (see sec. \prettyref{sub:Spherical-categories}),
because in order to construct the invariants one needs to take a non-degenerate
category since this property is needed for the semisimplicity condition
which allows us to construct a well defined manifold invariant. The
category of representations of the Hopf algebra may be degenerate
but the quotient \vpageref{thm:non-degenerate-quotient} is not. However,
it is not the category of the representations of any finite dimensional
Hopf algebra since it is not possible to assign a dimension to \textit{each}
object which would be additive and multiplicative under direct sum
and tensor product respectively. 

Given some isomorphisms of the label objects in a triangle $D_{a,b,c}$,
i.e. $\phi_{a}:\, a\rightarrow a'$, $\phi_{b}:\, b\rightarrow b'$
and $\phi_{c}:\, c\rightarrow c'$, there is an induced isomorphism
between the state spaces corresponding to $D_{a,b,c}$ and $D_{a',b',c'}$
\[
Hom(b,a\otimes c)\rightarrow Hom(b',a'\otimes c')
\]
given by $\alpha\mapsto\phi_{b}^{-1}\alpha(\phi_{a}\otimes\phi_{c})=\alpha'$. 

It follows that the map $V(M)\rightarrow V(M)$ is conjugated by the
induced isomorphism on the state space of each triangle in the triangulation
of the manifold $M$. Since the simplicial invariant mentioned above
is the trace of this map and the trace is invariant under conjugation
by a linear map we have, for a closed simplicial manifold $M$, that
the invariant $Z(M)$ only depends on the isomorphism class of the
labelling.

In fact, for a combinatorial%
\footnote{Combinatorial maps are maps of complexes, as opposed to simplicial
maps which are combinatorial maps that preserve orderings. The difference
is that a simplicial complex is a complex together with a total ordering
of the vertices of each simplex. Thus, a single simplex has no symmetries,
whereas the corresponding complex admits the permutations of its vertices
as its symmetries. Therefore, in the above case, the combinatorial
isomorphisms are representations of $S_{3}$.%
} isomorphism of labelled manifolds $f:\, M\rightarrow N$ the simplicial
invariants are equal. This follows from the pivotal structure of the
spherical category which governs the properties of the state space
of a triangle under combinatorial isomorphisms and the spherical property
which allows isotopy on the sphere%
\footnote{Any permutation of the vertices of a tetrahedron, i.e. elements of
$S_{4}$, can be extended to an isotopy of the sphere. %
} and governs the properties of the non-degenerate pairing.

To end this chapter, we present the main result given in \cite{barrett1996invariants},
which is stated for a more general type of 3-manifolds, however, we
present it here as a result for closed 3-manifolds.
\begin{thm}
\label{thm:Invariants-Sp-Cat}A finite semisimple spherical category
of non-vanishing dimension determines an invariant of oriented closed
3-manifolds. In other words, any simplicial closed manifold $M$ which
triangulates a given piecewise-linear manifold $\mathcal{M}$ determines
the manifold invariant given by
\[
C(M)=K^{-v}\sum_{l:E\rightarrow J}Z(M,l)\prod_{e\in E}\text{dim}_{q}(l(e))
\]
where $K$ is the dimension of the spherical category, $v$ the number
of vertices in $M$, the set of edges is denoted by $E$ and $J$
is a set of representatives from each isomorphism class of simple
objects of the category. The map $l:\, E\rightarrow J$ is the labelling,
thus, the labelled manifold is denoted by $(M,l)$. 
\end{thm}
As before, this state sum invariant does not depends neither on the
choice of simple objects $J$ nor on the choice of the simplicial
structure. The invariance under the choice of triangulation is a consequence
of the invariance of $C(M)$ under bistellar moves which follows from
the orthogonality and the Biedenharn-Elliot relation for the partition
function corresponding to the tetrahedron.

To summarize, we started in \prettyref{sec:Spin-Networks} with the
introduction of the classical model of Penrose's spin networks and
its state sum as the Regge-Ponzano theory. The term ``classical''
here refers to the construction of the spin networks out of representations
of Lie groups, $SU(2)$ in the case described in \prettyref{sec:Spin-Networks}.
However, the classical model has for some manifolds a divergent behavior
coming from the infinite sum over the set of labels in the coloring
of the triangulations. In order to avoid these divergences one has
to construct the manifold invariant from the q-deformed universal
enveloping algebra of the corresponding semisimple Lie algebra at
a root of unity; for instance, the Turaev-Viro invariant can be seen
as the regularization of the Regge-Ponzano state sum. Using this Hopf
algebra has the effect that the set of representation labels is finite,
hence, the state sum is finite, \cite{fairbairn2010quantum}. As mentioned
before, this is a special case of the state sums given in \prettyref{thm:Invariants-Sp-Cat},
where the data comes from the more abstract notion of spherical categories.

\chapter{Non-planar Spin Networks\label{cha:Non-planar-Spin-Networks}}

After giving an account of the physical motivation for spin networks,
their relation to general relativity and placing them in an algebraical
context as well as in the setting of TQFT we now turn to another language
useful to describe some aspects of these objects. We will use some
basic concepts of (topological) graph theory taken from \cite{hartsfield2003pearls,beineke1997graph}
to describe the embeddings of spin networks in surfaces, in particular
the cellular embeddings of non-planar graphs with a $(3,3)$-bipartite
graph as subgraph.

Moussouris' algorithm for the evaluation of planar spin networks and
the Decomposition Theorem are presented in the second section of this
chapter. We will apply and extend these ideas to evaluate the above
mentioned embeddings, which will allow us to define a toroidal symbol
in order to attempt a generalization of the algorithm for the evaluation
of non-planar networks. To achieve this a few identities for relating
the evaluations of the different embeddings in the torus are given.
We explain the main concepts involving the generalization of the evaluation
of non-planar spin networks in terms of toroidal symbols, however,
the process for obtaining the main result is still ongoing.

\section{Kuratowski's Theorem and the Embedding of Graphs in Surfaces\label{sec:Kuratowski's-Theorem-Embeddings}}

A \textbf{graph} $G$ is a pair of sets $(V_{G};E_{G})$ where $V_{G}\neq\emptyset$
is the set of vertices of $G$ and $E_{G}$ is a set of unordered
pairs of elements of $V_{G}$ which might be empty. The pairs of vertices
are called edges and they are defined, in an abstract way, as a relation
between the objects defining $G$. If two vertices form an edge, we
say they are adjacent or neighbors. A \textbf{subgraph} of a graph
$G$ is a graph $H=(V_{H};E_{H})$ such that $V_{H}\subseteq V_{G}$
and $E_{H}\subseteq E_{G}$, \cite{hartsfield2003pearls}. Another
concept related to subgraphs are the \textbf{minors} of a graph; these
are graphs obtained from $G$ by a succession of edge-deletions and
edge-contractions. If the minor $M$ was obtained only by edge contractions,
then $G$ is said to be contractible to $M$, \cite{beineke1997graph}.

As mentioned before, we are able to consider only trivalent vertices
since graphs of higher degree, i.e. with vertices of higher valence,
are expandable to cubic graphs. On the other hand, a graph is only
allowed to have at most one edge between two adjacent vertices. If
two vertices are joined by more than one edge, the structure is called
a multigraph. In the case of cubic graphs we only have trivalent vertices,
hence, a multigraph would have two vertices joined by at most three
edges. This impose, however, no constraints in the class of spin networks
since the double edge can be reduced to a single edge using \vref{eq:Edge-with-loop}
and a triple edge is exactly the theta function defined previously
as the value of a 3-vertex by \vref{eq:theta-function}.

Two graphs are \textbf{homeomorphic} if they can be obtained from
the same graph by subdividing its edges. Subdividing an edge $e=vw$
between two vertices $v,w$ is the operation of inserting a new vertex
$z$ such that $e$ is replaced by two new edges $vz$ and $zw$,
\cite{beineke1997graph}. 

The drawing of a graph $G$ in a (closed) surface $S$ consists of
points corresponding to vertices and simple curves, corresponding
to edges, joining the points. If there are no crossings in the drawing,
i.e. the curves do not meet except at their end-vertices, then the
drawing is an embedding. The embedding is called \textbf{cellular
}if each region is homeomorphic to an open disc. It is in this sense
that planarity is defined, namely, a graph is planar if it can be
embedded in a plane, hence, in the 2-sphere. Surprisingly, there is
a simple criterion for determining whether a graph is planar or not,
which is given by Kuratowski's theorem.
\begin{thm}
\textbf{\label{thm:Kuratowski's-Theorem}Kuratowski's Theorem.} A
graph is planar if and only if it has neither $K_{5}$ nor $K_{3,3}$
as a minor, i.e. there are no subgraphs homeomorphic or contractible
to $K_{5}$ and $K_{3,3}$.
\[
\xy/r2.0pc/:{\xypolygon5"C"{}},"C1"*@{*};"C3"*@{*}**@{-},"C1";"C4"**@{-},"C2"*@{*};"C4"*@{*}**@{-},"C2";"C5"*@{*}**@{-},"C3";"C5"**@{-},"C0"-(0,1.3)*{K_{5}}\endxy\:\text{ and }\xy/r2.0pc/:{\xypolygon4"B"{~:{(2.5,0):(0,.5)::}~>{}}},"B1";"B2"**{}?<>(.5)="B5"*@{*},"B3";"B4"**{}?<>(.5)="B6"*@{*},"B1"*@{*},"B2"*@{*},"B3"*@{*},"B4"*@{*},"B1";"B3"**@{-},"B1";"B4"**@{-},"B1";"B6"**@{-},"B5";"B3"**@{-},"B5";"B6"**@{-},"B5";"B4"**@{-},"B1";"B3"**@{-},"B2";"B3"**@{-},"B2";"B6"**@{-},"B2";"B4"**@{-},"B0"-(0,1.3)*{K_{3,3}}\endxy
\]
 
\end{thm}
The graph $K_{5}$ is called the complete graph on five vertices and
$K_{3,3}$ is the $(3,3)$-bipartite graph. Notice that the $K_{5}$
graph is 4-valent and can be expanded to obtain the Petersen graph
which is a cubic graph with 10 vertices as depicted below. If one
deletes any of its vertices and the edges incident to it, one finds
a graph homeomorphic to the $(3,3)$-bipartite graph. Hence, we only
need to focus on the latter graph.

\medskip{}

\begin{center}
\xy/r3.0pc/:
{\xypolygon5"C"{}},
"C1"*@{*};"C3"*@{*}**@{-},"C1";"C4"**@{-},"C2"*@{*};"C4"*@{*}**@{-},"C2";"C5"*@{*}**@{-},"C3";"C5"**@{-},"C0"-(0,1.3)*{K_{5}}
\endxy
$\rightarrow$
\xy/l1pc/:
{\xypolygon5"A"{}},
{\xypolygon5"B"{~:{(-2.5,0):}~>{}}},
{\xypolygon5"C"{~:{(-3.75,0):}}},
"B1"*@{*};"C1"*@{*}**@{-},"B2"*@{*};"C2"*@{*}**@{-},"B3"*@{*};"C3"*@{*}**@{-},"B4"*@{*};"C4"*@{*}**@{-},"B5"*@{*};"C5"*@{*}**@{-},
"B1";"A3"**@{-},"B1";"A4"**@{-},"B2";"A5"**@{-},"B2";"A4"**@{-},"B3";"A5"**@{-},"B3";"A1"**@{-},"B4";"A1"**@{-},"B4";"A2"**@{-},
"B5";"A3"**@{-},"B5";"A2"**@{-},
"A0"+(0,3.9)*{\text{Petersen graph}}
\endxy
\par\end{center}

\medskip{}

Closed surfaces, on the other hand, are categorized into orientable
and non-orientable surfaces, e.g. the sphere $S^{2}$, torus $T^{2}$
or the real projective plane $P^{2}$. Any oriented surface is homeomorphic
to the sphere or to the connected sum $T^{2}\#T^{2}\#\dots\#T^{2}$
of a finite number $h$ of tori, while any non-orientable closed surface
is homeomorphic to a connected sum of a finite number of copies of
the real projective plane, where the empty sum is defined as $P^{2}$
itself.

Whether a given graph $G$ is embeddable in an orientable surface
$S_{h}$ or not depends on its genus $\gamma(G)$, which is defined
to be the minimum genus of any orientable surface in which $G$ is
embeddable, i.e. $G$ is embeddable in $S_{h}$ if $h\geq\gamma(G)$.
In fact, any graph can be embedded in a surface with enough handles
just by adding a handle at each crossing, but we are rather interested
in cellular embeddings%
\footnote{In fact, if $\gamma(G)=h$, then every embedding of $G$ on $S_{h}$
is cellular, \cite{beineke1997graph}.%
} for which Euler's formula hold,
\begin{equation}
v-e+f=2-2h\label{eq:Eulers-formula-for-graph-embeddings}
\end{equation}
where $v,e,f$ are the number of vertices, edges and faces respectively
and $h$ is the (orientable) genus of $S_{h}$. In this context, a
planar graph has genus $\gamma(G_{planar})=0$.

There is no general formula for calculating the orientable genus of
a given graph, however, for the $(s,r)$-bipartite graph it is given
by
\begin{equation}
\gamma(K_{s,r})=\left\lceil \frac{(r-2)(s-2)}{4}\right\rceil \label{eq:genus-of-graph}
\end{equation}
where $\left\lceil x\right\rceil $ denotes the next integer bigger
than $x$. Hence, $K_{3,3}$ is embeddable in the torus but not in
the sphere since $\gamma(K_{3,3})=\left\lceil \frac{1}{4}\right\rceil =1$,
\cite{beineke1997graph}. 

There are similar relations for the non-orientable case, however,
our discussion will be only for embeddings in orientable surfaces
since embeddings of the graphs in non-orientable surfaces involve
a ``twist'' which is not clear how to deal with in the context of
spin networks.

Notice that from the relation \prettyref{eq:Eulers-formula-for-graph-embeddings}
there is a topological constraint to the allowed cellular embeddings
for a given graph. For instance, $K_{3,3}$ has 6 vertices and 9 edges,
thus, we obtain a constraint for the number of faces, $f=5-2h$, since
$f,h>0$. Furthermore, from \prettyref{eq:genus-of-graph} we have
$h\geq1$, hence, $f=3$ or $f=1$ in the case where $K_{3,3}$ is
embedded%
\footnote{From now on, whenever we refer to an embedding, it is meant a cellular
embedding.%
} in $T^{2}\text{ or }T^{2}\#T^{2}$ respectively.

Is there other information encoded in the graph that can help us to
further narrow down the possible embeddings? Does the orientation
of the vertices impose a constraint on the embedding? Now, having
found the number of possible faces (or 2-cells) for the embeddings,
we want to construct oriented surfaces such that the cellular embeddings
are automatically realized. In order to achieve this, we need to find
the circuits of the given graph $G$ and attach 2-cells to the regions
bounded by them. 

A \textbf{circuit} is a closed walk%
\footnote{A walk in $G$ is an alternating sequence $v_{1}e_{1}v_{2}e_{2}\dots v_{n-1}e_{n-1}v_{n}$
of vertices and edges of $G$, where every edge $e_{i}$ is incident
with $v_{i}\text{ and }v_{i+1}$, and $v_{i}\neq v_{i+1}$. If $v_{1}=v_{n}$,
then it is a closed walk, \cite{hartsfield2003pearls}. It is important
to notice here that this definition is a property of the graph itself.
We will, however, abuse the use of the language and refer to the regions
bounded by the circuits (cf. \prettyref{thm:Rotation-Scheme-Theorem})
as \textit{embedded circuits}, when they have repeated edges in both
directions, or as \textit{embedded cycles} when no edges are repeated.
To clarify, the difference is that the latter concepts are related
to the embedding of the graph in a surface and the definition given
above is a property of the graph related only indirectly to the embedding
of the graph through \prettyref{thm:Rotation-Scheme-Theorem}.%
} in $G$ such that no edge is repeated in the same direction. One
may imagine a circuit as walking along an edge in a certain direction
and, when getting to a vertex, the direction to follow (either clockwise
or counter-clockwise) is given by the orientation, also called rotation,
of the vertex under consideration. Thus, a given configuration of
the orientations of all vertices in the graph, called a rotation scheme,
induces a set of circuits giving rise to a specific cellular embedding.
Hence, all embeddings can also be described by giving the orientation%
\footnote{The orientation is given as a cyclic permutation of the neighbors
encountered while going clockwise around the vertices. In this convention,
the orientation of a vertex $v$ is the equivalent class of even permutations,
in the case of positive orientation denoted $(+1)$, or odd permutations,
in the case of negative orientation denoted $(-1)$, of its neighbors.
Such a definition takes into account the fact that both directions,
anticlockwise and clockwise, can be described by listing the neighbors
in order of their appearance when going \textit{clockwise}, for instance,
if $(abc)$ describes the positive orientation of $\xygraph{!{0;/r1.5pc/:}[u]!{\xcapv@(0)>{a}}*{\bullet}[lr(0.1)]!{\sbendv@(0)>{b}}[ll]!{\sbendh-@(0)>{c}}}$,
then $(acb)$ describes the negative orientation, meaning going around
counter-clockwise, which can also be regarded as permutating the edges
$vc$ and $vb$ giving $\xygraph{!{0;/r1.5pc/:}[u]!{\xcapv@(0)>{a}}*{\bullet}[lr(0.1)]!{\sbendv@(0)>{c}}[ll]!{\sbendh-@(0)>{b}}}$.
Notice that \textit{in the diagram} the listing of neighbors of $v$
is still clockwise.%
} of each vertex and specifying the regions bounded by the circuits
obtained from applying the so called \textbf{rotation rule}: after
the edge $xy$, take the edge $yz$, where $z$ is the successor to
$x$ in the permutation of the neighbors of the vertex $y$, see examples
below. The previous discussion is formalized in the next theorem,
\cite{beineke1997graph}.
\begin{thm}
\textbf{\label{thm:Rotation-Scheme-Theorem}Rotation Scheme Theorem.
}Let $G$ be a connected graph with v vertices and e edges, and let
$\Pi=\{\pi_{1},\pi_{2},\dots,\pi_{v}\}$ be a set of cyclic permutations
of the neighbors of the vertices $\{1,2,\dots,v\}$, i.e. a set of
a given orientation of all vertices. Let $W_{1},W_{2},\dots,W_{f}$
be the circuits obtained by applying the rotation rule to $\Pi$.
Then the circuits are the boundaries of the regions of a cellular
embedding of $G$ in $S_{h}$, with $h=(2-v+e-f)/2$, the genus of
the orientable surface $S_{h}$. Hence, all possible embeddings of
a graph are provided by the rotation schemes.
\end{thm}
Let us consider the graph $K_{3,3}$ with $v=6,\, e=9$. This is a
graph with trivalent vertices, hence, there are two possible rotations%
\footnote{In fact, if a vertex $v$ has degree $d$, then there are $(d-1)!$
different rotations of $v$.%
} for each of the six vertices. As a consequence, there are $2^{6}=64$
different sets $\Pi_{i=1,\dots,64}$. The question that arises immediately
is whether all these sets induce topological inequivalent embeddings
or not. In other words, if we disregard the labeling of the vertices,
how many different embeddings of $K_{3,3}$ in the torus or the double-torus
exist?

Before making a general claim, let us work out some examples to illustrate
the construction of embeddings by the rotation rule in order to understand
the relation between the set of orientations and the embeddings induced
by them. We will achieve this by listing the vertices and their neighbors
and from this list extract the circuits in the embedding using the
rotation rule, \cite{beineke1997graph,hartsfield2003pearls}. 

Notice that for a general $(r,s)$-bipartite graph the defining characteristics
are: \textit{(i)} The $r$ vertices corresponding to a set, say $R$,
are adjacent to each of the $s$ vertices corresponding to another
set, say $S$, and \textit{(ii)} $R\cap S=\emptyset$. In the case
of $K_{3,3}$ each of the odd vertices, $R=\{1,3,5\}$, are adjacent
to each of the even vertices, $S=\{2,4,6\}$. Thus, we denote (up
to cyclic permutations) positive orientations of even and odd vertices
as $(135)\text{ and }(246)$ respectively, and $(153)\text{ and }(264)$
for negative orientations. For instance, a standard way of picturing
$K_{3,3}$ with minimal crossing is

\medskip{}

\begin{center}
\def\objectstyle{\scriptscriptstyle}
\xy/r2pc/:
{\xypolygon4"L"{~:{(1,0):(0,3.5)::}~>{}{\bullet}}},
"L1";"L4"**{}?<>(.5)="L5"*{\bullet},"L2";"L3"**{}?<>(.5)="L6"*{\bullet},
"L2";"L1"**@{-},"L1";"L6"**@{-},"L6";"L5"**@{-},"L5";"L3"**@{-},"L3";"L4"**@{-},
"L4";"L2" **\crv{(1,-4) & (-4,-4.5)&(-1,2.5)},
"L4";"L6" **\crv{(.8,-3.5) & (-2,-2.5) & (-1,-.5)},
"L2";"L5" **@{-},
"L1";"L3" **\crv{(4,0)},
"L2"+(0,.2)*{1},"L1"+(0,.2)*{2},"L6"+(-.1,.2)*{3},"L5"+(.1,.2)*{4},"L3"+(0,.2)*{5},"L4"+(0,.2)*{6},
"L2"-(.2,.25)*{(246)},"L1"+(.4,0)*{(153)},"L6"+(0,-.3)*{(246)},"L5"+(.25,-.2)*{(153)},"L3"+(-.1,-.2)*{(264)},"L4"+(.4,-.1)*{(135)}
\endxy
\par\end{center}

\medskip{}

This configuration has orientations $(246)$ for the vertices $1$
and $3$, $(264)$ for vertex $2$ and $(153)$ for the vertices $2$
and $4$, $(135)$ for vertex $6$. 
\begin{example}
\label{exa:case-3-3}Consider the case where the vertices in each
set have the same orientation, i.e. either $(+1,+1,+1)$ or $(-1,-1,-1)$.
For instance, the vertices $\{1,3,5\}$ have negative orientation
while the vertices $\{2,4,6\}$ have all positive orientation:\\

\begin{tabular}{|c|c|c|c|c|c|c|}
\hline 
Vertex & 1 & 2 & 3 & 4 & 5 & 6\tabularnewline
\hline 
\hline 
Neighbors/Orientation & (264) & (135) & (264) & (135) & (264) & (135)\tabularnewline
\hline 
\end{tabular}\\

From this information we extract the circuits which will help us to
construct the embedding corresponding to this configuration. For instance,
take the edge $(12)$ and apply the rotation rule on it, i.e. the
neighbor of $2$ coming after $1$ in the cyclic permutation $(135)$
is $3$, hence, the next edge in the walk is $(23)$. Apply again
the rule to get $(36)$ and so on. After some steps, depending on
how long the walk is, one gets to the edge where the procedure started,
meaning that one has to stop and apply the same procedure to another
edge different than the ones encountered in the previous walk. In
this specific configuration, this algorithm results in the following
disjoint circuits,\end{example}
\begin{enumerate}
\item $(12)\rightarrow(23)\rightarrow(36)\rightarrow(65)\rightarrow(54)\rightarrow(41)\rightarrow(12)$;
\item $(25)\rightarrow(56)\rightarrow(61)\rightarrow(14)\rightarrow(43)\rightarrow(32)\rightarrow(25)$;
\item $(21)\rightarrow(16)\rightarrow(63)\rightarrow(34)\rightarrow(45)\rightarrow(52)\rightarrow(21)$. 
\end{enumerate}
Notice that there is a difference between the ``side'' $(xy)$ and
$(yx)$ reflecting the direction of the walk, hence, there are 18
``sides'' available to build the circuits. In this case, there are
three regions with six sides as boundaries, thus, the embedding has
three faces and corresponds to an embedding in the torus as depicted
below:

\medskip{}

\begin{center}
\def\objectstyle{\scriptscriptstyle}
\xy/r2pc/:
{\xypolygon4"T"{~:{(3,0):(0,.8)::}}},{\xypolygon6"A"{~:{(1,0):(0,1)::}{\bullet}}},
"T1";"T2"**{}?<>(.5)="T5"*{\circ},"T4";"T3"**{}?<>(.5)="T6"*{\circ},
"T1";"T4"**{}?<>(.5)="T7"*{*},"T2";"T3"**{}?<>(.5)="T8"*{*},
"A1";"T7"**@{-},"A4";"T8"**@{-},"A2";"T5"**@{-},"A5";"T6"**@{-},
"T7";"T4"**{}?<>(.5)="T9"*{=},"T8";"T3"**{}?<>(.5)="T10"*{=},
"T5";"T2"**{}?<>(.5)="T11"*{\times},"T6";"T3"**{}?<>(.5)="T12"*{\times},
"T9";"A6"**@{-},"T10";"T12"**\crv{"T12"+(0,1)},"T11";"A3"**\crv{"T11"-(.3,.75)},
"A1"+(.2,.2)*{1},"A2"+(.2,.2)*{2},"A3"+(0,.2)*{3},"A4"-(.2,.2)*{6},"A5"-(.2,.2)*{5},"A6"-(0,.2)*{4}
\endxy
\par\end{center}

\medskip{}

\begin{example}
\label{exa:case-1-1}Consider the case where one vertex has the opposite
orientation relative to the two other vertices of the same set. For
instance, the case where vertex $1$ and $2$ have positive orientation,
$(246)$ and $(135)$ respectively, and the rest have negative orientation:\\

\begin{tabular}{|c|c|c|c|c|c|c|}
\hline 
Vertex & 1 & 2 & 3 & 4 & 5 & 6\tabularnewline
\hline 
\hline 
Neighbors/Orientation & (246) & (135) & (264) & (153) & (264) & (153)\tabularnewline
\hline 
\end{tabular}\\

Using the rotation rule we obtain the following circuits,\end{example}
\begin{enumerate}
\item $(12)\rightarrow(23)\rightarrow(36)\rightarrow(61)\rightarrow(12)$;
\item $(21)\rightarrow(14)\rightarrow(45)\rightarrow(52)\rightarrow(21)$;
\item $(32)\rightarrow(25)\rightarrow(56)\rightarrow(63)\rightarrow(34)\rightarrow(41)\rightarrow(16)\rightarrow(65)\rightarrow(54)\rightarrow(43)\rightarrow(32).$
\end{enumerate}
Notice that this time, we obtain two circuits of length four and a
single one of length ten. Hence, the 18 sides available form three
faces and the embedding is in a torus:

\medskip{}

\begin{center}
\def\objectstyle{\scriptscriptstyle}
\xy/r2pc/:
{\xypolygon4"T"{~:{(3,0):(0,.8)::}}},{\xypolygon6"A"{~:{(1,0):(0,1)::}{\bullet}}},
"T1";"T2"**{}?<>(.5)="T5"*{\circ},"T4";"T3"**{}?<>(.5)="T6"*{\circ},
"T1";"T4"**{}?<>(.5)="T7","T2";"T3"**{}?<>(.5)="T8",
"A1";"A4"**@{-},"A3";"T5"**@{-},"A6";"T6"**@{-},
"T7";"T4"**{}?<>(.5)="T9"*{\times},"T8";"T3"**{}?<>(.5)="T10"*{\times},
"T10";"A5"**@{-},"T9";"A2"**\crv{"A6"+(.5,0)&"A2"+(1,0)},
"A1"-(.15,.15)*{1},"A2"+(0,.2)*{4},"A3"+(0,.2)*{5},"A4"-(.2,.15)*{2},"A5"-(.15,.2)*{3},"A6"-(0,.2)*{6}
\endxy
\par\end{center}

\medskip{}

\begin{example}
\label{exa:case-3-1}Finally, consider the case where the orientation
of all vertices in one set is the same while in the other set we have
one vertex with the opposite orientation relative to the other two
vertices. For instance, the case where all odd vertices have orientation
$(246)$ and vertex $2$ has positive orientation as well, while the
vertices $4$ and $6$ have orientation $(153)$:\\

\begin{tabular}{|c|c|c|c|c|c|c|}
\hline 
Vertex & 1 & 2 & 3 & 4 & 5 & 6\tabularnewline
\hline 
\hline 
Neighbors/Orientation & (246) & (135) & (246) & (153) & (246) & (153)\tabularnewline
\hline 
\end{tabular}\\

In this case we obtain, after using the described algorithm, only
one circuit with 18 sides, 
\begin{multline*}
(12)\rightarrow(23)\rightarrow(34)\rightarrow(41)\rightarrow(16)\rightarrow(65)\rightarrow(52)\rightarrow(21)\rightarrow(14)\rightarrow(45)\rightarrow(56)\rightarrow\dots\\
\dots\rightarrow(63)\rightarrow(32)\rightarrow(25)\rightarrow(54)\rightarrow(43)\rightarrow(36)\rightarrow(61)\rightarrow(12)
\end{multline*}
where all sides are walked exactly once%
\footnote{Notice that a circuit induced in this way may have repeated vertices
and edges used in both directions, however, if the edge is repeated
in the same direction the algorithm must stop, \cite{hartsfield2003pearls}.%
}. This configuration thus corresponds to an embedding in $T^{2}\#T^{2}$
which can be represented in a plane in a similar manner as the torus,
\[
\xy/r3pc/:{\xypolygon8"T"{~:{(3,0):(0,1)::}}},{\xypolygon6"A"{~:{(1.5,0):(0,1)::}~>{}{\bullet}}},"T1";"T2"**{}?<>(.5)="T9"*{\times}+(.3,-.1)*{\beta},"T2";"T3"**{}?<>(.5)="T10"*{\circ}+(.5,.2)*{\alpha},"T3";"T4"**{}?<>(.5)="T11"*{*}+(.1,.3)*{\delta},"T4";"T5"**{}?<>(.5)="T12"*{+}+(-.2,-.4)*{\gamma},"T5";"T6"**{}?<>(.5)="T13"*{*}+(.1,-.3)*{\delta},"T6";"T7"**{}?<>(.5)="T14"*{+}+(.3,-.2)*{\gamma},"T7";"T8"**{}?<>(.5)="T15"*{\times}+(.3,.1)*{\beta},"T8";"T1"**{}?<>(.5)="T16"*{\circ}+(.2,.5)*{\alpha},"A1"+(.2,.2)*{1},"A2"+(.2,.2)*{2},"A3"+(0,.2)*{3},"A4"-(.2,.2)*{4},"A5"-(.2,.2)*{6},"A6"+(.2,0)*{5},"T13";"A4"**@{-},"A4";"A3"**@{-},"A3";"A2"**@{-},"A2";"A1"**@{-},"A1";"A5"**@{-},"A5";"A6"**@{-},"A6";"T15"**@{-},"A2";"T9"**@{-},"A1";"T16"**@{-},"T10";"T11"**\crv{"T10"-(1,1)},"A6";"T14"**@{-},"A4";"T12"**@{-},"T11";"T4"**{}?<>(.5)="T17"*{=},"T12";"T5"**{}?<>(.5)="T18"*{/},"T5";"T13"**{}?<>(.5)="T19"*{=},"T6";"T14"**{}?<>(.5)="T20"*{/},"A3";"T17"**\crv{"T17"+(.3,-.3)},"T18";"T19"**\crv{"T18"+(.5,-.5)},"T20";"A5"**\crv{"A5"-(.2,0)},\endxy
\]
where the Greek letters denote the borders of the frame that have
to be glued together in order to obtain a double torus and the symbols
on them identify corresponding points.
\end{example}
Now, define the \textbf{value} $\mathfrak{v}$\textbf{ of a set of
vertices} as the modulus of the sum of orientations $\pm1$ of the
vertices in that set, e.g. the value of $\{1,3,5\}$ with orientation
$(-1-1-1)$ is $|-3|=3$. The definition is such that, if the orientations
of all vertices in a given set change, then the value remains invariant.
For instance, one can achieve a change of orientation of all vertices
in, say, the set $R=\{1,3,5\}$, e.g. $(+1+1-1)$, by an odd permutation
of the vertices in $S=\{2,4,6\}$ such that $(+1+1-1)\mapsto(-1-1+1)$
but $\mathfrak{v}_{++-}=|+1|=|-1|=\mathfrak{v}_{--+}$. In fact, by
permutations of the vertices in a set, one can construct all equivalent
diagrams%
\footnote{One has to consider the operation $R\rightleftarrows S$ as well.%
}, i.e. giving the same embedding, since these operations do not change
the 2-cells of the embedding, it merely results in a permutation of
the vertices in it.

Observe that the three cases in the examples above are the only cases
possible if we consider only the \textit{relative} orientation between
vertices of the same set, in which case the value is either $\mathfrak{v}=3$,
when all vertices have the same orientation, or $\mathfrak{v}=1$,
when one of the vertices has the opposite orientation relative to
the other two in the same set. The value of each set is independent
of each other, hence, we have the cases where the pair of values are
$(3,3)$, $(3,1)$, $(1,3)$ and $(1,1)$. However, since $K_{3,3}$
is symmetric under exchange of sets $R\leftrightarrow S$ preserving
the orientation, we can regard $(1,3)$ and $(3,1)$ as equivalent
cases.

Another way of looking at this is to consider the partition of 18
in summands with some constraints. Each of the summands represents
a circuit and their value represent the length of the circuit. The
defining characteristics of the $(3,3)$-bipartite graph do not allow
the construction of circuits with an odd number of edges since this
would mean that two vertices of the same set are adjacent. Thus, the
partition of 18 cannot contain any odd numbers. It follows that the
smallest possible circuit has length 4. Furthermore, the only number
of summands in the partition can be 1 or 3 since they represent the
regions of the embeddings. From these restrictions we conclude that
the only partitions of 18 allowed are $18,\;6+6+6,\;4+4+10$ and $4+8+6$.
However, the latter partition is not realized. To see this, notice
that it is not possible to construct a 4-circuit which is not a 4-cycle%
\footnote{A \textbf{cycle} is a circuit which does not have any repeated edges
in any direction.%
}, since this would mean that either one edge is repeated, in which
case the edge would need to have two loops at each vertex, or two
edges are repeated, this is not possible since all vertices in the
graph are 3-valent. Thus, the only possibility is to have a 4-cycle,
which implies automatically the existence of another 4-cycle. To see
this, notice that the 4-cycle has two vertices of each set, therefore,
there are two more vertices available to construct the graph, one
of each kind. The defining characteristics of $K_{3,3}$ impose the
constraint that these two vertices must be adjacent to each other
and to the corresponding vertices in the original 4-cycle. This leaves
no other possibility but to construct another 4-cycle, in contradiction
to the partition $4+8+6$.
\begin{rem*}
Observe the symmetry of the bipartite graph under permutation of its
vertices reflected in the pair of values of the sets as well as in
the partition of 18, thus, we may call the $(3,3)$ case ``maximal
symmetric'', the $(1,1)$ case ``minimal symmetric'' and the $(3,1)=(1,3)$
case ``asymmetric''. Therefore, we can think of this pair of values
as the ``degree of symmetry'' of the graph.
\end{rem*}
From the examples \prettyref{exa:case-3-3}, \prettyref{exa:case-1-1}
and \prettyref{exa:case-3-1} we see that the partitions $[6+6+6],\;[4+4+10]$
and $18$ correspond to the pair of values $(3,3),\:(1,1)$ and $(1,3)$
respectively. We say that the maximal (minimal) symmetric graph has
a $[6+6+6]$-type ($[4+4+10]$-type) embedding and the asymmetric
graph has a $18$-type embedding. Therefore we can say that the type
of embedding is only dependent on the ``degree of symmetry'' of
the graph given by the pair of values of the two sets $R$ and $S$.
In other words, the embeddings are topological invariant under permutations
acting on the sets of even and odd vertices. Odd permutations on one
set, merely change the orientation of all vertices in the other set,
in which case the value of the set is not affected. Even permutations
on one set only affect the cyclic order of the orientations in that
set, e.g. if we have an orientation $(-1+1+1)$ of the set $R$ an
even permutation acting on $R$ would only result in, say, $(-1+1+1)\mapsto(+1+1-1)$. 

Thus, from all 64 possible configurations of the orientations of vertices
in $K_{3,3}$ only three of them induce inequivalent embeddings. If
the pair of values is $(3,3)$, then there are 4 equivalent configurations
which induce a $[6+6+6]$-type embedding; either all 6 vertices have
positive (or negative) orientation or 3 vertices from one set have
positive (or negative) orientation while the vertices from the other
set have opposite orientation. Therefore, we are left with $64-4=60$
configurations; $36$ from them belong to the case where the pair
of value is $(1,1)$. This is a $[4+4+10]$-type embedding, hence,
one of the vertices on each set has the opposite orientation relative
to the vertices from the set which belongs to, i.e. the sets have
the orientations of the form $(+1-1-1)$ and cyclic, or of the form
$(-1+1+1)$ and cyclic. Therefore, for each relative orientation there
are 3 cases, which make $3\times2=6$ for each set. The rest $24$
of the configurations belong to the case where the pair of values
is $(3,1)=(1,3)$. This gives an embedding in the double torus. There
are 6 cases where the value of a set is 1 and 2 cases where the value
is 3, hence, for each of the pairs $(1,3)$ and $(3,1)$ there are
12 configurations to consider.

Finally, we can summarize the above discussion by the following claim,
\begin{claim}
If the value of the two disjoint sets of vertices in $K_{3,3}$ is
unequal, then the cellular embedding of $K_{3,3}$ corresponds to
an embedding in $T^{2}\#T^{2}$; otherwise the only cellular embeddings
of $K_{3,3}$ are in the torus $T^{2}$, such that it is a $[6+6+6]$-type
embedding for the $(3,3)$-value or a $[4+4+10]$-type embedding for
the $(1,1)$-value.
\end{claim}
This result is important since it would allow us to extract information
of the terms needed in the evaluation of non-planar spin networks
to account for their topology%
\footnote{We mean by the topology of a graph, the topology of the surface in
which the graph is embedded.%
}. The graph contains topological information about the surface in
which it is cellular embeddable and we can use recoupling theory to
extract that information. The reason for this is that we are considering
only cellular embeddings and we use all the information contained
in the graph (number of edges, vertices and their orientation) to
build the surfaces by the Rotation Scheme. Hence, the information
of the topology of the surface must be contained in the graph itself;
in other words, by reducing the graph in the embedding, we receive
a factor in the evaluation that reflects the information of the graph
being non-planar. That is why it is important to consider only cellular
embeddings, the faces are only 2-cells homeomorphic to discs with
no information about the global topology. For instance, in the case
of the tetrahedron we have two cellular embeddings in the torus, one
with an embedded 3- and another with an embedded 4-cycle as cells
and both wrapping the two circles of the torus. Both cellular embeddings
give, in fact, different evaluations, however only up to a constant
involving powers of $q$ (or $A$). In fact, these spin networks are
contained in $K_{3,3}$, in the sense that reducing the graph of $K_{3,3}$
in the torus via Moussouris' algorithm leads to such diagrams. We
will call the embedding of the tetrahedron in the torus with an embedded
3-cycle the \textit{toroidal Racah coefficient}.

There are of course (non-cellular) embeddings of a graph in surfaces
with higher genus, however, it is not the graph containing the information
about the topology of the surface but the surface itself. If we consider
a non-cellular embedding in the torus of the complete graph on 4 vertices,
i.e. the tetrahedron, and we \textquotedblleft{}cut\textquotedblright{}
the surface along the edge of the graph wrapping the circle of the
torus, we will get a surface which is not homeomorphic to a disc and
which contains the information about the topology of the torus. Thus,
the graph in that configuration has no information about a non-trivial
topology.

Due to the classification of closed oriented surfaces we expect that
the information extracted from the torus is sufficient to extend Moussouris
algorithm for the evaluation of planar spin networks to the non-planar
case. We believe that by knowing the evaluation of the spin network
corresponding to the torus we can use it to evaluate all spin networks
with higher genus in terms of products of this evaluation. To evaluate
these spin networks it would be necessary to arrange them such that
it is possible to \textquotedblleft{}cut\textquotedblright{} their
components (using the generalized Wigner-Eckart theorem, cf. \cite{moussouris1983quantum})
corresponding to each handle of the oriented surface and evaluate
each torus separately, this would give hopefully a sum of products
of toroidal symbols.

\section{The Evaluation of Non-planar Spin Networks\label{sec:non-planar-SN}}

In this section we discuss Moussouris' Decomposition Theorem, give
an improved version of it and present its algorithm for the evaluation
of planar spin networks, which relates these objects to the Ponzano-Regge
partition function. We then apply this algorithm to the graph $K_{3,3}$
in order to extract the information needed to extend the algorithm
to non-planar spin networks, i.e. we give the explicit form of the
toroidal phase factor for the $q$-deformed case, and discuss what
needs to be done to achieve such a generalization.

\subsection{The Decomposition Theorem\label{sub:The-Decomposition-Theorem}}

In \cite{moussouris1983quantum} J. P. Moussouris proved a theorem
which relates the spin networks with the Ponzano-Regge theory by reducing
a recoupling graph%
\footnote{Recall that a recoupling graph of a group $G$ is a labelled 3-valent
graph representing a contraction of tensors of $G$, cf. Sec. \prettyref{sub:Recoupling-Theory}.
In the following, the term ``recoupling graph'' will denote such
a graph together with an orientation of its vertices.%
} to a sum of products of Racah coefficients. This reduction, known
as the Decomposition Theorem, gives an evaluation of the spin network
only dependent on the labelling of the graph, as in \prettyref{cha:Invariants-of-3-Manifolds}
for a manifold with boundary. 

There are two versions of the mentioned theorem which we will present
and analyze in this section in order to understand how the expansion
of the algorithm for evaluating non-planar spin networks could be
done. The first version of the theorem, called network version, is
more general than the second version since it does not assume the
spin network to be planar, however, it assumes implicitly the existence
of an embedded cycle for the recoupling graph $F$ to be reduced.
This implies that the embedding of the graph in some surface has at
least two 2-cells since a cycle induces a region homeomorphic to a
disc by using only one side of each edge in the Rotation Scheme.

Moussouris does mention the importance of the orientation of the vertices
in the evaluation of the graph, pointing out that considering the
orientation of the vertices results in so-called phase factors, which
can be isolated as values of graphs with two vertices with the same
orientation, \cite{moussouris1983quantum}. However, in the proof
of the first version of the theorem this consideration enters only
in the first and second steps of the induction on the number of vertices
$V$ in $F$, disregarding the fact that the orientation of the vertices
of a recoupling graph affects the embedding of it in a surface, which
might be such that there is no embedded cycle at all. This would mean
that the spin network cannot be reduced straightforward. We will discuss
this case later. Moreover, in the proof it is also assumed implicitly
that after reducing all embedded cycles the only diagram left is either
a Racah coefficient or a toroidal phase factor (cf. Section \prettyref{sub:The-Toroidal-SN});
in the latter case we can call such a recoupling graph a toroidal
spin network since the phase factor left at the end of the reduction
contains the information of the graph being embedded in the torus%
\footnote{The Racah coefficient with a toroidal phase factor has as its cellular
embedding exactly the one discussed at the end of the previous section.%
}. 

As seen in example \prettyref{exa:case-3-1}, the appearance of an
embedded cycle is not always the case and there exist spin networks
which are irreducible if we only consider the operations described
in \cite{moussouris1983quantum}, thus, Moussouris' Decomposition
Theorem is limited to planar and toroidal spin networks. Hence, it
is necessary to rewrite the Decomposition Theorem in a more precise
manner in order to account for the case where the spin network is
toroidal, i.e. for the phase factors, which are especially important
for the $q$-deformed case.

We will now give both versions of the theorem and the proof of the
network version following \cite{moussouris1983quantum}. We modified
the first version of the theorem to account for the discussion above.
The second version is the special case where $F$ is planar and it
is proven by applying Schur's lemma and the Alexander moves to the
graph-theoretic dual version of $F$, which gives a triangulation
of the sphere. This dual version allows the connection with the Regge-Ponzano
theory.
\begin{thm}
\textbf{\textup{\label{thm:Decomposition-Theorem}Decomposition Theorem: }}

A recoupling graph $F$ of a compact semisimple group $G$, which
is at most toroidal, can always be evaluated as a sum of products
of Racah coefficients of $G$ and a toroidal phase factor.\end{thm}
\begin{proof}
The proof is by induction on the number of vertices $V$ in $F$ and
the size $l$ of the smallest cycle.

If $V=2$, the recoupling graph is a toroidal phase factor or a theta-evaluation
of a vertex. The case $V=3$ is not possible.

If $V=4$, the recoupling graph is a (toroidal) Racah coefficient
or two phases.

If $V>4$, we look for the smallest cycle in $F$ and reduce it as
follows, depending on the size $l$ of it. A 2-cycle is reduced using
Schur's identity \vref{eq:Edge-with-loop}. This results in a new
graph containing $V-2$ vertices. A 3-cycle is eliminated by producing
a single Racah coefficient by the Wigner-Eckart theorem. Alternatively,
one can regard the so-called ``crossing identity'' described below
to reduce the 3-cycle to a 2-cycle and apply Schur's identity. The
resulting graph contains $V-2$ vertices.

For the case $l>3$ we have a cycle with $l$ edges labelled by $j_{1},j_{2},\dots,j_{l}$.
This reduces to a $(l-1)$-cycle by the crossing identity derived
from using \vref{eq:Recoupling-Thm} on the edge, say, $j_{l}$.
This operation results in a Racah coefficient multiplied by a recoupling
graph in which the edge $j_{l}$ is removed while a new ``internal''
edge $x$ is introduced, coupling $j_{1}$ to $j_{l-1}$ and the other
two edges, which were coupling to $j_{l}$, are also coupled to $x$
and to each other. This resulting product is summed over the new edge
$x$ as in the Recoupling Theorem \vref{eq:Recoupling-Thm}. The cycle
is then reduced until $l=3$. 

This process of vertex reduction is repeated until $V=4$ giving as
a result a product of Racah coefficients and a phase factor summed
over all internal variables.
\end{proof}
For completeness we give the second version of the theorem above,
\begin{thm}
\textbf{\textup{Planar version of the Decomposition Theorem:}}

Let $F$ be a \textbf{planar }recoupling graph and $D(F)$ its dual
relative to a particular embedding in the sphere. Let $C(F)$ be a
combinatorial 3-manifold produced by dissecting $D(F)$ with internal
edges $x_{1},x_{2},\dots,x_{p}$ into tetrahedra $T_{1},\dots,T_{q}$.
Then, the evaluation of the recoupling graph is given by the amplitude
\[
\Psi(F)=\sum_{x_{1},\dots,x_{p}}\prod_{j=1}^{p}[x_{j}]\prod_{k=1}^{q}[T_{k}]
\]
where $[x_{j}]$ is the loop-value of the edge $x_{j}$ and the $[T_{k}]$'s
are the Racah coefficient associated with the tetrahedra $T_{k}$.
\end{thm}
The successive application of the Alexander moves, which correspond
in the dual form to the elimination of a 3-cycle and the crossing
identity, results in the introduction of sufficient internal edges
to dissect the interior of $D(F)$ into tetrahedra, giving a combinatorial
3-manifold $C(F)$ with $D(F)$ as its boundary. This decomposition
process is non-unique, however, the Biedenharn-Elliott identity and
the orthogonality of the $6j$-symbols ensures the equivalence of
the decompositions, \cite{moussouris1983quantum}. This is a special
case of the procedure to obtain the invariant%
\footnote{Notice that in the amplitude given above the theta-net factors are
missing. This is due to the fact that in \cite{moussouris1983quantum}
the spin networks are normalized such that the theta-nets are evaluated
to one. %
} described in \prettyref{cha:Invariants-of-3-Manifolds}.

\subsection{\label{sub:The-Toroidal-SN}The Evaluation of the Toroidal Spin Network
$K_{3,3}$ }

We will now apply the algorithm described in the above proof to the
$[4+4+10]$- and $[6+6+6]$-type embeddings of $K_{3,3}$ on a torus,
denoted by $K_{3,3}^{(1,1)}$ and $K_{3,3}^{(3,3)}$ respectively.
This will be done in order to extract information for the evaluation
of non-planar spin networks. 

In the case of the embedding $K_{3,3}^{(1,1)}$ we may start by applying
the crossing identity to the common edge of the 4-cycles and then
eliminating the two resulting 3-cycles by extracting two Racah coefficients.
The result is a sum over a single internal edge $x$ of a product
of three $6j$-symbols weighted by a factor of $(-1)^{2x}[2x+1]$.
These are, however, not all the factors since the diagram left encodes
the information of the graph being embedded in a torus. This diagram,
which we will call \textbf{toroidal phase factor}, can be represented
in a torus as follows:\smallskip{}

\xy/r1.5pc/:{\xypolygon4"T"{~:{(4,0):(0,.6)::}}},
"T0"-(.25,0)*{x},"T0"-(0,.5)="X1"*{\bullet},"T0"+(0,.5)="X2"*{\bullet},
"X1";"X2"**@{-},"T1";"T2"**{}?<>(.5)="M2";"X2"**@{-},"T4";"T3"**{}?<>(.5)="M1";"X1"**@{-},
"X1"+(.25,-.2)*{m},"X2"+(.25,.2)*{m},
"T1";"T4"**{}?<>(.5)="L2";"X2"**@{-},"T3";"T2"**{}?<>(.5)="L1";"X1"**@{-},
"L1";"X1"**{}?<>(.5)-(0,.25)*{l},"L2";"X2"**{}?<>(.5)+(0,.25)*{l}
\endxy

\smallskip{}

If we ``project'' this diagram to the plane by connecting the loose
ends of the edges $m$ and $l$, once we have disregarded the frame
of the above diagram, we get a theta-net with these edges crossing.
In order to get a more familiar theta-net, which can then be set to
have the value of 1, we need to ``twist'' the edge $x$. This is
done by following operation on a vertex%
\footnote{\label{fn:definition-of-twist-factor-under-cross}Notice that here
we are presenting the case of a vertex with an over-crossing, however,
this operation is defined for an undercrossing as well. In this case
we exchange $A\rightarrow A^{-1}$, cf. \prettyref{sec:Some-More-Diagrammatics}.%
}, \cite{carter1995classical,kauffman1994temperley},
\begin{equation}
\xygraph{!{0;/r1pc/:}[u(2)]!{\xcapv@(0)>{2j}}*{\bullet}[lr(0.1)]!{\sbendv-@(0)|{2b}}[ll]!{\sbendh-@(0)|{2a}}[ld]!{\xunderv[2]}}=(-1)^{a+b-j}A^{2[a(a+1)+b(b+1)-j(j+1)]}\xygraph{!{0;/r1.5pc/:}[u]!{\xcapv@(0)>{2j}}*{\bullet}[lr(0.1)]!{\sbendv-@(0)|{2a}}[ll]!{\sbendh-@(0)|{2b}}}.\label{eq:Twist-factor}
\end{equation}
The result of applying Moussouris algorithm and the above twisting
rule is a sum of products of Racah coefficients as in the planar case,
however, the non-planar nature of the graph is reflected in the ``twist
factor'' given above, i.e. in the evaluation of the toroidal phase
factor. Thus, we have 
\begin{equation}
\big[K_{3,3}^{(1,1)}\big]=\sum_{x}\Delta_{x}\left\{ \begin{array}{ccc}
j_{3} & j_{4} & k\\
j_{6} & j_{1} & x
\end{array}\right\} \left\{ \begin{array}{ccc}
m & j_{5} & j_{6}\\
j_{4} & x & l
\end{array}\right\} \left\{ \begin{array}{ccc}
j_{2} & l & j_{1}\\
x & j_{3} & m
\end{array}\right\} (-1)^{l+m-x}A^{2[l(l+1)+m(m+1)-x(x+1)]}\label{eq:quantum-9j-symbol}
\end{equation}
where $\Delta_{x}=(-1)^{2x}[2x+1]$ is the loop-evaluation \vref{eq:loop-value-in-TL-alg}
and $A$ is the square root of the deformation parameter $q$ introduced
before, cf. Sec. \prettyref{sub:The-Quantum-Group}. The squared brackets
denote the evaluation of a graph in terms of Racah coefficients.

The relation \prettyref{eq:quantum-9j-symbol} looks similar to the
$9j$-symbol. However, considering the cases where $A=\pm1$, we have
an overall factor of $(-1)^{l+m+x}$ which corresponds to one of the
Racah coefficients having a vertex with the ``wrong'' orientation.
This can be seen by expressing one of the $6j$-symbols where the
labels $x,l,m$ form an admissible triple in terms of $3j$-symbols
and permuting the order of the labels in the corresponding $3j$-symbol
by the following relation, \cite{edmonds1996angular},
\[
(-1)^{j_{1}+j_{2}+j_{3}}\left(\begin{array}{ccc}
j_{1} & j_{2} & j_{3}\\
m_{1} & m_{2} & m_{3}
\end{array}\right)=\left(\begin{array}{ccc}
j_{2} & j_{1} & j_{3}\\
m_{2} & m_{1} & m_{3}
\end{array}\right).
\]
The resulting factor is exactly the one described in \cite[p. 65]{moussouris1983quantum},
i.e. a tetrahedron with one of the vertices having an orientation
so that two edges cross. In fact, the diagram left after applying
the crossing identity and eliminating only one of the two 3-cycle
gives such a tetrahedron.

Consider now the evaluation $\big[K_{3,3}^{(3,3)}\big]$. It is possible
to reduce the $[6+6+6]$-type embedding to the $[4+4+10]$-type one
by applying the algorithm on two edges of the hexagonal figure shown
in example \prettyref{exa:case-3-3} which belong to the same ``exterior''
6-cycle, for instance the edges $(12)$ and $(54)$. From this procedure
we get
\begin{equation}
\big[K_{3,3}^{(3,3)}\big]=\sum_{v,w}\Delta_{v}\Delta_{w}\left\{ \begin{array}{ccc}
j_{4} & j_{5} & l\\
m & v & j_{6}
\end{array}\right\} \left\{ \begin{array}{ccc}
j_{6} & j_{1} & k\\
l & w & j_{2}
\end{array}\right\} \big[K_{3,3}^{(1,1)}\big](v,w)\label{eq:Reduction-(3,3)-to-(1,1)}
\end{equation}
where the first $6j$-symbol is the result of the crossing identity
on the edge $j_{5}=(54)$ and the second one is the result of the
same identity on the edge $j_{1}=(12)$, and the factor $\big[K_{3,3}^{(1,1)}\big](v,w)$
corresponds to the relation \prettyref{eq:quantum-9j-symbol} with
$j_{1}\rightarrow m,\, j_{2}\rightarrow l,\, j_{3}\leftrightarrow k,\, j_{5}\rightarrow j_{6},\, j_{6}\rightarrow j_{2},\, m\rightarrow w,\, l\rightarrow v$.

If we compare $\big[K_{3,3}^{(1,1)}\big]$ with the\textit{ }$9j$-symbol
we may recognize the possibility to use the following relation between
a $9j$-symbol (without twist factor) and $6j$-symbols, as in \cite{edmonds1996angular},
\begin{equation}
\sum_{\mu}(2\mu+1)\left\{ \begin{array}{ccc}
j_{11} & j_{12} & \mu\\
j_{21} & j_{22} & j_{23}\\
j_{31} & j_{32} & j_{33}
\end{array}\right\} \left\{ \begin{array}{ccc}
j_{11} & j_{12} & \mu\\
j_{23} & j_{33} & \lambda
\end{array}\right\} =(-1)^{2\lambda}\left\{ \begin{array}{ccc}
j_{21} & j_{22} & j_{23}\\
j_{12} & \lambda & j_{32}
\end{array}\right\} \left\{ \begin{array}{ccc}
j_{31} & j_{32} & j_{33}\\
\lambda & j_{11} & j_{21}
\end{array}\right\} \label{eq:9j-and-6j-relation}
\end{equation}
However, this relation does not account for the twist factor, hence,
it is not possible to use straightforward.
\begin{claim}
If we define the \textit{toroidal symbol}
\begin{multline*}
\left[\begin{array}{ccc}
j_{11} & j_{12} & j_{13}\\
j_{21} & j_{22} & j_{23}\\
j_{31} & j_{32} & j_{33}
\end{array}\right]_{A}:=\big[K_{3,3}^{(1,1)}\big]=\sum_{x}\Delta_{x}\left\{ \begin{array}{ccc}
j_{11} & j_{21} & j_{31}\\
j_{32} & j_{33} & x
\end{array}\right\} \left\{ \begin{array}{ccc}
j_{12} & j_{22} & j_{32}\\
j_{21} & x & j_{23}
\end{array}\right\} \left\{ \begin{array}{ccc}
j_{13} & j_{23} & j_{33}\\
x & j_{11} & j_{12}
\end{array}\right\} \\
\dots\times(-1)^{j_{21}+j_{32}-x}A^{2[j_{21}(j_{21}+1)+j_{32}(j_{32}+1)-x(x+1)]}.
\end{multline*}

Then the identity corresponding to \prettyref{eq:9j-and-6j-relation}
is
\begin{multline}
\sum_{\mu}[2\mu+1]\left[\begin{array}{ccc}
j_{11} & j_{12} & \mu\\
j_{21} & j_{22} & j_{23}\\
j_{31} & j_{32} & j_{33}
\end{array}\right]_{A}\left\{ \begin{array}{ccc}
j_{11} & j_{12} & \mu\\
j_{23} & j_{33} & \lambda
\end{array}\right\} =(-1)^{2\lambda}(-1)^{j_{21}+j_{32}-\lambda}A^{2[j_{21}(j_{21}+1)+j_{32}(j_{32}+1)-\lambda(\lambda+1)]}\\
\dots\times\left\{ \begin{array}{ccc}
j_{21} & j_{22} & j_{23}\\
j_{12} & \lambda & j_{32}
\end{array}\right\} \left\{ \begin{array}{ccc}
j_{31} & j_{32} & j_{33}\\
\lambda & j_{11} & j_{21}
\end{array}\right\} \label{eq:relation-to-reduce-(3,3)-to-(1,1)-1}
\end{multline}
where $[2\mu+1]$ corresponds to the quantum integer defined in section
\prettyref{sub:The-Quantum-Group} and equation \vref{eq:loop-value-in-TL-alg}.\end{claim}
\begin{proof}
The only term containing $\mu$ in the expansion of the l.h.s. in
term of $6j$-symbols is of the form
\[
\sum_{\mu}[2x+1][2\mu+1]\left\{ \begin{array}{ccc}
j_{11} & j_{12} & \mu\\
j_{23} & j_{33} & x
\end{array}\right\} \left\{ \begin{array}{ccc}
j_{11} & j_{12} & \mu\\
j_{23} & j_{33} & \lambda
\end{array}\right\} =\delta_{x,\lambda}.
\]
Thus, the only term left is the one on the r.h.s.\end{proof}
\begin{rem}
Notice that \prettyref{eq:relation-to-reduce-(3,3)-to-(1,1)-1} only
holds if the toroidal symbol has that exact form, i.e. $\mu$ must
be in any counter-diagonal position. Labels in that position appear
in the expansion \prettyref{eq:quantum-9j-symbol} only in one $6j$-symbol,
thus, the orthogonality of the $6j$-symbols may be used straightforward.
Moreover, it is possible to transpose the toroidal symbol since this
only changes the ordering of the admissible triples in the $6j$-symbols,
i.e. it is possible to use the symmetry properties of the $6j$-symbols
to achieve a transposition of the toroidal symbol. In the regular
case without twist factor, the $9j$-symbols have some symmetries
and this constraint does not appear. However, the symmetries of the
above defined symbol, if any besides the transposition, are not clear
at the moment.
\end{rem}
Even without having the symmetries needed it is possible to transform
$\big[K_{3,3}^{(1,1)}\big](v,w)$ in order to bring it in a form suitable
for the use of \prettyref{eq:relation-to-reduce-(3,3)-to-(1,1)-1}
to achieve a further simplification of $\big[K_{3,3}^{(3,3)}\big]$.
Consider the following relation similar to the one given in \cite{carter1995classical}%
\footnote{The relation is given in the reference in a different form, namely,
as a sum of products of three $6j$-symbols and a factor similar to
the twist factor described above. We used the proof for the case $A=\pm1$
given in \cite{carter1995classical} as a guide to reconstruct the
relation in order to present it as a ``symmetry'' of the toroidal
symbol. Notice that by using this relation six times, one obtains
the original form of the symbol, thus, this transformation can be
regarded as a symmetry.%
},
\[
\big[K_{3,3}^{(1,1)}\big](v,w)=\left[\begin{array}{ccc}
j_{4} & v & j_{6}\\
k & l & w\\
j_{3} & m & j_{2}
\end{array}\right]_{A}=(-1)^{l+j_{6}-k-j_{2}}A^{2[l(l+1)+j_{6}(j_{6}+1)-k(k+1)-j_{2}(j_{2}+1)]}\left[\begin{array}{ccc}
m & l & v\\
j_{3} & k & j_{4}\\
j_{2} & w & j_{6}
\end{array}\right]_{A}
\]
 Thus, using $1=(-1)^{-2(l+m+v)}$ for all $v$ and the relation 
\begin{multline*}
\sum_{v}[2v+1]\left\{ \begin{array}{ccc}
j_{4} & j_{5} & l\\
m & v & j_{6}
\end{array}\right\} \left[\begin{array}{ccc}
m & l & v\\
j_{3} & k & j_{4}\\
j_{2} & w & j_{6}
\end{array}\right]_{A}=(-1)^{j_{3}+w+j_{5}}A^{2[j_{3}(j_{3}+1)+w(w+1)-j_{5}(j_{5}+1)]}\times\dots\\
\dots\times\left\{ \begin{array}{ccc}
j_{3} & k & j_{4}\\
l & j_{5} & w
\end{array}\right\} \left\{ \begin{array}{ccc}
j_{2} & w & j_{6}\\
j_{5} & m & j_{3}
\end{array}\right\} 
\end{multline*}
we can simplify \prettyref{eq:Reduction-(3,3)-to-(1,1)} further.

Summarizing the discussion above we obtain the result that the $[6+6+6]$-type
embedding of the $(3,3)$-bipartite graph is a spin network with following
evaluation
\begin{multline}
\big[K_{3,3}^{(3,3)}\big]=\mathcal{A}(l,j_{6},k,j_{2};j_{3})\sum_{w}\Delta_{w}\left\{ \begin{array}{ccc}
k & j_{6} & j_{1}\\
j_{2} & l & w
\end{array}\right\} \left\{ \begin{array}{ccc}
j_{3} & m & j_{2}\\
j_{6} & w & j_{5}
\end{array}\right\} \left\{ \begin{array}{ccc}
j_{4} & j_{5} & l\\
w & k & j_{3}
\end{array}\right\} \times\dots\\
\dots\times(-1)^{w-j_{3}-j_{5}}A^{2[w(w+1)-j_{3}(j_{3}+1)-j_{5}(j_{5}+1)]},\label{eq:9j-symbol?}
\end{multline}
 where $\mathcal{A}(l,j_{6},k,j_{2};j_{3})=(-1)^{2j_{3}-(l+j_{6}+k+j_{2})}A^{4j_{3}(j_{3}+1)}A^{2[l(l+1)+j_{6}(j_{6}+1)-k(k+1)-j_{2}(j_{2}+1)]}$.

Notice that the above relation is not exactly the toroidal symbol
defined in \prettyref{eq:quantum-9j-symbol}, however, it looks very
similar and it could be argued that it is, in fact, a toroidal symbol
with an under-crossing instead of an over-crossing, cf. Footnote \prettyref{fn:definition-of-twist-factor-under-cross}.
It is yet unclear why the result is different and further work on
this would need to be done. The difference could be related to the
fact that there is a certain arbitrariness when it comes to project
the diagram embedded in the torus into the plane, hence, the need
to ``choose'' which (and even how) edges will cross. This corresponds
to choose the orientation of the surface in which the spin network
is embedded. However, the general form of the symbol remains and we
can observe that the topology of the surface in which the diagram
is embedded is refelected in the evaluation of this spin network.

Recall that we reduced the graph by choosing two common edges of the
same embedded 6-cycles. If we reduce the diagram by applying the crossing
identity to two edges of the central hexagonal region in example \prettyref{exa:case-3-3}
belonging to two different 6-cycles in the exterior of this hexagon,
e.g. $(14),(45)$, then the resulting reduction is given by a relation
of the following form, 
\[
\begin{array}{cc}
\big[\tilde{K}_{3,3}^{(3,3)}\big]= & \sum_{x,y,z}\Delta_{x}\Delta_{y}\Delta_{z}\left\{ \begin{array}{ccc}
j_{5} & j_{6} & m\\
k & x & j_{1}
\end{array}\right\} \left\{ \begin{array}{ccc}
j_{1} & z & j_{3}\\
m & j_{2} & l
\end{array}\right\} \left\{ \begin{array}{ccc}
j_{4} & j_{5} & l\\
x & y & j_{1}
\end{array}\right\} \left\{ \begin{array}{ccc}
y & j_{4} & j_{1}\\
j_{3} & z & k
\end{array}\right\} \times\dots\\
 & \dots\sum_{w}\Delta_{w}(-1)^{x+z-w}A^{2[x(x+1)+z(z+1)-w(w+1)]}\left\{ \begin{array}{ccc}
m & y & w\\
x & z & l
\end{array}\right\} \left\{ \begin{array}{ccc}
m & x & k\\
z & y & w
\end{array}\right\} 
\end{array}
\]
which can be simplified using, \cite{carter1995classical}, 
\begin{multline}
\sum_{w}\Delta_{w}(-1)^{x+z-w}A^{2[x(x+1)+z(z+1)-w(w+1)]}\left\{ \begin{array}{ccc}
m & y & w\\
x & z & l
\end{array}\right\} \left\{ \begin{array}{ccc}
m & x & k\\
z & y & w
\end{array}\right\} =\dots\\
\dots=(-1)^{k+l-y-m}A^{2[k(k+1)+l(l+1)-y(y+1)-m(m+1)]}\left\{ \begin{array}{ccc}
m & x & k\\
y & z & l
\end{array}\right\} \label{eq:Reduction-(3,3)-to-(3,3)}
\end{multline}
twice, first summing over $w$ to get an expression that can be reduced
further by the Biedenharn-Elliott identity on the internal edge $x$
and the second time summing over $y$ to obtain the following simplified
expression
\[
\big[\tilde{K}_{3,3}^{(3,3)}\big]=\tilde{\mathcal{A}}\sum_{z}\Delta_{z}\left\{ \begin{array}{ccc}
j_{4} & l & j_{5}\\
m & j_{6} & z
\end{array}\right\} \left\{ \begin{array}{ccc}
j_{3} & j_{2} & m\\
l & z & j_{1}
\end{array}\right\} \left\{ \begin{array}{ccc}
k & j_{1} & j_{6}\\
z & j_{4} & j_{3}
\end{array}\right\} (-1)^{j_{1}+j_{3}-z}A^{2[j_{1}(j_{1}+1)+j_{3}(j_{3}+1)-z(z+1)]}
\]
where $\tilde{\mathcal{A}}$ is a constant dependent on $A$ similar
to the one in \prettyref{eq:9j-symbol?}. Hence, we have
\[
\big[\tilde{K}_{3,3}^{(3,3)}\big]\propto\left[\begin{array}{ccc}
j_{4} & j_{3} & k\\
l & j_{2} & j_{1}\\
j_{5} & m & j_{6}
\end{array}\right]_{A}
\]

Even if $\big[K_{3,3}^{(1,1)}\big]$ and $\big[K_{3,3}^{(3,3)}\big]$
look similar, we were not able to conclude that they are exactly the
same; this is partly also due to the freedom in the choice of the
orientation of the surface, which might have affected the results.
We expect, however, that they are equal up to a sign and a factor
involving the parameter $A$, as in the case of the complete graph
on 4 vertices embedded in the torus. We might be able to solve this
ambiguity by defining the evaluation of the toroidal phase factor
as a sum over both possible crossings rather than just a single twist
factor, i.e. we would have a factor of $A^{+2[\dots]}+A^{-2[\dots]}$
instead. This is only a suggestion which will be verify in a paper
coming soon.

Finally, we consider the embedding of $K_{3,3}$ in the double torus.
As mentioned before, this embedding has only one 2-cell, thus, it
is not possible to reduce by means of Moussouris' algorithm. Hence,
in order to decompose it, it would be necessary to change the orientation
of a vertex by the ``twisting'' operation defined above. This would
give an overall twist factor dependent on the edges involved. This
operation is, however, highly arbitrary since, depending on the choice
of the vertex to be twisted, one obtains either of the embeddings
above or even the original embedding, thus, it is not a viable way
to proceed.

Nevertheless, since we were able to identify the toroidal phase factor
with the handle of the torus and obtained (up to orientation of the
surface) a symbol corresponding to this surface, one might ask if
all spin networks embeddable in an orientable closed surface with
genus $>0$ could be expressed as a sum of products of (quantum) $6j$-
and toroidal symbols, one for each handle. It is not hard to imagine
the existence of graphs with such evaluation. At this point the embedding
of the $(3,3)$-bipartite graph in the double torus is of great interest
since it might be the missing link needed to generalize the Decomposition
Theorem for spin networks with genus $>1$. We could use the inverse
operations of the ones used in Moussouris' algorithm on this embedding
in order to introduce enough vertices and edges such that the resulting
graph has two components, one on each handle, which are at least 3-edge
connected to each other. It would then be possible to separate the
components using the generalized Wigner-Eckart theorem, \cite{moussouris1983quantum}.
This could help us to study the possibility of an evaluation of non-planar
graphs as a sum of products of (quantum) $6j$- and toroidal symbols.
The results of these considerations are expected in the near future.

\section*{Conclusion}

In this dissertation we were able to explore a broad scope of different
topics involved in the description of combinatorial manifolds in terms
of spin networks. We explored briefly the possible significance of
these objects for a description of space in terms of abstract objects
derived from the properties of the category of spin representations
and their non-classical counterpart, the quantum group $U_{q}(\mathfrak{sl}_{2})$.
Each one of the fields presented here is of interest on its own, nevertheless,
together they give a description of spin networks at very different
levels. For instance, the diagrammatical language of these objects
encoding algebraical notions, relations and operations can be analyzed
in the setting of (topological) graph theory, as well as in the context
of combinatorial manifolds. This point of view helped us to identify
some key aspects of the structure of spin networks, such as non-planarity
and the information encoded in the graphs representing these objects. 

From this, we were able to analyze the Decomposition Theorem and Moussouris'
algorithm involved in its proof in order to improve its statement
which, as we noticed, was not clear enough. However, the generalization
of the Decomposition Theorem for networks of higher genus is not completed.
This is in part due to a generalization of Kuratowski's theorem, which
needs to be considered in a thorough manner. It states the existence
of a finite family of minimal forbidden subgraphs for each surface,
i.e. graphs which are not embeddable in the given surface. This theorem
might lower our expectations of being able to express a given graph
in terms of toroidal symbols since we expect other spin networks to
be non-toroidal. For instance, there are more than 800 minimal forbidden
graphs known for the torus, \cite{beineke1997graph}. Thus, we need
to analyze the general case in order to determine if the evaluation
of all spin networks with genus $>1$ is expressable as a sum of products
of quantum $6j$- and toroidal symbols. We expect, however, no constraint
regarding the form of the evaluation of these type of networks since
we could use the generalized Wigner-Eckart-theorem to divide the network
in its components representing the tori of the surface.

The fact that the only graph necessary to study was $K_{3,3}$ is
not surprising since it is the only 3-valent graph responsible for
non-planarity and every graph of higher valence is expandable to a
trivalent graph. However, what about trivalent graphs belonging to
other minimal forbidden families? Again, the information used for
extracting the factors needed in the evaluation is the only information
contained in the network as a graph theoretical object. Thus, we do
not expect a further complication due to these graphs. In any case
we can be sure that the evaluation can be expressed, at least, in
terms of $q$-$6j$-symbols and twist factors.

We saw explicitly the importance of the orientation of the vertices
for the evaluation of spin networks via Moussouris' algorithm. The
introduction of the twist factor in the evaluation of the spin networks
destroys or complicates many of the identities between $6j$-symbols
and, especially, $9j$-symbols which have, as we saw, the same form
as the toroidal symbols if we disregard the twist factor. Further
study of this symbol is necessary, for instance its symmetries or
whether there are other relations between them similar to the ones
for $9j$-symbols.

Notice that the objects needed for a possible description of quantum
gravity are not spin networks themselves but rather their 4-dimensional
Lorenzian analog, called spin foams, which can be seen as the time
evolution of the spin networks in the underlying space. The rich structure
of these objects arising from the different perspectives gives a strong
argument for their study in the context of quantum gravity. For instance,
an indirect related theory called Causal Dynamical Triangulation (CDT)
succeeded in constructing, as an infrared limit, the 4-dimensional
Minkowski spacetime, \cite{ambjorn2005reconstructing,ambjorn2005universe,ambjorn2006quantum}.
On the other hand, one may ask whether it is possible to use the framework
presented in this dissertation, or related ones, to study the description
of the Standard Model in terms of similar objects called braided ribbon
networks and the emergence of matter as topological excitations of
a given quantum geometry, \cite{bilson2007quantum,bilson2008particle,bilson2011emergent},
or even the emergence of locality and geometry itself, \cite{konopka2008quantum}.
Furthermore, there is evidence that spin foam models could be related
to gravity with a positive cosmological constant given by some relation
involving the root of unity $r$ of the parameter $q$ of the quantum
group, \cite{fairbairn2010quantum}. 

The above list is not, in any way, exhaustive and it is probably excluding
many other interesting aspects of topics related to spin networks,
however, it is merely intended as a suggestion of further reading
and as an example of the rich structure behind the concepts described
throughout the dissertation.

Finally, we give a consideration related to the topics above. There
is a similar theorem to Kuratowski's one for graphs embedded in three
dimensional spaces. It states that there are seven distinct graphs,
all containing $K_{3,3}$, which are not embeddable in 3 dimensions
without a link. However, the only networks considered in the structures
described, for instance, in CDT are trivial embeddings of graphs in
3-dimensional spaces; trivial in the sense that there are no links
to consider. It would be interesting to study whether this extra structure
gives new features useful to describe physical concepts such as, for
example, matter in a theory of quantum gravity. The author is conscious
of the highly speculative nature of this consideration, however, he
regards as important to raise humbly the issue (in a probably very
naive way) of the possible necessity to consider more basic notions,
than for instance geometry, to tackle the difficulties encounter in
quantum gravity.

\cleardoublepage{}

\bibliographystyle{amsplain}
\addcontentsline{toc}{chapter}{\bibname}\bibliography{Intro_SN_arXiv}

\providecommand{\bysame}{\leavevmode\hbox to3em{\hrulefill}\thinspace}
\providecommand{\MR}{\relax\ifhmode\unskip\space\fi MR }
\providecommand{\MRhref}[2]{%
  \href{http://www.ams.org/mathscinet-getitem?mr=#1}{#2}
}
\providecommand{\href}[2]{#2}
\begin{thebibliography}{10}

\bibitem{ambjorn2005reconstructing}
J.~Ambj{\o}rn, J.~Jurkiewicz, and R.~Loll, \emph{Reconstructing the universe},
  Physical Review D \textbf{72} (2005), no.~6, 064014.

\bibitem{ambjorn2005universe}
\bysame, \emph{The universe from scratch},  (2005).

\bibitem{ambjorn2006quantum}
J.~Ambjorn, J.~Jurkiewicz, and R.~Loll, \emph{Quantum gravity, or the art of
  building spacetime}, Arxiv preprint hep-th/0604212 (2006).

\bibitem{atiyah1988topological}
M.~Atiyah, \emph{{Topological quantum field theories}},  \textbf{68} (1988),
  no.~1, 175--186.

\bibitem{atiyah1990geometry}
M.F. Atiyah, \emph{{The geometry and physics of knots}}, Cambridge Univ Pr,
  1990.

\bibitem{barrett2009ponzano}
J.W. Barrett and I.~Naish-Guzman, \emph{The ponzano-regge model}, Classical and
  Quantum Gravity \textbf{26} (2009), 155014.

\bibitem{barrett1996invariants}
J.W. Barrett and B.W. Westbury, \emph{{Invariants of piecewise-linear
  3-manifolds}}, Transactions of the American Mathematical Society \textbf{348}
  (1996), no.~10, 3997--4022.

\bibitem{barrett1999spherical}
\bysame, \emph{{Spherical categories}}, Advances in Mathematics \textbf{143}
  (1999), no.~2, 357--375.

\bibitem{beineke1997graph}
L.W. Beineke and R.J. Wilson, \emph{Graph connections: Relationship between
  graph theory and other areas of mathematics}, Clarendon Press, 1997.

\bibitem{bilson2008particle}
S.~Bilson-Thompson, J.~Hackett, L.~Kauffman, and L.~Smolin, \emph{Particle
  identifications from symmetries of braided ribbon network invariants}, Arxiv
  preprint arXiv:0804.0037 (2008).

\bibitem{bilson2011emergent}
S.~Bilson-Thompson, J.~Hackett, L.~Kauffman, and Y.~Wan, \emph{Emergent braided
  matter of quantum geometry}, Arxiv preprint arXiv:1109.0080 (2011).

\bibitem{bilson2007quantum}
S.O. Bilson-Thompson, F.~Markopoulou, and L.~Smolin, \emph{Quantum gravity and
  the standard model}, Classical and Quantum Gravity \textbf{24} (2007), 3975.

\bibitem{carter1995classical}
J.S. Carter, D.E. Flath, and M.~Saito, \emph{The classical and quantum
  6j-symbols}, vol.~43, Princeton Univ Pr, 1995.

\bibitem{edmonds1996angular}
A.R. Edmonds, \emph{{Angular momentum in quantum mechanics}}, Princeton Univ
  Pr, 1996.

\bibitem{fairbairn2010quantum}
W.J. Fairbairn and C.~Meusburger, \emph{Quantum deformation of two
  four-dimensional spin foam models}, Arxiv preprint arXiv:1012.4784 (2010).

\bibitem{hartsfield2003pearls}
N.~Hartsfield and G.~Ringel, \emph{Pearls in graph theory: a comprehensive
  introduction}, Dover Pubns, 2003.

\bibitem{ito1993encyclopedic}
K.~Ito, \emph{{Encyclopedic dictionary of mathematics}}, The MIT Pr., 1993.

\bibitem{jones2005jones}
V.F.R. Jones, \emph{{The Jones Polynomial}}, preprint (2005).

\bibitem{kauffman2001knots}
L.H. Kauffman, \emph{{Knots and Physics: Series on Knots and Everything}},
  World Scientific Publishing Company, 2001.

\bibitem{kauffman1994temperley}
L.H. Kauffman and S.~Lins, \emph{{Temperley-Lieb recoupling theory and
  invariants of 3-manifolds}}, Princeton Univ Pr, 1994.

\bibitem{KellyLaplaza1980}
G.M. Kelly and M.I. Laplaza, \emph{Coherence for compact closed categories}, J.
  Pure Appl. Algebra 19 (1980), 193-213.

\bibitem{konopka2008quantum}
T.~Konopka, F.~Markopoulou, and S.~Severini, \emph{Quantum graphity: a model of
  emergent locality}, Physical Review D \textbf{77} (2008), no.~10, 104029.

\bibitem{lickorish1999simplicial}
WBR Lickorish, \emph{{Simplicial moves on complexes and manifolds}}, Geometry
  and Topology Monographs \textbf{2} (1999), no.~299-320, 314.

\bibitem{MacLane1971}
S.~Mac~Lane, \emph{{Categories for the Working Mathematician}}, 1971.

\bibitem{majid2000foundations}
S.~Majid, \emph{{Foundations of quantum group theory}}, Cambridge Univ Pr,
  2000.

\bibitem{major1999spin}
S.A. Major, \emph{{A spin network primer}}, American Journal of Physics
  \textbf{67} (1999), 972.

\bibitem{misner1973gravitation}
C.W. Misner, K.S. Thorne, and J.A. Wheeler, \emph{Gravitation}, WH Freeman \&
  co, 1973.

\bibitem{moussouris1983quantum}
J.P. Moussouris, \emph{{Quantum models of space-time based on recoupling
  theory}},  (1983).

\bibitem{Nakahara2003Geometry}
Mikio Nakahara, \emph{Geometry, topology and physics, second edition (graduate
  student series in physics)}, 2 ed., Taylor \& Francis, 2003.

\bibitem{penrose1971angular}
R.~Penrose, \emph{{Angular momentum: an approach to combinatorial space-time}},
  Quantum theory and beyond (1971), 151--180.

\bibitem{penrose1971applications}
\bysame, \emph{{Applications of negative dimensional tensors}}, Combinatorial
  mathematics and its applications (1971), 221--244.

\bibitem{piunikhin1992turaev}
S.~Piunikhin, \emph{{Turaev-Viro and Kauffman-Lins invariants for 3-manifolds
  coincide}}, J. Knot Theory Ramifications \textbf{1} (1992), 105--135.

\bibitem{ponzano1969semiclassical}
G.~Ponzano and T.~Regge, \emph{{SEMICLASSICAL LIMIT OF RACAH COEFFICIENTS.}},
  Tech. report, Princeton Univ., NJ, 1969.

\bibitem{regge1961general}
T.~Regge, \emph{{General relativity without coordinates}}, Il Nuovo Cimento
  (1955-1965) \textbf{19} (1961), no.~3, 558--571.

\bibitem{rovelli1995spin}
C.~Rovelli and L.~Smolin, \emph{Spin networks and quantum gravity}, Physical
  Review D \textbf{52} (1995), no.~10, 5743.

\bibitem{sawin1996links}
S.~Sawin, \emph{Links, quantum groups and tqfts}, BULLETIN-AMERICAN
  MATHEMATICAL SOCIETY \textbf{33} (1996), 413--446.

\bibitem{Smith2009}
J.~Smith, \emph{Introduction to abstract algebra}, CRC Press, 2009.

\bibitem{Soprunov}
Ivan Soprunov, \emph{Spherical geometry and euler's formula}, Spherical
  Geometry Class Notes.

\bibitem{turaev1992state}
V.G. Turaev and O.Y. Viro, \emph{{State sum invariants of 3-manifolds and
  quantum 6j-symbols}}, Topology \textbf{31} (1992), no.~4, 865--902.

\bibitem{viro1992moves}
O.~Viro, \emph{{Moves of Triangulations of a PL-Manifold}}, Quantum Groups
  (1992), 367--372.

\end{thebibliography}

\end{document}